\documentclass[lettersize,journal]{IEEEtran}
\usepackage{amsmath,amsfonts}
\usepackage{array}
\usepackage{textcomp}
\usepackage{stfloats}
\usepackage{url}
\usepackage{verbatim}
\usepackage{graphicx}
\usepackage[numbers]{natbib}

\usepackage{graphicx}
\usepackage{multirow}
\usepackage{booktabs}
\usepackage{subcaption}
\usepackage{threeparttable, array, float}
\usepackage{colortbl}

\usepackage{pgfplots}
\pgfplotsset{compat=1.17}

\usepackage{textcomp}

\usepackage{xcolor}

\usepackage{hyperref}
\definecolor{mydarkblue}{rgb}{0,0.08,0.45}
\hypersetup{ %
    pdftitle={},
    pdfsubject={},
    pdfkeywords={},
    pdfborder=0 0 0,
    pdfpagemode=UseNone,
    colorlinks=true,
    linkcolor=mydarkblue,
    citecolor=mydarkblue,
    filecolor=mydarkblue,
    urlcolor=mydarkblue,
}

\usepackage{amsmath,amsfonts, amssymb}
\usepackage{bm, bbm}
\usepackage{mathtools}
\usepackage{nicefrac}
\usepackage{physics}
\usepackage{nccmath}

\mathtoolsset{showonlyrefs,showmanualtags}

\makeatletter
\newcommand\addstarred[1]{%
    \expandafter\let\csname\string#1@nostar\endcsname#1%
    \edef#1{\noexpand\@ifstar\expandafter\noexpand\csname\string#1@star\endcsname\expandafter\noexpand\csname\string#1@nostar\endcsname}%
    \expandafter\newcommand\csname\string#1@star\endcsname%
}
\makeatother
\newcommand{\1}[1]{\mathbbm{1}{\{#1\}}}
\addstarred\1[1]{\mathbbm{1}{\left\{#1\right\}}}

\DeclareMathOperator*{\argmax}{arg\,max}
\DeclareMathOperator*{\argmin}{arg\,min}

\renewcommand{\ge}{\geqslant}
\renewcommand{\le}{\leqslant}

\newcommand{\type}[1]{Type-\uppercase\expandafter{\romannumeral#1}}

\DeclareMathOperator*{\Cov}{{\normalfont Cov}}

\usepackage{amsthm}
\newtheorem{theorem}{Theorem}
\newtheorem{lemma}[theorem]{Lemma}

\newtheorem{condition}{Condition}

\usepackage{algorithm}
\usepackage{algpseudocode}

\let\oldComment=\Comment
\renewcommand{\Comment}[1]{\oldComment{\texttt{#1}}}
\algnewcommand{\LeftComment}[1]{\Statex $\triangleright$ \texttt{#1}}
\algnewcommand{\RightComment}[1]{\Statex \leavevmode\hfill$\triangleright$ \texttt{#1}}

\algnewcommand\algorithmicinput{\textbf{Input:}}
\algnewcommand\Input{\item[\algorithmicinput]}%

\algnewcommand\algorithmicoutput{\textbf{Output:}}
\algnewcommand\Output{\item[\algorithmicoutput]}%

\algnewcommand\algorithmicinitial{\textbf{Initialize:}}
\algnewcommand\Initial{\item[\algorithmicinitial]}%

\usepackage{xspace}
\usepackage{comment}

\usepackage{fontawesome5}

\usepackage[colorinlistoftodos,prependcaption,textsize=tiny,textwidth=2.5cm,color=green]{todonotes}

\usepackage{cuted} 

\usetikzlibrary{positioning, shapes, arrows.meta}

\usetikzlibrary{decorations.pathmorphing}

\newcommand{\opal}{{\normalfont \textsc{Opal}}\xspace}

\begin{document}

\title{Optimal Online Probe Allocation for \\ Classical and Quantum Network Tomography}

\author{Xuchuang Wang, \IEEEmembership{Member, IEEE},  Yu-Zhen Janice Chen, Matheus Guedes
  de Andrade, \IEEEmembership{Member, IEEE}, Mohammad Hajiesmaili, \IEEEmembership{Member, IEEE}, John C.S. Lui, \IEEEmembership{Fellow, IEEE},
  Ting He, \IEEEmembership{Senior Member, IEEE},\\ Don Towsley, \IEEEmembership{Life Fellow, IEEE}%
  \thanks{
    Xuchuang Wang, Yu-Zhen Janice Chen, Matheus Guedes de Andrade, Mohammad Hajiesmaili, and Don Towsley are with the Manning College of  Computer Informations and Science, University of Massachusetts Amherst, Amherst, MA. (Email: \{xuchuangwang, yuzhenchen, mguedesdeand, hajiesmaili, towsley\}@cs.umass.edu)
  }\thanks{John C.S. Lui is with the Department of Computer Science and Engineering, The Chinese University of Hong Kong, Hong Kong. (Email: cslui@cse.cuhk.edu.hk)
  }\thanks{Ting He is with the the School of Electrical Engineering and Computer Science at Pennsylvania State University, University Park, PA. (Email: tinghe@psu.edu)} 
}

\markboth{
  Preprint, October~2025}%
{Wang \MakeLowercase{\textit{et al.}}: Optimal Online Probe Allocation for Classical and Quantum Network Tomography}


\maketitle

\begin{abstract}
  How to efficiently perform network tomography is a fundamental problem in network management and monitoring.
  A network tomography task usually consists of applying multiple probing experiments, e.g., across different paths or via different casts (e.g., unicast and multicast).
  We study how to optimize the network tomography process through online sequential decision-making.
  From the methodology perspective, we introduce an online probe allocation algorithm that sequentially performs network tomography based on the
  principles of
  optimal experimental design and the maximum likelihood estimation.
  We rigorously analyze the regret of the algorithm under the conditions that {\em i)} the optimal allocation is Lipschitz continuous in the parameters being estimated and {\em ii)} the parameter estimators satisfy a concentration property.
  From the application perspective, we present two case studies: {\em a)} the classical lossy packet-switched
  network and {\em b)} the quantum bit-flip network.
  We show that both cases fulfill the two theoretical conditions and provide their corresponding regrets when deploying our proposed online probe allocation algorithm.
  Besides case studies with theoretical guarantees, we also conduct simulations to compare our proposed algorithm
  with existing methods and demonstrate our algorithm's effectiveness in a broader range of scenarios.
  In an experiment on the Roofnet topology, our algorithm improves the estimation accuracy by \(13.64\%\) compared with the state-of-the-art baseline.
\end{abstract}


\section{Introduction}

\IEEEPARstart{N}{etwork} tomography~\citep{coates2002internet,he2021network,de2023characterization}
is essential for inferring the internal network (e.g., link) parameters, such as the loss rate, delay, and bandwidth, via end-to-end (external) measurements.
These external measurements are practically more accessible than the network's internal components (e.g., routers), which may be owned by different internet service providers (ISPs)~\citep{caceres1999multicast} and therefore difficult to access.
A stochastic network tomography problem consists of a network with
a set \(\mathcal{L}\) of \(L\coloneqq \abs{\mathcal{L}}\)
links, each link \(\ell \in\mathcal{L}\) with an unknown parameter \(\mu_\ell\) that characterizes its stochastic property (e.g., \(\mu_\ell\) may represent loss rate or average delay of the link), and a set \(\mathcal{M}\) of \(M\coloneqq\abs{\mathcal{M}}\) probes.
Performing a probe on the network refers to generating one or multiple stochastic measurements (i.e., \emph{observations}) that depend on the network parameters.
Network tomography aims to estimate the network parameters \(\mu_\ell\) from the observations.

Most prior works on network tomography focus on devising estimators for the network parameters from the stochastic observations, e.g., in packet delay tomography~\citep{duffield2000multicast,presti2002multicast,gu2010optimal}
and loss network tomography~\citep{caceres1999multicast,xi2006estimating} (detailed related works are in Section~\ref{sec:related-works}).
That is, after obtaining these stochastic observations from probes, prior works aim to devise good \emph{static} estimators to infer the network parameters.
While static estimator design is crucial for network tomography, another essential yet less explored task is, without knowing the network parameters \((\mu_\ell)_{\ell\in \mathcal{L}}\),
how to \emph{dynamically} collect {online} observations to \emph{efficiently} estimate the network parameters using the least number of probes. In short, this is a sequential decision-making problem (a.k.a., online learning)~\citep{barto1989learning} for network tomography.

However, to apply online learning techniques to network tomography, one needs to tackle unique challenges not present in common online learning scenarios.
A key technical challenge comes from the complex feedback mechanism in network tomography. In standard online learning settings, e.g., basic multi-armed bandits~\citep{auer2002finite} or linear bandits~\citep{li2010contextual}, the feedback is usually a scalar stochastic reward or loss, which is only related to the action (probe) taken by the algorithm. In contrast, in network tomography, feedback consists of the stochastic observations (e.g., a random vector) generated by the probes via the compound effect of the links across a set of paths traversed by the probe.
The observations often have a non-linear relation with the network link parameters, which may correlate with each other, making the combinatorial bandits framework~\citep{chen2013combinatorial} difficult to apply. This complex feedback mechanism makes the design of the online algorithm and the analysis of its regret challenging.




\subsection{Contributions}
This paper considers the network tomography task as an online sequential decision making problem and focuses on two new perspectives:
(1) A general framework that allows us to handle both classical and quantum network tomography tasks~\citep{de2022quantum,de2023characterization};
(2) fine-grained analysis, where we develop new algorithms to perform online network tomography with regret guarantees. Specifically, this paper makes the following contributions:

\begin{figure}[tb]
    \centering
    \begin{tikzpicture}
        \node[draw=blue!40!black, fill=blue!15, rounded corners=5pt,
            minimum width=6.45cm, minimum height=0.95cm, line width=1pt]
        (alg) at (0, 1.4) {};

        \node[above=of alg, yshift=-2cm, align=left]
        {\fontsize{10}{10}\selectfont
            \textcolor{blue!40!black}{\textbf{Online Probe Allocation (\opal)} (\S\ref{sec:algorithm})}
            \\[-2pt]
            \textcolor{blue!40!black}{Algorithm for Network Tomography
            }
        };

        \node[draw=green!30!black, fill=green!30, rounded corners=5pt,
            minimum width=6.45cm, minimum height=1.3cm, line width=1pt]
        (analysis) at (0, -0.1) {};

        \node[above=of analysis, yshift=-2.3cm, align=left]
        {\fontsize{10}{10}\selectfont
            \textcolor{green!30!black}{\textbf{Regret analysis of \opal}
                (\S\ref{sec:analysis})}
            \\[-2pt]
            \textcolor{green!30!black}{Condition~\ref{cond:lipschitz}: Lipschitz continuity}
            \\[-2pt]
            \textcolor{green!30!black}{Condition~\ref{cond:finite-confidence-interval}:
                Estimator concentration}};

        \node[draw=orange!40!black, fill=orange!30, rounded corners=5pt,
            minimum width=3.855cm, minimum height=0.84cm, line width=1pt]
        (cs1) at (-2, -1.5) {};

        \node[above=of cs1, yshift=-1.9cm, align=left]
        {\fontsize{10}{10}\selectfont
            \textcolor{orange!40!black}{\textbf{Case Study 1}
                (\S\ref{sec:case-study-classical}):}
            \\[-2pt]
            \textcolor{orange!40!black}{Classical Unicast}};

        \node[draw=orange!40!black, fill=orange!30, rounded corners=5pt,
            minimum width=3.855cm, minimum height=0.84cm, line width=1pt]
        (cs2) at (2, -1.5) {};

        \node[above=of cs2, yshift=-1.95cm, align=left]
        {\fontsize{10}{10}\selectfont
            \textcolor{orange!40!black}{\textbf{Case Study 2}
                (\S\ref{sec:case-study-quantum}):}
            \\[-2pt]
            \textcolor{orange!40!black}{Quantum Star Multicast}};

        \draw[->, thick, -Triangle] (alg)      -- (analysis);
        \draw[->, thick, -Triangle] (analysis) -- (cs1);
        \draw[->, thick, -Triangle] (analysis) -- (cs2);
    \end{tikzpicture}
    \caption{Paper structure and contribution summary
    }
    \label{fig:paper-structure}
\end{figure}

$\blacktriangleright$ First, we formulate a general network tomography task as an online experimental design problem (Section~\ref{sec:model}).
One attractive feature of this general model is that it
covers both classical and quantum tomography scenarios.
The problem of adaptively collecting probe observations and efficiently inferring the network parameters is formulated as a \emph{regret} minimization problem, where the objective is to minimize the difference between the performance of an online algorithm and the optimal policy in terms of optimal experimental design (OED) criteria~\citep{pukelsheim2006optimal}. To our knowledge, this paper is the first to design network tomography algorithms for the purpose of minimizing regret.


$\blacktriangleright$ Second, we devise the \underline{O}nline \underline{P}robe \underline{Al}location (\opal) algorithm for the general network tomography task (Section~\ref{sec:algorithm}).
The design of this algorithm is inspired by the idea of ``chasing optimal bounds,'' which is the backbone of a class of online learning algorithms with optimal performance~\citep{garivier2016optimal,combes2017minimal}.
We refer to the actual probe allocation ratio (the allocation fraction of each probe) of the algorithm as \emph{actual allocation}
and the optimal allocation ratio based on the estimated network parameters as \emph{estimated optimal allocation}.
\opal uses the difference between the actual allocation and the {estimated optimal allocation}
to direct the algorithm to sample the most inadequately performed probes (compared with the estimated allocation ratio) sequentially.

$\blacktriangleright$ Third, we propose a new analysis framework for studying the regret of \opal (Section~\ref{sec:analysis}). Analyzing the regret of \opal is challenging:
The chased estimated optimal allocation---determined by the estimated network parameters whose accuracy depends on the past probing decisions---varies over time.
Furthermore, the online algorithm's actual allocation depends on all its past probe decisions, which lags behind the time-varying estimated optimal allocation it chases.
We propose two critical conditions (Section~\ref{subsec:network-tomography-conditions}): \emph{Lipschitz continuity} (Condition~\ref{cond:lipschitz}) and \emph{confidence interval concentration} (Condition~\ref{cond:finite-confidence-interval}), that abstract the property of any network tomography task.
With these two theoretical conditions, we then rigorously prove that \opal achieves a \(\tilde{O}((\xi T)^{-\gamma_{\min}} + \xi \1{\exists m\in\mathcal{M}: \phi^*_m < \xi})\) regret (convergence rate), where \(T\)
is the total number of probes, \(\gamma_{\min} \in (0,1/2)\) is the smallest concentration rate of Condition~\ref{cond:finite-confidence-interval}, \(\xi\in (0,1)\) is an input constant of \opal, \(\1{}\) is an indicator function, \(\phi_m^*\) is the optimal allocation ratio of the \(m\)-th probe, and the \(\tilde{O}(\cdot)\) notation hides poly-logarithmic \(\log T\) terms.

$\blacktriangleright$ Last but not least, we present
two case studies on classical and quantum networks to demonstrate the effectiveness of our proposed framework and algorithms from both theoretical and empirical aspects (Sections~\ref{sec:case-study-classical} and~\ref{sec:case-study-quantum}).
For \emph{classical network tomography} (Section~\ref{sec:case-study-classical}), we first consider loss tomography in a classical star network
with unicast probes and verify that the two theoretical conditions in our analysis framework are indeed satisfied in this tomography task, providing a \(\tilde{O}(T^{-\nicefrac{1}{2}})\) regret of \opal for this classical tomography task (Section~\ref{subsec:classical-unicast-case}).
Then, we conduct extensive experiments to corroborate the empirical performance of \opal in classical network tomography for both star and general networks (Section~\ref{subsec:empirical-classical}).
For \emph{quantum network tomography} (QNT) (Section~\ref{sec:case-study-quantum}), we verify the two theoretical conditions for bit-flip channel tomography in a quantum star network---the state-of-the-art quantum network studied in QNT literature~\citep{de2022quantum,de2023characterization,de2024quantum}---with multicast probes and prove a \(\tilde{O}(T^{-\nicefrac{1}{3}})\) regret for this quantum tomography task (Section~\ref{subsec:quantum-multicast-case}).
Finally, we provide the empirical results of \opal for this quantum network tomography task (Section~\ref{subsec:empirical-quantum}) and show that \opal outperforms alternative baseline policies.
The main structure of this paper is illustrated in Figure~\ref{fig:paper-structure}.

\subsection{Related Works}\label{sec:related-works}

Network tomography has been studied for decades~\citep{caceres1999multicast,duffield2000multicast,presti2002multicast,coates2002internet,bu2002network,tsang2003network,xi2006estimating,gu2010optimal,gopalan2011identifying,ma2014node,he2015fisher,de2022quantum,de2023characterization}, and most prior works focus on network parameter estimation. For example, for delay network tomography~\citep{tsang2003network}, \citet{presti2002multicast,gu2010optimal} study estimator design for unicast probes,
and \citet{duffield2000multicast} study the multicast probes.
For loss network tomography, \citet{caceres1999multicast} devise the maximum likelihood estimators, and~\citet{xi2006estimating} propose EM-based estimations, both of which are for multicast probes.
Recently, network tomography in quantum networks is studied by~\citet{de2022quantum,de2023characterization}, where, due to the complexity of the quantum network and information, only multicast probes on a quantum star network have been considered.
The estimators and probes studied in these previous works form a solid foundation for our study from the online experimental design perspective. But, as our work focuses on the sequential probe selection and theoretical guarantees, the existing works are not directly applicable to our study.

Optimal experimental design (OED) aims to design experiments to achieve specific statistical criteria~\citep{pukelsheim2006optimal}. It has been applied in a wide range of applications, including clinical trial~\citep{haines2003bayesian}, synthetic biology~\citep{gilman2021statistical}, and chemistry~\citep{leardi2009experimental}, etc.
OED for network tomography has also been studied, ~\citet{he2015fisher,kveton2022optimal,gu2010optimal,xi2006estimating}, to name a few.
Among them, the most related work to this paper is~\citet{he2015fisher}, where they propose a heuristic iterative algorithm in terms of OED for the unicast network tomography task, and only an asymptotic convergence guarantee is provided.
In this paper, we make a concrete step forward by proposing a rigorous \textit{online} experimental design algorithm, \opal, for network tomography, and we are the \emph{first} to provide a fine-grained regret analysis for the network tomography task.
Outside the network tomography literature, online experimental design for linear regression is studied by~\citet{fontaine2021online,allen2021near}. However, the linear regression model is too simple to apply to network tomography, whose feedback can be non-linear and correlated.

\section{Modeling and Formulation
 }\label{sec:model}

In this section, we first formulate the general network tomography problem in Section~\ref{subsec:general-model}
and then present the classical loss and quantum bit-flip network tomography problems in Sections~\ref{subsec:loss-model} and~\ref{subsec:flip-model}, respectively, to further elaborate on the tomography tasks.
Preliminaries of the optimal experimental design are introduced in Section~\ref{subsec:objective}. Lastly, we cast the network tomography problem as an online experimental design problem in Section~\ref{subsec:online-learning}. We summarize the notations used in this paper in Table~\ref{tab:notation}.

\begin{table}[tb]
    \centering
    \caption{Summary of Notations}
    \label{tab:notation}
    \resizebox{\columnwidth}{!}{
        \begin{tabular}{ll}
            \toprule
            \textbf{Symbol}
                                                         & \textbf{Description}
            \\
            \midrule
            $\mathcal{G} = (\mathcal{N}, \mathcal{L})$
                                                         & Network graph: nodes $\mathcal{N}$, links $\mathcal{L}$
            \\
            $N$
                                                         & Number of nodes ($|\mathcal{N}|$)
            \\
            $L$
                                                         & Number of links ($|\mathcal{L}|$)
            \\
            $M$
                                                         & Number of probes
            \\
            $\mathcal{M}$
                                                         & Set of probing experiments (probes), $|\mathcal{M}| = M$
            \\
            $X_\ell$                                     & Random variable for link $\ell$ (e.g., loss, bit-flip)       \\
            $\bm Y_m$                                    & Observation vector from probe $m$                            \\
            $\mathcal{N}_m$                              & Destination nodes for probe $m$                              \\
            $f_m(\bm Y_m; \bm\mu)$                       & Probability density/mass function for probe $m$              \\
            $T$                                          & Total number of decision rounds                              \\
            $m_t$                                        & Probe selected at round $t$                                  \\
            $\mathcal{H}_t$                              & History up to round $t$: $\{(m_s, \bm Y_{m_s, s})\}_{s=1}^t$ \\
            $S_{m,t}$                                    & Number of times probe $m$ has been performed up to $t$       \\
            $\bm S_t$                                    & Vector of all $S_{m,t}$ at time $t$                          \\
            $\bm\phi_t$                                  & Actual allocation at $t$: $(S_{m,t}/t)_{m\in\mathcal{M}}$    \\
            $\bm\phi$                                    & Generic allocation vector (probability simplex)              \\
            $\bm\phi^*$                                  & Optimal allocation for given $\bm\mu$ and OED criterion      \\
            $\bm\mu = (\mu_\ell)_{\ell \in \mathcal{L}}$ & Vector of unknown link parameters                            \\
            $\hat{\bm\mu}_t$                             & Estimated link parameters at time $t$                        \\
            $\hat{\bm\phi}_t^*$                          & Estimated optimal allocation at time $t$                     \\
            $\texttt{MLE}(\mathcal{H})$                  & Maximum likelihood estimator given history $\mathcal{H}$     \\
            $\bm I_m(\bm\mu)$                            & Fisher information matrix for probe $m$                      \\
            $\bm I(\bm\mu; \bm\phi)$                     & Fisher information matrix under allocation $\bm\phi$         \\
            $F(\bm\mu; \bm\phi)$                         & OED criterion (e.g., A-optimal: $\tr(\bm I^{-1})$)           \\
            $R_T$                                        & Regret: $F(\bm\mu; \bm\phi_T) - F(\bm\mu; \bm\phi^*)$        \\
            $\gamma_{\ell,m}$                            & Concentration rate exponent for link $\ell$, probe $m$       \\
            $c_{\ell,m}$                                 & Constant in confidence interval for link $\ell$, probe $m$   \\
            $\xi$                                        & Initial allocation lower bound parameter                     \\
            $\Delta^{M-1}$                               & $(M-1)$-dimensional probability simplex                      \\
            \bottomrule
        \end{tabular}
    }
\end{table}


\subsection{General Network Tomography
}\label{subsec:general-model}

Let \(\mathcal{G}=(\mathcal{N},\mathcal{L})\) denote the network topology with a set of \(N{\in\mathbb{N}^+}\)
nodes \(\mathcal{N}\coloneqq \{1,2,\dots,N\}\) and a set of \(L\in\mathbb{N}^+\) links \(\mathcal{L}\coloneqq \{1,2,\dots, L\}\).
Associated with each link \(\ell\in\mathcal{L}\) is a stochastic model \(X_\ell\) representing the characteristic of the link, e.g., a Bernoulli random variable for signal loss~\citep{caceres1999multicast}, or a Gaussian random variable for packet delay variation (PDV)~\citep{tsang2003network},
or the noise level or fidelity for quantum channel~\citep{nielsen2010quantum}.
We refer to a connection between two nodes (source and destination) across several consecutive links as a \emph{path} in the network.
The tomography task assumes that the link model is known, but that its parameters, denoted as a vector \(\bm\mu \coloneqq (\mu_\ell)_{\ell \in \mathcal{L}}\) (e.g., the link success rates),
are unknown.

The network tomography task aims to estimate these unknown link parameters via a given set of \(M\in\mathbb{N}^+\) probing experiments \(\mathcal{M}\coloneqq \{1,2,\dots, M\}\) (called \emph{probes} later), examples of which are unicast~\citep{he2015fisher} and multicast~\citep{caceres1999multicast}.
Each probe \(m\in\mathcal{M}\) is associated with one path or a tree rooted at a source node (sender) whose leaves are the set of destination nodes (receivers) denoted as \(\mathcal{N}_{m}\subseteq \mathcal{N}\) (see Figure~\ref{fig:star_network} for an example).   The feedback \(\bm Y_m \in \mathbb{R}^{\abs{\mathcal{N}_m}}\) from the probing experiment \(m\) consists of received signals at the destination nodes in \(\mathcal{N}_m\).
Denote the probability of observing feedback \(\bm Y_m\) given network parameters \(\bm\mu\) and probing experiment \(m\) as \(f_m(\bm Y_m; \bm\mu)\), where \(f_m(\cdot; \bm\mu)\) is the probability density function.
The goal is to estimate the network parameters based on the feedback from the probing experiments.
For ease of presentation, we follow the common identifiability assumption in network tomography literature~\citep{bu2002network,he2015fisher} that the network parameters \(\bm \mu\) can be uniquely determined from feedback from
an appropriate set of
probing experiments in \(\mathcal{M}\).
Given a set of \(T\in\mathbb{N}^+\) probe-feedback pairs \(\mathcal{H}\coloneqq\{(m_t, \bm Y_{m_t, t})\}_{t=1}^T\), we define the log-likelihood function as \(L(\bm \mu;\mathcal{H}) \coloneqq \sum_{t=1}^T  \log  f_{m_t}(\bm Y_{m_t, t}; \bm \mu).\)
The maximum likelihood estimation (MLE) estimator for the network parameters \(\bm\mu\) is defined as \(\texttt{MLE}(\mathcal{H}) \coloneqq \argmax_{\bm\mu} L(\bm \mu;\mathcal{H}).\)\footnote{Throughout this paper, we assume the outputs of the \(\argmin\) and \(\argmax\) functions are unique; otherwise, one can break the tie arbitrarily.}

Next, we introduce classical loss network tomography (Section~\ref{subsec:loss-model}) and quantum bit-flip network tomography (Section~\ref{subsec:flip-model}) as examples to illustrate network tomography.

\subsection{Classical Loss Network Tomography}\label{subsec:loss-model}


In a loss network, associated with each link \(\ell\in\mathcal{L}\) is a Bernoulli random variable \(X_\ell\sim \mathcal{B}(\mu_\ell)\) representing loss on that link, where \(\mu_\ell\in (0,1)\) is the probability that a transmitted signal is \emph{not} lost, i.e., \(\mathbb{P}(X_\ell = 1)=\mu_\ell\).
We assume these Bernoulli random variables are independent across links and time.
As an example, a star loss network with \(N=4\) nodes and \(L=3\) links is depicted in Figure~\ref{fig:star_network}.
We  theoretically and empirically investigate this model in Section~\ref{sec:case-study-classical}.

\begin{figure}[tb]
    \centering
    \begin{tikzpicture}
        \node[circle,fill=black] at (360:0mm) (center) {} node[below,yshift=-1mm] {$4$};
        \foreach \n in {1,...,3}{
                \node[circle,fill=black,minimum size=2mm] at ({\n*360/3 - 30}:2cm) (n\n) {};
                \draw[line width=1.5pt] (center)--(n\n) node[midway,fill=white] {$X_{\n}$} node[below left] {$\n$};
            }
        \draw[<->,bend right=30,dashed,line width=1pt,color=blue] (n2) to (n1) node[midway,below] {};


        \draw[->,bend left=15,dotted,line width=1.5pt,color=red] (n1) to (n2) node[midway,below] {};
        \draw[->,bend right=15,dotted,line width=1.5pt,color=red] (n1) to (n3) node[midway,below] {};

    \end{tikzpicture}
    \caption{A star network with \(4\) nodes: Denote the central node as \(4\), and the peripheral nodes are denoted as \(1,2,3.\) The solid lines represent the loss characteristics of the network (modeled by Bernoulli random variables \(X_1, X_2, X_3\)), the dashed blue line represents a unicast probe (probe at node \(1\), receive at node \(2\), or vice versa), and the two dotted red lines represent a multicast probe (probe at node \(1\), duplicate at intermediate node \(4\), receive at nodes \(2\) and \(3\)).}
    \label{fig:star_network}
\end{figure}

Unicast~\citep{he2015fisher} and multicast~\citep{caceres1999multicast} are two common probes in the loss networks. To elaborate on their details, let us consider the star network in Figure~\ref{fig:star_network} and focus on end-to-end tomography where one only has access to the peripheral nodes \(1,2,3\), for either sending or receiving signals from or to any other peripheral nodes.
A unicast probe refers to sending signals across a path from a source node to a single destination node, and a multicast probe refers to sending
a signal from a source node down a tree to multiple destination nodes, where the signal is duplicated at bifurcation nodes within the tree.
For example, in Figure~\ref{fig:star_network},
the dashed blue arrow between nodes \(1\) and \(2\) represents a unicast probe, sending signals from node \(1\), across central node \(4\), and then received at node \(2\) or vice versa from node \(2\) to node \(1\), denoted as \((1\leftrightarrow 2)\) since a unicast probe is \emph{symmetric}.
The probability of receiving a signal at node \(2\) when sent from node \(1\) is \(f_{m=(1\leftrightarrow 2)}(\bm Y_m = 1; \bm \mu) = \mathbb{P}(X_1 X_2 = 1) = \mu_1\mu_2.\)
In this \(3\)-link star network, there are in total three possible unicast probes: \((1\leftrightarrow 2), (1\leftrightarrow 3), (2\leftrightarrow 3).\)
The two dotted red arrows in Figure~\ref{fig:star_network} together represent a multicast probe.
Here node \(1\) is the source node and send signals to nodes \(2\) and \(3.\)
Then, the probability of receiving signals at both nodes \(2\) and \(3\) is \(f_{m=(1\to (2, 3))}(\bm Y_m = (1,1); \bm \mu) = \mathbb{P}(X_1X_2=1, X_1X_3 = 1) = \mu_1\mu_2\mu_3,\) and the probability of receiving signals at node \(2\) but not at node \(3\) is \(f_{m=(1\to (2, 3))}(\bm Y_m = (1,0); \bm \mu) = \mathbb{P}(X_1X_2=1, X_1X_3 = 0) = \mu_1\mu_2(1-\mu_3),\) and so on.
Using the feedback from the unicast and multicast probes, one can infer the loss characteristics of all links---parameters \(\mu_1, \mu_2, \mu_3\) of the star network in Figure~\ref{fig:star_network}.

\subsection{Quantum Bit-flip Network Tomography}\label{subsec:flip-model}\label{subsec:quantum-tomography-model}

\textbf{Background of General quantum network tomography.}
Quantum networks are communication systems used to interconnect quantum devices such as quantum processors. Despite differences between the physical properties of classical and quantum information systems, network tomography extends naturally to the quantum setting. Similar to the classical case of learning link parameters, quantum network tomography (QNT) addresses the characterization of quantum channels representing links in a quantum network through end-to-end measurements~\cite{de2024quantum}.
QNT is defined under the assumption that the quantum channel representing link $\ell \in \mathcal{L}$, has the parametric form
$\mathcal{E}(\rho) = \sum_{k} M_{\ell k}^{\dagger}(\mu_{\ell}) \rho M_{\ell k}(\mu_{\ell})$, where $\mu_{\ell} \in \mathbb{R}^{d_{\ell}}$ for $d_{\ell} \in \mathbb{N}^+$ are the parameters,
$\rho$ is an input state,
$\{M_{\ell k}\}$ denote the Kraus operators for $\mathcal{E}_{\ell}$.
The goal of QNT is to estimate $\mu_{\ell}$ for each link $\ell \in \mathcal{L}$.

QNT probes consist of
(a) a quantum algorithm for the distribution of quantum states in a network and
(b) quantum measurement operators performed in the end nodes for each distributed state. The algorithm for state distribution simultaneously specifies which links and intermediate node operations are used for distribution. Like the classical case, the output of quantum measurements performed in the end nodes yields a random vector $\bm Y_m$ whose statistics depend on network parameters $\bm\mu$.

\begin{figure}[tb]
    \centering
    \includegraphics[width=0.75\linewidth]{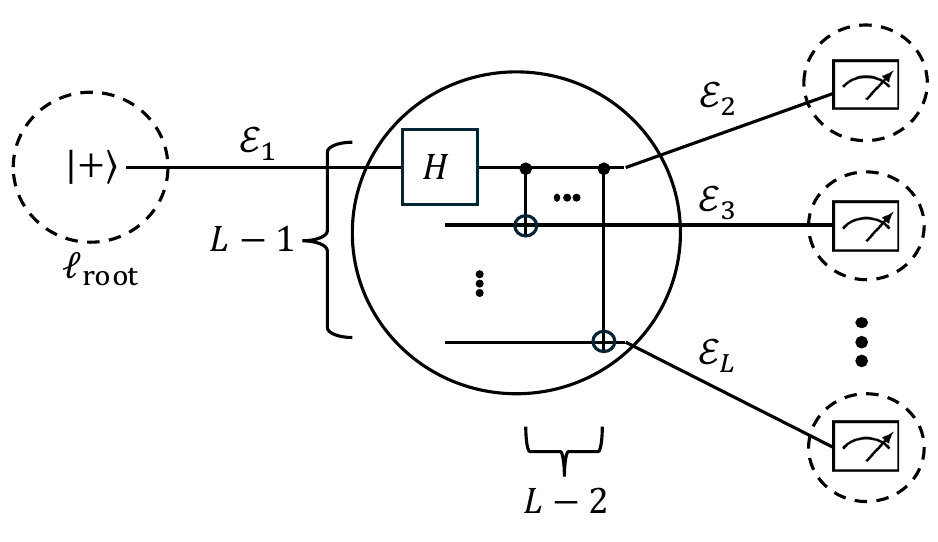}
    \caption{Root independent multicast for quantum star network tomography. The root $\ell_{\text{root}} = 1$ prepares state $\ket{+}$, transmits it to the central node, and an $(L-1)$ qubit state is prepared from the application of CNOT gates. Each output is transmitted to an end node via the links representing the quantum channels $\mathcal{E}_{\ell}$, for $\ell = 2,\ldots, L$. $\sigma_{\text{Z}}$ basis measurements performed in the end nodes generate boolean random variables $\bm Y_{\ell_{\text{root}}} \in \{0,1\}^{L-1}$.}
    \label{fig:root-independent}
\end{figure}

\textbf{Quantum bit-flip network tomography.}
In this work, we apply our online experimental design algorithm to optimize the quantum network tomography of bit-flip channels in a quantum network with a star topology~\cite{de2022quantum, de2023characterization}. Quantum network links are assumed to be \emph{single-qubit, bit-flip} channels of the form
$\mathcal{E}(\rho) = \mu_{\ell} \rho + (1 - \mu_{\ell}) \sigma_{\text{X}} \rho \sigma_{\text{X}}$, where $\mu_{\ell} \in (0, 1)$ denotes the probability of no bit-flip
for link $\ell \in \mathcal{L}$, and $\sigma_{\text{X}}$ is the Pauli X operator~\cite{nielsen2010quantum}. We utilize probes defined in~\cite{de2023characterization} that rely on multipartite quantum state distribution, focusing on the \textit{Root Independent} distribution strategy depicted in Fig~\ref{fig:root-independent}, which we now describe.
The number of nodes \(N\) in the star network equals the number of links \(L\) plus one, i.e., \(N=L+1\).
The process starts with the selection of a root node (link) $\ell_{\text{root}}$, which sends a qubit in state $\ket{+}$ to the intermediate node.\footnote{We use Dirac notation for quantum states sparingly in this paper, such as, \(\ket{+}\) and \(\ket{0}\), for rigor. However, understanding these notations is not necessary for following the content. We refer interested readers to~\citet{nielsen2010quantum}.} Upon receiving the qubit from $\ell_{\text{root}}$, the intermediate node applies a Hadamard gate~\cite{nielsen2010quantum} in the qubit and performs an $(L - 1)$-qubit CNOT gate on $(L - 2)$ ancilla qubits initialized in state $\ket{0}$ controlled by the received qubit. Each output of the CNOT gate is sent to the end nodes of the star, excluding $\ell_{\text{root}}$, such that an $(L - 1)$-qubit quantum state is prepared in the rest \((L-1)\) end nodes. The process terminates with $\sigma_{\text{Z}}$-basis measurements performed on the qubits received by the end nodes.
The probability of observing a \(\ket{b_2\,b_3\dots\,b_L}\) quantum state (for \(b_\ell \in \{0,1\}\)) at the \(L-1\) end nodes after sending \(\ket{+}\) at node \(\ell_{\text{root}} = 1\) is \(f_{m=1}(Y_m = \ket{b_2\,b_3\dots\,b_L}; \bm\mu) = \prod_{\ell=2}^L \mu_{\ell}^{b_\ell} (1 - \mu_{\ell})^{1 - b_{\ell}}\).

Later, we provide a detailed case study of applying our new algorithm to this quantum bit-flip network tomography model in Section~\ref{sec:case-study-quantum}.
The reason for choosing the simple quantum bit-flip star network is mainly because the current research on QNT is still quite limited, and the QNT protocols for the bit-flip star are already state-of-the-art in the literature~\citep{de2022quantum,de2023characterization,de2024quantum}.
We emphasize that this paper focuses on devising an online experimental design algorithm to optimize network tomography procedure. Once new QNT protocols are proposed for more practical quantum network infrastructures, one can apply our algorithm to optimize the probe allocation over these protocols.

\subsection{Preliminaries of Optimal Experimental Design} \label{subsec:objective}

\textbf{Fisher Information Matrix and Cram\'er-Rao Bound.}
The Fisher information matrix (FIM) of probing the network via probe \(m\in\mathcal{M}\) regarding link parameters \(\bm\mu\) is as follows, \begin{align}
     & \bm I_m(\bm\mu) \coloneqq
    \\
     & \,\, \left\{ \mathbb{E}_{\bm Y_m}\left[\left. \frac{\partial \log f_m(\bm Y_m; \bm\mu) }{\partial \mu_\ell} \frac{\partial \log f_m(\bm Y_m; \bm\mu) }{\partial \mu_{\ell'}} \right\rvert \bm\mu\right] \right\}_{(\ell, \ell') \in \mathcal{L}\times \mathcal{L}}.
\end{align}
Let \(\bm\phi \in\Delta^{M-1}\) denote the allocation ratio (probability distribution) of the number of times of each of the \(M\) probes \(S_m\) over the total number of times of probes,
where \(\Delta^{M-1}\) is the probability simplex in \(\mathbb{R}^M.\)
{Let \(\bm I(\bm\mu; \bm\phi)\) denote the expected FIM of a single probe chosen according to distribution $\bm\phi$  with respect to the network parameters \(\bm\mu.\)}
Given the linearity of FIM, we have \(
\bm I (\bm\mu; \bm\phi) = \sum_{m=1}^M \phi_m \bm I_m(\bm\mu).
\)
The Cram\'er-Rao bound (CRB) states that the covariance matrix of any unbiased estimator \(\hat{\bm\mu}\coloneqq (\hat \mu_\ell)_{\ell \in \mathcal{L}}\) of the network parameters \(\bm\mu\) (under allocation \(\bm \phi\)) is lower bounded by the inverse of the FIM, i.e., \(\Cov(\hat{\bm\mu}) \succcurlyeq \bm I^{-1}(\bm\mu; \bm\phi)\)~\citep{lehmann2006theory}, where \(\bm A\succcurlyeq \bm B\) denotes that \(\bm A \setminus \bm B\) is a positive semi-definite matrix.
In particular, when the MLE estimator \(\hat\mu_{\text{MLE}}\) is unbiased, it achieves the CRB, i.e., \(\Cov(\hat{\bm\mu}_{\text{MLE}}) = \bm I^{-1}(\bm\mu; \bm\phi)\).


Below, we introduce two typical optimal experimental design objectives, A-Optimal and D-Optimal,
and our model can handle other OED objectives.
\textit{A-Optimal} experimental design~\citep{pukelsheim2006optimal,lopez2023optimal}  minimizes the trace of the inverse of the FIM \(\bm I (\bm\mu; \bm\phi)\) with respect to the allocation \(\bm\phi\), i.e. \(\tr(\bm I^{-1}(\bm\mu;\bm\phi))\).
This quantity is a lower bound for the average of the variances of the estimated link parameters \(\sum_{\ell\in\mathcal{L}} \text{Var}(\hat\mu_\ell)\), when the estimator is unbiased.
\textit{D-Optimal} experimental design~\citep{lopez2023optimal} minimizes the reciprocal of the determinant of the FIM \(\bm I (\bm\mu; \bm\phi)\) with respect to the allocation \(\bm\phi\), which minimizes the volume of the confidence ellipsoid of the estimated parameters.


In this paper, we let \(F(\bm\mu;\bm\phi)\) denote a general optimal experimental design criterion and denote \(\bm\phi^*(\bm\mu) \coloneqq \argmin_{\bm\phi} F(\bm\mu;\bm\phi)\) as the optimal allocation function given parameter \(\bm\mu\) and OED criterion \(F\).\footnote{
    We slightly abuse the \(\bm\phi^*\) notation without the input \((\bm\mu)\) to denote the optimal allocation based on the actual parameters.
}
For example, in the A-optimal scenario, we have \(F(\bm\mu; \bm\phi) = \tr(\bm I^{-1}(\bm\mu;\bm\phi))\), and in the D-optimal scenario, we have \(F(\bm\mu; \bm\phi) = \left( \det(\bm I (\bm\mu; \bm\phi)) \right)^{-1}\).

\floatname{algorithm}{Procedure}
\begin{algorithm}[tp]
    \caption{Online Experimental Design for Network Tomography}
    \label{proc:online-learning}
    \begin{algorithmic}[1]
        \Input \(\mathcal{M}\) (set of probes), \(\mathcal{N}\) (set of nodes), \(\mathcal{L}\) (set of links), \(T\) (number of decision rounds), \(F\) (optimal experimental design criterion)
        \For{each time slot \(t=1,2,\dots, T\)}
        \State Select a probe \(m_t\) from set \(\mathcal{M}\) to perform the network tomography once
        \State Observe \(\bm Y_{m_t,t} = (Y_{n,t})_{n\in\mathcal{N}_{m_t}}\) from all destination nodes \(n\in\mathcal{N}_{m_t, t}\) of the probe \(m_t\)
        \State Update the parameters and estimates \(\hat{\bm\mu}_t\)
        \EndFor
        \Output \(\hat{\bm\phi}^*\) (estimated optimal allocation)
    \end{algorithmic}
\end{algorithm}
\floatname{algorithm}{Algorithm}

\subsection{Online Experimental Design Problem and Objective Formulation}\label{subsec:online-learning}
As tomography aims to reveal the values of the unknown network parameters \(\bm\mu\), we formulate the task as an online experimental design problem with a total of \(T\in\mathbb{N}^+\) decision rounds.
In each decision round \(t\le T\), the learner selects one tomography probe \(m_t\) and then observes the received signals \(\bm Y_{m_t,t}\) in the destination nodes \(n\in\mathcal{N}_{m_t}\) generated by this probe.
The learner then updates the estimates of the network parameters based on the received signals until the end of the \(T\) decision rounds.
Let \(S_{m,t}\) denote the number of times that a specific algorithm performs probe \(m\) up to time slot \(t\),
and let \(\bm\phi_t = (S_{m,t}/t)_{m\in\mathcal{M}}\) denote the actual allocation ratio of an algorithm at time step \(t\).
The online experimental design problem is summarized in Procedure~\ref{proc:online-learning}.

As the goal of optimal experimental design (OED) is to minimize function \(F(\bm\mu;\bm\phi)\) in terms of the allocation \(\bm\phi\), we take the difference between \(F(\bm\mu; \bm\phi_T)\) of the concerned algorithm and the minimum \(F(\bm\mu;\bm\phi^*)\) with optimal allocation \(\bm\phi^*\) as our objective, that is, \[
    R_T \coloneqq F(\bm\mu; \bm\phi_T) - F(\bm\mu; \bm\phi^*).
\]
Following the online learning literature, we refer to this {metric} as \emph{regret} since it describes the gap between the performed algorithm and the hindsight optimal.\footnote{
    Some literature may refer to this regret definition as \emph{simple regret}, the performance gap at the end of the online learning process,
    to differentiate from the cumulative regret.
}
It is important to note that it differs from cumulative regret, which is the sum over all \(T\) rounds.
The cumulative regret is unsuitable in our setting since \(\bm\phi_T\) already considers all probes up to the end time \(T\).
Our regret measures an algorithm's convergence rate: The smaller the regret, the better the algorithm performs.


\section{\opal: An Online Experimental Design Algorithm for Network Tomography}\label{sec:algorithm}

An online experimental design
algorithm aims is to estimate link parameters \(\bm\mu\) and approach the optimal allocation \(\bm\phi^*\) as close as possible.
However, without knowledge of link parameters
\(\bm\mu\) a priori, it is challenging for an online algorithm to achieve the exactly optimal allocation \(\bm\phi^*\) during a finite number of rounds, \(T\).
This difficulty is due to that the final allocation ratio \(\bm\phi_T\) is the accumulation of its whole probe sequence over all \(T\) decisions, and many of the early probe decisions, accounted for in the final ratio \(\bm\phi_T\), are made without accurate estimates of link parameters \(\bm\mu\). Therefore, the actual allocation $\bm\phi_T$ will inevitably deviate from $\bm{\phi}^*$.

We propose an online tomography algorithm that dynamically chases the {estimated} optimal probe allocation, called \underline{O}nline \underline{P}robe \underline{Al}location (\opal). The key idea of \opal is inspired by the ``chasing the (theoretical) optimal bound'' idea used in multi-armed bandits literature as exemplified by the optimal sampling algorithm that minimizes the regret of any structured bandits in~\citet{combes2017minimal} and the Track-and-Stop algorithm achieving the optimal sample complexity for best arm identification in~\citet{garivier2016optimal}.

While we use the high-level idea of ``chasing'' from prior literature, the design of \opal for the network tomography problem poses \emph{several unique challenges} that make it impossible to directly apply existing algorithmic designs to our problem, thus requiring a substantially different analysis than previous bandit algorithms~\citep{combes2017minimal,garivier2016optimal}. The main difference is that the feedback mechanism in network tomography is more complex: one scalar observation (e.g., obtained by a unicast probe) can be non-linearly related to multiple link parameters along a path, and a vector observation (e.g., obtained by a multicast probe) may depend on even more link parameters through a set of non-linear relations. Indeed, this non-linearity is far more involved than the simple feedback mechanism present in bandit learning problems, where one scalar feedback usually corresponds precisely to one unknown parameter or a linear combination of unknown parameters, and independence across observations is often assumed. The non-linearity and correlation in the feedback of network tomography make the problem more challenging, requiring a different algorithmic design and analysis.

We first present the basic version of \opal in Section~\ref{subsec:opach-algorithm} and then introduce a lazy update strategy in Section~\ref{subsec:lazy-update} to reduce the computational cost of the algorithm.


\subsection{\opal Algorithm Design}\label{subsec:opach-algorithm}

\begin{algorithm}[tp]
    \caption{Optimal Probe Allocation (\opal)}
    \label{alg:chase}
    \begin{algorithmic}[1]
        \Input \(\mathcal{M}\) (set of probes), \(\mathcal{N}\) (set of nodes), \(\mathcal{L}\) (set of links), \(T\) (number of decision rounds), \(F\) (optimal experimental design criterion), \(\texttt{MLE}\) (estimator),
        \(\bm S_0 \coloneqq (S_{m,0})_{m\in\mathcal{M}}\) (initial sample sizes)
        \Initial \(\mathcal{H}_0 \gets \emptyset\) (history), \(\hat{\bm\mu}_0 \gets \bm 0\) (parameter estimates), \(T_0 \gets \sum_{m\in\mathcal{M}} S_{m,0}\) (initial length)

        \For {each time step \(t = 1, 2, \dots, T\)} \label{line:for-loop}

        \If {\(\exists m\in\mathcal{M}: S_{m,t} < S_{m,0}\)}\label{line:initial-phase} \label{line:initial-probe-condition}
        \Comment{Initial phase}
        \State \(m_t\gets \text{randomly pick from }\{m{\in}\mathcal{M}: S_{m,t} {<} S_{m,0}\}\) \label{line:initial-probe}

        \Else
        \Comment{Chasing phase} \label{line:chase-phase}

        \State \(\hat{\bm\mu}_{t} \gets \texttt{MLE}(\mathcal{H}_{t-1})\) \Comment{Update the estimates \(\hat{\bm\mu}_t\)} \label{line:greedy-MLE}\label{line:initial-mle-estimate}

        \State \(\hat{\bm \phi}_t^* \gets \argmin_{\bm\phi} F(\hat{\bm\mu}_{t}; \bm\phi)\) \label{line:min-F}
        \RightComment{Estimated optimal allocation}
        \label{line:greedy-a-optimal}

        \State \(m_t \gets \argmax_{m\in\mathcal{M}} \hat\phi_{m,t}^* - (S_{m,t}/t)\) \RightComment{Chase the estimated allocation} \label{line:greedy-allocation}\label{line:chase}

        \EndIf



        \State Perform probe \(m_t\) and observe signals \(\bm Y_{m_t, t} = (Y_{n,t})_{n\in \mathcal{N}_{m_t}}\) for its destination nodes \label{line:greedy-observe}
        \State \(\mathcal{H}_{t} \gets \mathcal{H}_{t-1}\cup\left\{(m_t, \bm Y_{m_t, t})\right\}\), \(S_{m_t, t} \gets S_{m_t, t-1} + 1\) and \(S_{m,t}\gets S_{m,t-1}\) for all \(m\neq m_t\)\label{line:greedy-update}

        \EndFor \label{line:greedy-for-loop-end}

        \Output \(\hat{\bm\mu}_{T+1}\gets \texttt{MLE}(\mathcal{H}_{T})\)
        (final estimates) and \(\hat{\bm\phi}_{T+1} \gets \argmax_{\bm\phi} F(\hat{\bm\mu}_{T+1}; \bm\phi)\) (final allocation)
    \end{algorithmic}
\end{algorithm}

\opal (Algorithm~\ref{alg:chase}) consists of two phases: (a). an initial sampling phase (Line~\ref{line:initial-probe}) and (b). a chasing phase (Lines~\ref{line:greedy-MLE}-\ref{line:chase}).
The initial phase takes initial sample sizes \(\mathcal{S}_0 = (S_{m,0})_{m\in\mathcal{M}}\) as input and performs each probe \(m\in\mathcal{M}\) accordingly (Lines~\ref{line:initial-probe-condition}--\ref{line:initial-probe}).
Collecting these \(T_0\coloneqq \sum_{m\in\mathcal{M}} S_{m,0}\) initial samples ensures that the first estimates \(\hat{\bm \mu}_{T_0}\) for link parameters, Line~\ref{line:initial-mle-estimate} are not too far away from the true parameters \(\bm \mu\).
The input initial samples \(\bm S_0\) can either be
set universally as fixed values for guaranteed good performance (discussed in Section~\ref{subsec:rate-analysis})
or determined by the specific network tomography task and prior knowledge of the network case by case (details discussed in two case studies in Sections~\ref{subsec:classical-unicast-case} and~\ref{subsec:quantum-multicast-case}).


The chasing phase proceeds step by step for each of the remaining time steps \(t = T_0 + 1, T_0 + 2, \dots, T\).
Denote by \(S_{m,t}\) the number of times that probe \(m\) is performed up to and including time \(t\), and use vector \(\bm{S}_t \coloneqq (S_{m,t})_{m\in\mathcal{M}}\) to represent all the sample sizes.
For each time step \(t>T_0\), the algorithm first updates the MLE estimates \(\hat{\bm\mu}_t\) of the network parameters using the latest history \(\mathcal{H}_{t-1} = \{(m_s, \bm Y_{m_s, s})\}_{s=1}^{t-1}\) (Line~\ref{line:greedy-MLE}).
With the updated estimates \(\hat{\bm\mu}_t\), the algorithm
generates an \emph{estimated optimal allocation} \(\hat{\bm\phi}_t^*\) based on the latest link parameter estimates \(\hat{\bm\mu}_{t}\) (Line~\ref{line:greedy-a-optimal}).
Then, the algorithm subtracts the actual allocation \(\bm S_t/t\) from the estimated optimal allocation \(\hat{\bm\phi}_t^*\) element-wise (both are \(M\)-entry vectors), where the (possibly negative) entries of the output vector represent the inadequacies of the allocated fractions to each probe.
Last, the algorithm performs the probe \(m_t\) with the worst (highest) allocation inadequacy \(\hat{\phi}_{m,t}^* - (S_{m,t}/t)\) once (Line~\ref{line:greedy-allocation})---\emph{chasing the estimated optimal allocation} \(\hat{\bm\phi}_t^*\)---to collect a new observation, and updates the sample sizes \(\bm S_t\) (Lines~\ref{line:greedy-observe}--\ref{line:greedy-update}).
After both phases, the algorithm outputs the final MLE estimates \(\hat{\bm\mu}_{T+1}\) and final estimated optimal allocation \(\hat{\bm\phi}_{T+1}\).

\subsection{\opal with lazy updates}\label{subsec:lazy-update}

\begin{algorithm}[tp]
    \caption{Lazy Update for \opal (replace \opal's Lines~\ref{line:greedy-MLE} and~\ref{line:greedy-a-optimal} as follows, called \opal-lazy)}
    \label{alg:lazy-chasing}
    \begin{algorithmic}[1]
        \Input inputs of \opal and integer \(B\) (lazy update batch size)


        \If{\(t \equiv 1 \pmod B\)} \label{line:lazy-update-start}
        \Comment{Lazy update}

        \State \(\hat{\bm\mu}_{t} \gets \texttt{MLE}(\mathcal{H}_{t-1})\) \Comment{Update the mean estimates } \label{line:lazy-MLE}

        \State \(\hat{\bm \phi}_t^* \gets \argmin_{\bm\phi} F(\hat{\bm\mu}_{t-1}; \bm\phi)\)\RightComment{Estimated optimal allocation}
        \label{line:lazy-optimal}

        \Else
        \State \(\hat{\bm \phi}_t^*  \gets \hat{\bm \phi}_{t-1}^* \) \Comment{No update}

        \EndIf




    \end{algorithmic}
\end{algorithm}

Lines~\ref{line:greedy-MLE} and~\ref{line:greedy-a-optimal} in \opal involve the MLE estimation and the estimated optimal allocation calculation, both of which can incur nontrivial computation costs. While one may derive analytical expressions for the MLE and optimal allocation for some specific network tomography models (as shown in our two case studies in Sections~\ref{subsec:classical-unicast-case} and~\ref{subsec:quantum-multicast-case}), the general case may require computationally expensive numerical optimizations.
To reduce the computation cost, one may
replace Lines~\ref{line:greedy-MLE} and~\ref{line:greedy-a-optimal} in \opal with the lazy update in Algorithm~\ref{alg:lazy-chasing} (denoted as \opal-lazy).
The \opal-lazy algorithm updates the MLE estimates and the estimated optimal allocation every \(B\) time steps, and during the other steps keeps the estimated optimal allocation unchanged. Varying the input batch size \(B\) trades off computation cost against performance of \opal.


\section{Analysis}
\label{sec:analysis}


We start with two critical conditions for the network tomography tasks (Section~\ref{subsec:network-tomography-conditions}) and then characterizes the regret (Section~\ref{subsec:rate-analysis}) based on these conditions.

\subsection{Conditions on Network Tomography} \label{subsec:network-tomography-conditions}

We present two general conditions for the network tomography tasks, which are crucial for the regret analysis of \opal.
Both conditions are natural and typically fulfilled in practice.
Later in Sections~\ref{subsec:classical-unicast-case} and~\ref{subsec:quantum-multicast-case}, we verify both conditions in the classical and quantum network tomography case studies.

\begin{condition}[Lipschitz Continuity]\label{cond:lipschitz}
    The optimal experimental design criterion \(F(\bm\mu;\bm\phi)\) is Lipschitz continuous regarding the allocation \(\bm\phi\).
    That is, for constant \(\alpha > 0\) and any two allocations \(\bm\phi, \bm\phi' \in \mathcal{F}\) where \(\mathcal{F} \subseteq \Delta^{M-1}\) is the set of all feasible allocations that support identifiability,\footnote{
        The set \(\mathcal{F}\) contains all interior points and perhaps a part of the boundary points of the simplex set \(\Delta^{M-1}\).
        For example, in the classical loss star network tomography via unicasts in Section~\ref{subsec:classical-unicast-case}, all unicast probes are necessary for identifying the network parameters. Hence, any allocation on the boundary (i.e., allocations contain zero entries) does not correspond to a feasible allocation, and \(\mathcal{F}=\Delta^{M-1}\setminus \partial\Delta^{M-1}\).
    }
    we have \[
        F(\bm\mu;\bm\phi) - F (\bm\mu;\bm\phi') \le \alpha\norm{\bm\phi - \bm\phi'}_\infty.
    \]

    Further, the optimal allocation function \(\bm\phi^*(\bm\mu)\) is Lipschitz continuous with respect to \(\bm\mu\). That is, for some constant \(\beta > 0\) and any two sets of feasible network parameters \(\bm\mu, \bm\mu'\), we have \[
        \norm{\bm\phi^*(\bm\mu) - \bm\phi^*(\bm\mu')}_\infty \le \beta\norm{\bm\mu - \bm\mu'}_\infty.
    \]
\end{condition}

Lipschitz continuity is a common condition in practical scenarios, as it implies that the network tomography problem is well-behaved, and the optimal allocation \(\bm\phi^*\) is stable with respect to the network parameters \(\bm\mu\).


\begin{condition}[Estimator Concentration / Confidence Interval]\label{cond:finite-confidence-interval}
    At any decision round \(t\), the confidence interval (with confidence \(1-\delta\) for parameter \(\delta \in (0,1)\))  for any parameter \(\mu_\ell\) of a link \(\ell\in\mathcal{L}\) is an interval centered at its MLE estimate \(\hat{\mu}_{\ell,t}\) with radius \(\sum_{m\in\mathcal{M}} c_{\ell,m} \left( \log \delta^{-1} / S_{m,t} \right)^{\gamma_{\ell,m}}\), where \(c_{\ell,m}\ge 0, \gamma_{\ell,m}>0\) are parameters depending on the network and tomography probes, and \(S_{m,t}\) is the number of times that probe \(m\) is performed up to round \(t\).
\end{condition}

Condition~\ref{cond:finite-confidence-interval} describes the concentration rate of the MLE estimator (or any other estimator) in terms of the sample times \(S_{m,t}\) of all probes \(m\).
With the central limit theorem, one can derive an asymptotic version of the confidence interval with radius \(c_0\sqrt{1/{S_{m,t}}}\)~\citep[Theorem 4]{caceres1999multicast},
for some positive constant \(c_0\).
This implies that the concentration rate \(\gamma_{\ell,m}\) of the probing experimental \(m\) for the parameter of link \(\ell\) approaches \(1/2\) for large sample sizes.

\begin{figure*}[tb]
    \centering
    \begin{subfigure}{.3\textwidth}
        \centering
        \begin{tikzpicture}
            \draw[ultra thick] (0,5) -- (4,5); 
            \draw[dashed, white] (0,5.5) -- (4,5.5); 
            \draw[dashed, white] (0,4.5) -- (4,4.5); 
            \node[below] at (2, 4) {\(\hat\phi_{m,t}^*\) is constant over time \(t\)}; 
        \end{tikzpicture}
        \caption{\(\abs*{\frac{S_{m,t}}{t}  - \hat\phi_{m,t}^*} \le \frac 1 t\)}
        \label{fig:horizontal_line}
    \end{subfigure}%
    \hfill
    \begin{subfigure}{.3\textwidth}
        \centering
        \begin{tikzpicture}
            \draw[ultra thick]
            (0,5) --
            (0.25,4.65) --
            (0.5,5.2) --
            (0.75,4.8) --
            (1,4.9) --
            (1.25,5.1) --
            (1.5,5.4) --
            (1.75,4.75) --
            (2,4.7) --
            (2.25,5.1) --
            (2.5,4.63) --
            (2.75,5.3) --
            (3,4.9) --
            (3.25,5.2) --
            (3.5,5.3) --
            (3.75,4.85) --
            (4,5.3);
            \draw[dashed] (0,5.5) -- (4,5.5); 
            \draw[dashed] (0,4.5) -- (4,4.5); 
            \node[below] at (2, 4.3) {\(\hat\phi_{m,t}^*\in \left( \phi_{m}^* - \epsilon, \phi_{m}^* + \epsilon \right)\)}; 
        \end{tikzpicture}
        \caption{\(\abs*{\frac{S_{m,t}}{t}  - \hat\phi_{m,t}^*} \le 2\epsilon + \frac 1 t\)}
        \label{fig:zigzag_line}
    \end{subfigure}%
    \hfill
    \begin{subfigure}{.3\textwidth}
        \centering
        \begin{tikzpicture}

            \draw[ultra thick]
            (0.5,5.8) --
            (0.75,4.3) --
            (1,5.5) --
            (1.25,5.1) --
            (1.5,5.4) --
            (1.75,4.75) --
            (2,4.7) --
            (2.25,5.1) --
            (2.5,4.7) --
            (2.75,5.3) --
            (3,4.9) --
            (3.25,4.7) --
            (3.5,5.2) --
            (3.75,4.85) --
            (4,5);

            \draw[dashed, domain=0.5:4, samples=100] plot (\x, {5 + 1/(sqrt(2 * \x))});
            \draw[dashed, domain=0.5:4, samples=100] plot (\x, {5 - 1/(sqrt(2 * \x))});

            \draw[dashed, blue] plot (2,4) -- (2,6);
            \node[below, blue] at (2, 4) {\(\xi t\)};
            \draw[dashed, blue] plot (4,4) -- (4,6);
            \node[below, blue] at (4, 4) {\(t\)};
            \draw[dotted, blue] plot (2,5.5) -- (4,5.5);
            \draw[dotted, blue] plot (2,4.5) -- (4,4.5);
            \node[below] at (3, 4.5) {\footnotesize{Case (b)}};
        \end{tikzpicture}
        \caption{\(\abs*{\frac{S_{m,t}}{t}  - \hat\phi_{m,t}^*} \lessapprox {\frac{2C}{(\xi t)^{\gamma_{\min}}}} + \frac 1 t\)}
        \label{fig:shrinking_zigzag_line}
    \end{subfigure}%
    \caption{A high-level illustration of the novel bounding technique used in Proof of Theorem~\ref{thm:convergence-rate} (esp.~\eqref{eq:phi_star_diff}).
        The thick solid line is the chased allocation \(\hat{\phi}_{m,t}^*\) varying over time, and the dashed black line is the confidence interval of the ground truth allocation \(\phi_{m}^*\).
        The derivation of~\eqref{eq:phi_star_diff} comes from two simplifications of the chasing step of Line~\ref{line:chase} of Algorithm~\ref{alg:chase} in Figs.~\ref{fig:horizontal_line} and~\ref{fig:zigzag_line}.
        Fig.~\ref{fig:horizontal_line}: if the chased estimated optimal allocation \(\hat{\bm\phi}_t^*\) is fixed over time \(t\), then the actual allocation \(\bm\phi_t\) is close to the chased allocation up to a distance of \(1/t\);
        Fig.~\ref{fig:zigzag_line}: if the chased allocation \(\hat{\bm\phi}_t^*\) is changing over time \(t\) but its variation is bounded by an interval with width \(2\epsilon\), then the empirical allocation is close to the final chased allocation up to \(2\epsilon + 1/t\).
        These two simplifications lead to~\eqref{eq:phi_star_diff}, illustrated in Fig.~\ref{fig:shrinking_zigzag_line}.
    }
    \label{fig:bound-illustration}
\end{figure*}

\subsection{Regret (Convergence Rate) of \opal}\label{subsec:rate-analysis}

In this section, we present the main theorem that characterizes the regret of \opal.
We recall the probe allocation ratio of the algorithm at time \(t\) as \(\bm\phi_t \coloneqq ({S_{m,t}}/{t})_{m\in\mathcal{M}}\).

\begin{theorem}\label{thm:convergence-rate}
    Given Conditions~\ref{cond:lipschitz} and~\ref{cond:finite-confidence-interval} with \(\delta = 1/(LT^2)\),
    for time horizon \(T>0\), initial probe times \({S}_{m,0} = \xi T\), \(m\in\mathcal{M}\), for any \(\xi \in (0,1)\), with probability of at least \(1- 1/T\), \emph{\opal} satisfies,
    \begin{equation}\label{eq:convergence-rate}
        \begin{split}
            R_T
             & = F(\bm\mu; \bm\phi_T) - F(\bm\mu; \bm\phi^*) \le C \left( \frac{\log (LT)}{\xi T} \right)^{\gamma_{\min}}
            \\
             & \quad+ \alpha \xi \1*{\exists m\in\mathcal{M}: \phi_m^* < \xi + C\left( \frac{\log (LT)}{\xi T} \right)^{\gamma_{\min}}},
        \end{split}
    \end{equation}
    where \(C = c_0\cdot\alpha\beta c_{\max}\) is for some constant \(c_0>0\) that depends on the probing experiments and network, independent of time horizon \(T\).
    \(\alpha\) and \(\beta\) come from Condition~\ref{cond:lipschitz},
    parameters \(c_{\max} \coloneqq \max_{\ell\in\mathcal{L}} \sum_{m\in\mathcal{M}} c_{\ell,m}\) and \(\gamma_{\min} \coloneqq \min_{(\ell, m)\in\mathcal{L}\times\mathcal{M}:\gamma_{\ell,m}>0} \gamma_{\ell,m}\) come from Condition~\ref{cond:finite-confidence-interval},
    and \(\1{\cdot}\) is the indicator function.
\end{theorem}

\textbf{Implications of Theorem~\ref{thm:convergence-rate}.} Theorem~\ref{thm:convergence-rate} characterizes the regret of \opal. To better understand the regret, we discuss several special cases.

\emph{First}, for many network tomography tasks, even without knowing the actual network channel parameters, one can select a compact set of probing experiments \(\mathcal{M}\) such that each probe in the set is needed at least a minimum fraction of the time---that is, \(\phi^*_m > \phi_{\text{threshold}}\) for all probes \(m\in\mathcal M\)---to identify the network parameters.
One can set \(\xi\) in Theorem~\ref{thm:convergence-rate} to this threshold, and then the second term in the RHS of~\eqref{eq:convergence-rate} vanishes, leading to a \(\tilde{O}\left( T^{-\gamma_{\min}} \right)\) regret.
For example, in the case study in Section~\ref{subsec:classical-unicast-case}, spreading the first \(O(\log T)\) time steps uniformly across each probe yields a rough estimate of the network parameters.
With these estimates, one can find such a threshold and set \(\xi\) as the threshold to cancel the second term.

\emph{Second}, one can always set \(\xi = T^{-\frac{\gamma_{\min}}{1 + \gamma_{\min}}}\) in Theorem~\ref{thm:convergence-rate}, guaranteeing a decent \(\tilde O(T^{-\frac{\gamma_{\min}}{1 + \gamma_{\min}}})\) regret, independent of whether the second term of the regret in~\eqref{eq:convergence-rate} can be canceled or not.
For example, the second term cannot be canceled by any constant \(\xi > 0\) in a scenario where the optimal allocation \(\bm\phi^*\) contains some zero entries, i.e., \(\phi_m^* = 0\) for some probe $m$ (e.g., the case study in Section~\ref{subsec:quantum-multicast-case}).
That is,
when there are unavoidable inferior probes, letting \(\xi = T^{-\frac{\gamma_{\min}}{1 + \gamma_{\min}}}\) is a plausible choice, and \opal yields a \(\tilde O(T^{-\frac{\gamma_{\min}}{1 + \gamma_{\min}}})\) regret, slightly worse than the \(-\gamma_{\min}\) rate.

\emph{Third}, from an asymptotic perspective (for large sample sizes) and by the central limit theorem, the concentration rate of Condition~\ref{cond:finite-confidence-interval} is \(\gamma_{\ell,m} = 1/2\) for all probes \(m\) and links \(\ell\).
Assuming the second term in~\eqref{eq:convergence-rate} is canceled, the regret of \opal becomes \(\tilde{O}(T^{-\frac 1 2})\).
This rate is near-optimal in the sense that, even for the online linear regression (simpler than our network tomography task), minimizing the regret for A-optimal design criterion (with the number of probes greater than that of links, \(M>L\)) is known to suffer at least a \(\Omega(T^{-1/2})\) regret~\citep[Theorem 4]{fontaine2021online}.

\subsubsection{Proof Overview of Theorem~\ref{thm:convergence-rate}}\label{subsubsec:proof-convergence-rate}
\opal utilizes the actual allocation \({{\bm\phi}_t = ({S_{m,t}}/{t})_{m\in\mathcal{M}}}\) to \emph{``chase''} the optimal allocation \(\bm\phi^*\) (Line~\ref{line:chase}).
To show convergence of the algorithm, we need to analyze the effectiveness of such a chasing strategy.
That is, \emph{how good is the actual allocation \({\bm\phi}_t\) in approximating the optimal allocation \(\bm\phi^*\) as the algorithm proceeds?}
Such an analysis is challenging due to the sequential decision-making nature of the process: On one hand, as the actual network parameters \(\bm\mu\) are unknown, the algorithm does not know the optimal allocation \(\bm\phi^*\) and, hence, needs to {\em estimate} the  optimal allocation \(\hat{\bm\phi}_t^*\) based on the parameter estimates \(\hat{\bm\mu}_t\).
This introduces an \emph{accuracy gap} in the chasing process.
The actual allocation \({\bm\phi}_t\) depends on all past probe decisions, which often lag behind the \emph{time-varying} estimated optimal allocation \(\hat{\bm\phi}_t^*\). This introduces a \emph{lag gap}. The proof of Theorem~\ref{thm:convergence-rate} bounds these two gaps, and
the high-level derivation technique is illustrated in Fig.~\ref{fig:bound-illustration}.
We refer to Appendix~\ref{sec:proof-convergence-rate} for the complete proof of Theorem~\ref{thm:convergence-rate}.

\section{Case Study for Classical Network Tomography}
\label{sec:case-study-classical}

This section presents a case study on classical network tomography. In Section~\ref{subsec:classical-unicast-case}, we investigate
loss tomography in a star network
and verify the Lipschitz continuity of the A-optimal design and the concentration rate of the MLE estimator for the link parameters. We then provide empirical results in Section~\ref{subsec:empirical-classical}  for both star and general networks with comparisons against known baselines.

\subsection{Analytical Case 1: Classical Star Network Unicast Setting}\label{subsec:classical-unicast-case}

In this analytical case study, we examine loss tomography in a star network with \(L\) links, and we focus solely on unicast probes due to their analytical tractability.
An example with \(3\) links is shown in Figure~\ref{fig:star_network}. For the \(L\)-link star network, we select a set of \(M = L\) probes that ensure identifiability of the link parameters \(\bm\mu\).
These \(M\) probes are represented by an \(M \times L\) matrix \(\bm Q \in \{0,1\}^{M \times L}\), known as the \emph{measurement matrix}, where \(Q_{m,\ell} = 1\) if the \(m\)-the unicast probe passes through the \(\ell\)-th link, and \(Q_{m,\ell} = 0\) otherwise. In other words, each row of \(\bm Q\) corresponds to a probing experiment, with the two non-zero entries in each row indicating the source and destination of the unicast probe.
Denote by \(\nu_m \coloneqq \mu_{\ell_{m,1}} \mu_{\ell_{m,2}}\)the success probability of the \(m\)-th probe, where \(\ell_{m,1}\) and \(\ell_{m,2}\) represent the source and destination nodes/links of the \(m\)-th probe, respectively. With measurement matrix \(\bm Q\), we have \({\log \bm \nu = \bm Q \log \bm \mu}\).\footnote{The \(\log\) denotes the element-wise natural logarithm.}

\textbf{Lipschitz Continuity of the A-optimal Design.}
We select the A-optimal experimental design as the objective to verify the Lipschitz continuity in Condition~\ref{cond:lipschitz}. Based on the derivation in~\citet[Theorem 6]{he2015fisher}, we compute the trace of the inverse of the Fisher information matrix \(\bm I(\bm\mu; \bm\phi)\) for the \(L\)-link star network as follows,
\begin{equation}\label{eq:classical-A-optimal-criterion}
    F(\bm\mu; \bm\phi) = \tr \bm I^{-1} (\bm\mu; \bm\phi) =
    \sum_{m=1}^M \frac{1}{\phi_m}  A_m(\bm\mu),
\end{equation}
where
\(
A_m (\bm \mu) \coloneqq \frac{1 - \nu_m}{\nu_m} \sum_{\ell = 1}^L \mu_\ell^2 \kappa_{\ell, m}
\) is a function of the link parameters \(\bm\mu\),
and \(\kappa_{\ell, m} \coloneqq (\bm Q^{-1})_{\ell, m}\) denotes the \((\ell, m)\)-entry of the inverse of the measurement matrix \(\bm Q\).

By the method of Lagrange multipliers, the A-optimal solution for minimizing \(F(\bm\mu; \bm\phi)\) is
\begin{equation}\label{eq:classical-A-optimal-allocation}
    \phi_m^* = \frac{\sqrt{A_m(\bm\mu)}}{\sum_{m' \in \mathcal{M}} \sqrt{A_{m'}(\bm\mu)}}, \quad \forall m \in \mathcal{M}.
\end{equation}
As the analytical solutions in~\eqref{eq:classical-A-optimal-criterion} and~\eqref{eq:classical-A-optimal-allocation} only contain basic calculations and are thus differentiable (recall \(\mu_\ell\in (0,1)\)), we verify the Lipschitz continuity in Condition~\ref{cond:lipschitz}.

\textbf{MLE Estimator for Link Parameters.}
The MLE estimator for the success probability of each probe \(\nu_m\) is \(\hat{\nu}_m = \frac{B_m}{S_m}\), where \(S_m\) is the number of times the probe \(m\) is performed, and \(B_m\)
is the number of successful receptions among the \(S_m\) probes.
Link parameters \(\bm\mu\) can be expressed as a function of the success probabilities \(\bm \nu\) in the form of
\({
        \log \bm\mu = (\bm Q^T \bm Q)^{-1} \bm Q^T \log \bm\nu.
    }\)
Using the invariance property of MLE under one-to-one transformations, the MLE for the link parameters \(\bm\mu\) is
\begin{equation}
    \label{eq:classical-MLE-estimator}
    \hat{\bm\mu} = \exp\left( (\bm Q^T \bm Q)^{-1} \bm Q^T \log \hat{\bm \nu}  \right)
    \overset{(a)}= \exp\left( \bm Q^{-1} \log \hat{\bm\nu} \right),
\end{equation}
where equality (a) is a consequence of matrix \(\bm Q\) being square and invertible for the star network.
This matrix form in~\eqref{eq:classical-MLE-estimator} can be further expressed in an element-wise
form as
\(\hat\mu_\ell = \prod_{m\in\mathcal{M}} \hat\nu_m^{\kappa_{\ell, m}}\) for every link \(\ell \in \mathcal{L}\).

\textbf{Concentration Rate of the MLE Estimator.}
Next, we examine the concentration rate in Condition~\ref{cond:finite-confidence-interval} for MLE in~\eqref{eq:classical-MLE-estimator} in the following theorem. Its proof is provided in Appendix~\ref{app:proof-case-study-classical}.

\begin{theorem}\label{thm:MLE-estimator-fix-radius}
    For unicast star network, we have, with a probability of at least \(1-\delta\) (for any $\delta\in (0,M^{-1})$),
    \begin{align}\label{eq:concentration-interval-classical}
        \mu_\ell \in (\hat\mu_\ell - \text{\normalfont rad}_{m,\ell}, \hat\mu_\ell + \text{\normalfont rad}_{m,\ell})
    \end{align} where
    \(\text{\normalfont rad}_{m,\ell} = \sum_{m\in\mathcal{M}:  \kappa_{\ell, m}\neq 0} c_{m,\ell} \sqrt{{\log (\delta^{-1}/M)}/ {S_m}}\),
    and \(c_{m,\ell}\)'s are positive constants depending on the network topology and link parameters.
\end{theorem}

This verifies Condition~\ref{cond:finite-confidence-interval} regarding the concentration rate of the MLE estimator for the link parameters in the unicast star network. Specifically, the exponent parameters are \(\gamma_{m,\ell} \in \{0,\frac{1}{2}\}\), resulting in a minimum nonzero exponent \(\gamma_{\min} = \frac{1}{2}\).

\textbf{Regret of \opal on the Classical Star Network.}
With the confidence interval in~\eqref{eq:concentration-interval-classical} and \(\delta = 1/(LT^2)\),
we know that \(O(\log T)\) samples for each unicast probe suffice to estimate the link parameters \(\bm\mu\) with relatively good accuracy.
Furthermore, utilizing the Lipschitz continuity of the A-optimal design allocation in~\eqref{eq:classical-A-optimal-allocation}, one can derive a lower bound \(\xi\)
for the optimal allocation ratio \(\phi_m^*\) for all probes \(m \in \mathcal{M}\). By using this constant lower bound \(\xi\) as input to Theorem~\ref{thm:convergence-rate}, the second term of the general regret bound in~\eqref{eq:convergence-rate} is canceled. Consequently, the regret of \opal for the A-optimal design under the loss star network is
\[
    R_T = \tr \bm I^{-1}(\bm\mu; {\bm\phi}_T) - \tr\bm I^{-1}(\bm\mu; \bm\phi^*) = O\left( \sqrt{\frac{\log T}{ T}} \right).
\]


\textbf{Complexity of \opal for Classical Network.}
In each time slot \(t\) of \opal for classical network tomography, the major computational cost is from calculating~\eqref{eq:classical-A-optimal-allocation} (Line~\ref{line:greedy-a-optimal} in \opal), whose time complexity is \(O(ML+L^3) = O(L^3)\).
As~\eqref{eq:classical-A-optimal-allocation} is derived from general classical networks, this time complexity is also valid for general classical networks.


\subsection{Simulations for Classical Network Tomography}\label{subsec:empirical-classical}

We evaluate the empirical performance of the proposed online probe allocation (\opal) algorithm in terms of A-optimal design for loss tomography in general classical networks under unicast (beyond the star network).
We consider three baselines: an iterative algorithm proposed in~\citep[Algorithm 2]{he2015fisher}, the A-Optimal allocation, and a uniform allocation.
The iterative algorithm is a state-of-the-art experimental design solution for network tomography without prior knowledge of the link parameters.
It performs probes in batch (each consisting of \(100\) time steps) according to its iteratively updated policy.
The A-optimal allocation uses \(\bm\phi^*\), minimizes the A-optimal experimental design, and represents an ideal solution with full knowledge of the link parameters,
and the uniform allocation represents a simple baseline.

\begin{figure*}[tp]
    \centering
    \begin{subfigure}{.3\textwidth}
        \centering
        \includegraphics[width=.9\linewidth]{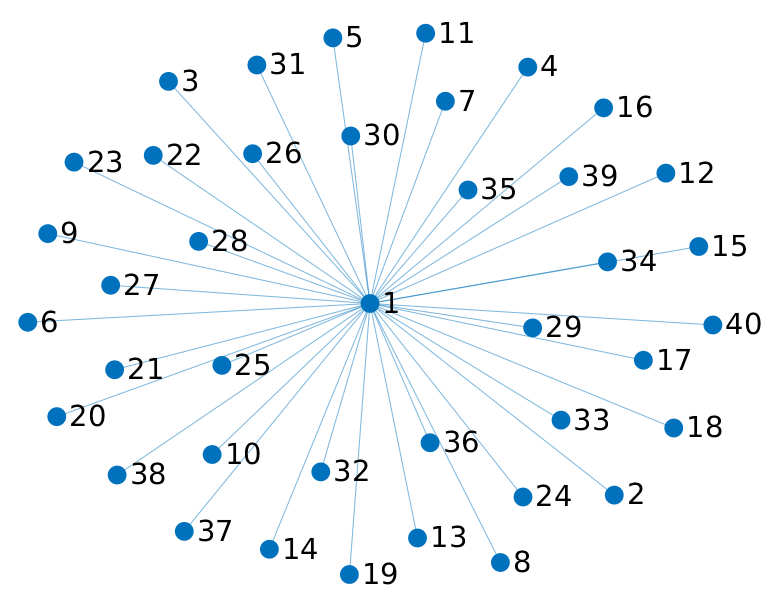}
        \caption{\(40\)-node star network}\label{subfig:star-illustration}
    \end{subfigure}%
    \begin{subfigure}{.3\textwidth}
        \centering
        \includegraphics[width=.9\linewidth]{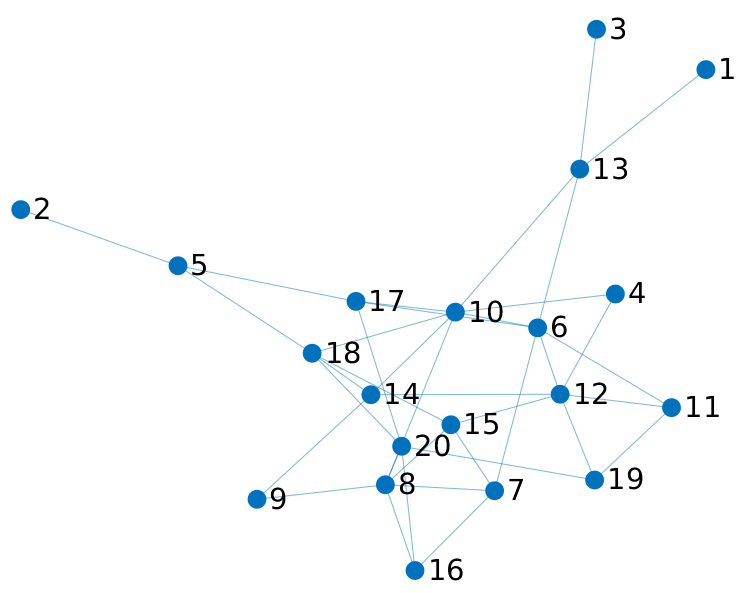}
        \caption{ER graph with \(20\) nodes}
        \label{subfig:ER-illustration}
    \end{subfigure}%
    \begin{subfigure}{.3\textwidth}
        \centering
        \includegraphics[width=.9\linewidth]{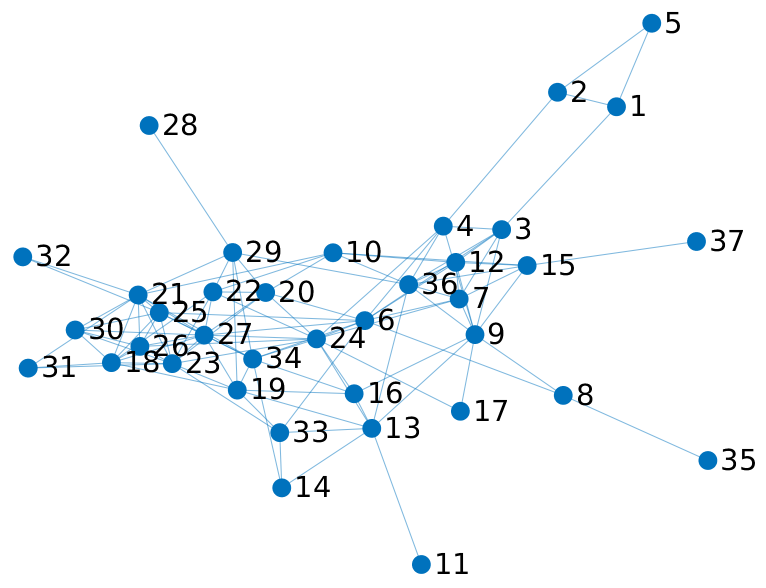}
        \caption{Roofnet with \(37\) nodes}\label{subfig:roofnet-illustration}
    \end{subfigure}%
    \caption{Topology illustration for the three types of networks used in our experiments: (a) A \(40\)-node star network, where one central node is connected to all other nodes. (b) A random Erdős-Rényi (ER) network with \(20\) nodes and approximately \(35\) links, and (c) a Roofnet topology with \(37\) nodes and \(114\) links, based on real-world data~\citep{aguayo2004link}.
    }
    \label{fig:topology-illustration}
\end{figure*}

\begin{figure*}[tp]
    \centering
    \begin{subfigure}{.25\textwidth}
        \centering
        \includegraphics[width=\linewidth]{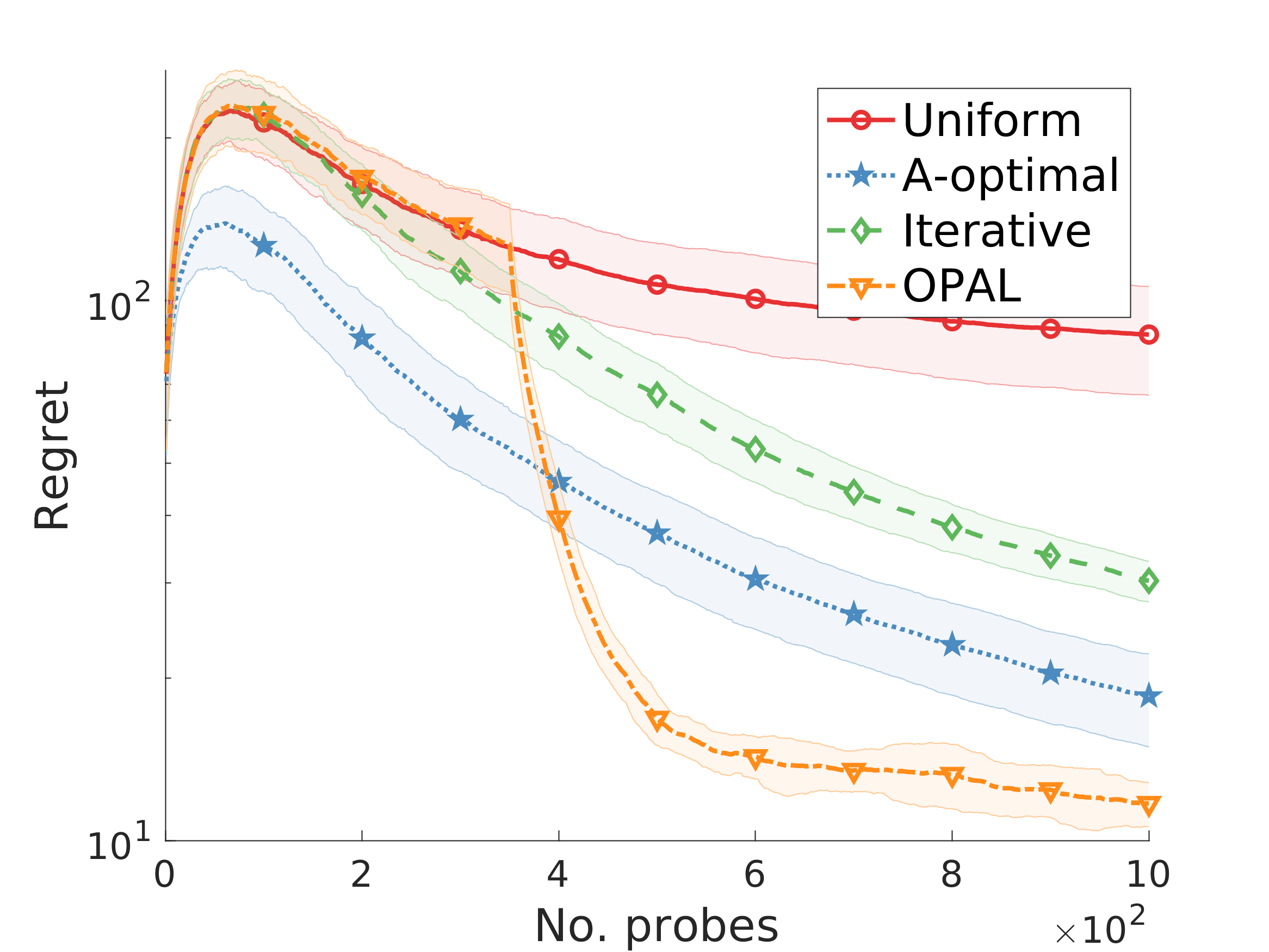}
        \caption{Regret: $R_t{=}F\!(\bm\mu,\!\bm\phi_t\!) {-} F\!(\bm\mu,\!\bm\phi^*\!)$}
    \end{subfigure}%
    \begin{subfigure}{.25\textwidth}
        \centering
        \includegraphics[width=\linewidth]{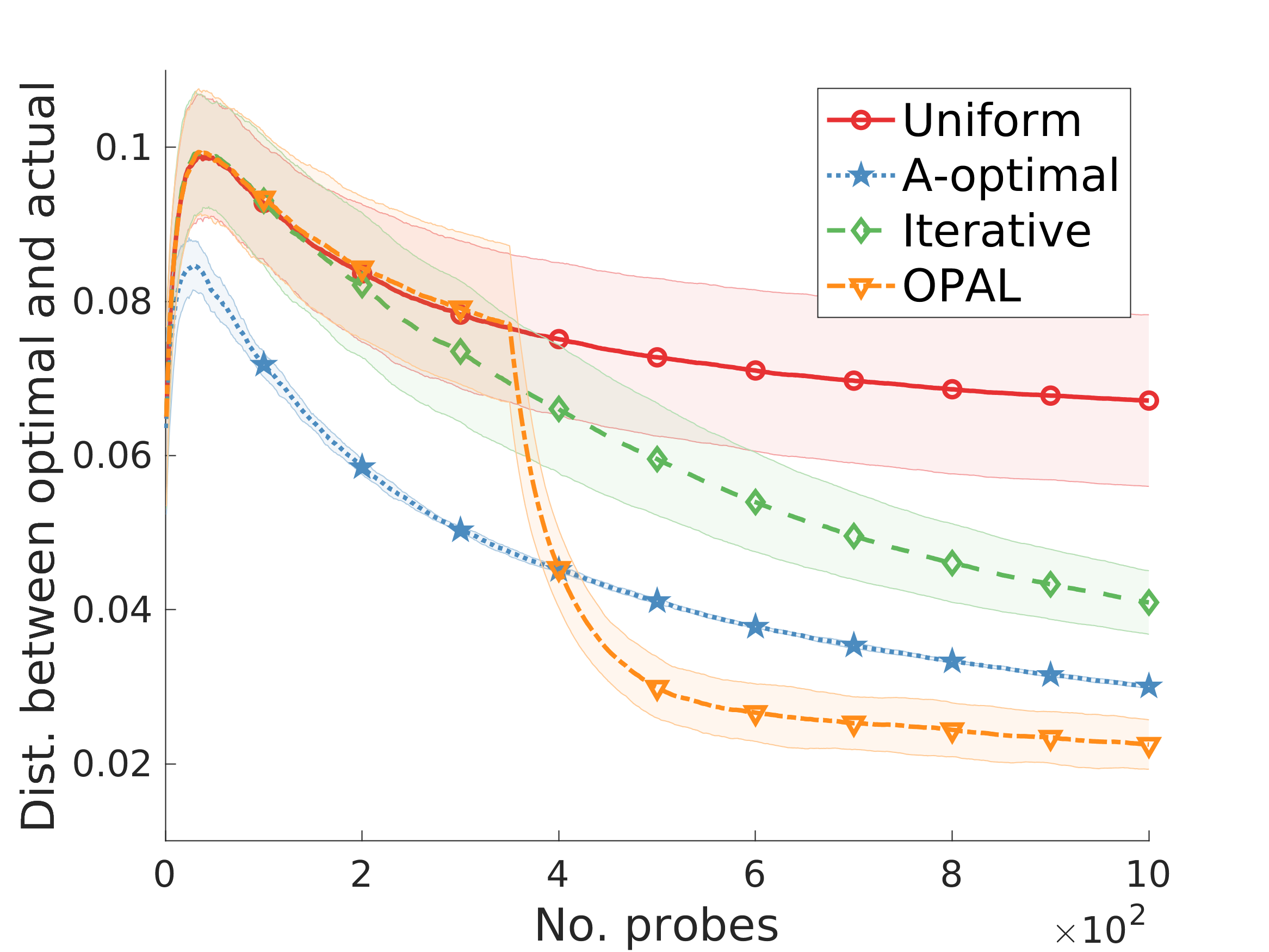}
        \caption{Dist. of actual: \(\norm*{\bm{\phi}^* - {\bm\phi}_t}_2\)}
    \end{subfigure}%
    \begin{subfigure}{.25\textwidth}
        \centering
        \includegraphics[width=\linewidth]{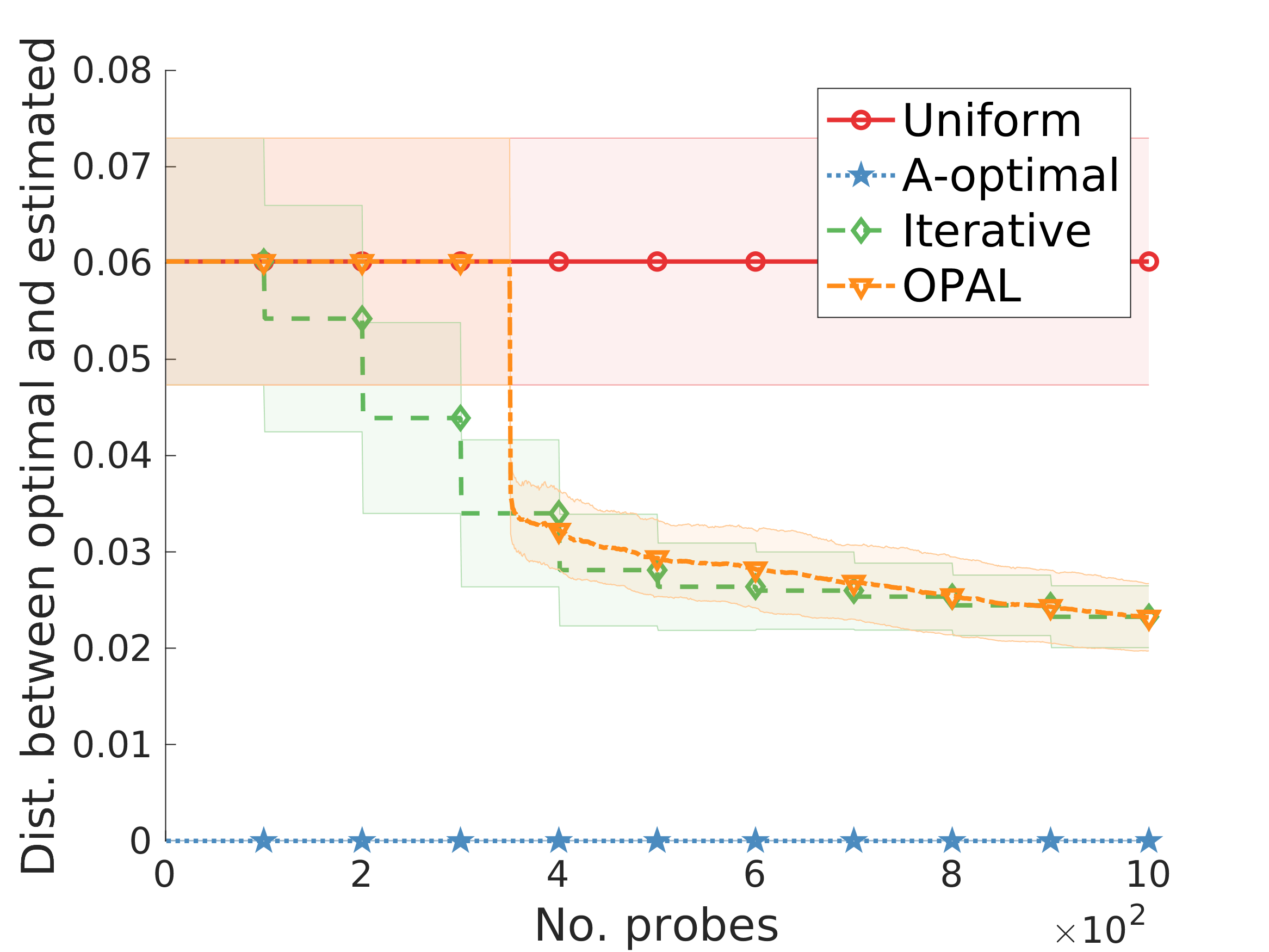}
        \caption{Dist. of estimated: \(\norm*{\bm{\phi}^* - \hat{\bm\phi}_t}_2\)}
    \end{subfigure}%
    \begin{subfigure}{.25\textwidth}
        \centering
        \includegraphics[width=\linewidth]{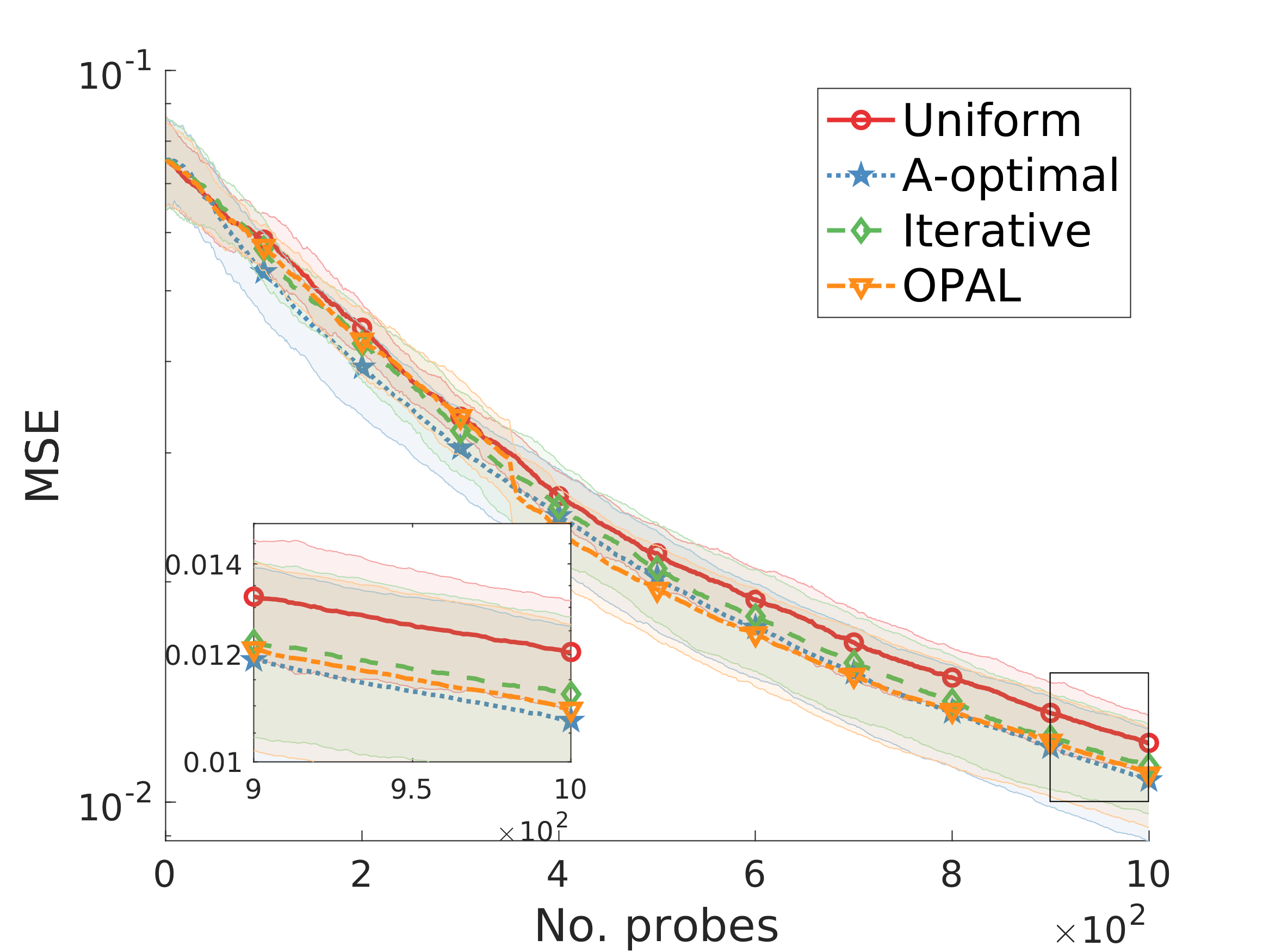}
        \caption{MSE: \(\mathbb{E}[\norm{\hat{\bm\mu}_t - \bm\mu}_2^2]\)}\label{subfig:star-mse}
    \end{subfigure}%
    \caption{Comparison of algorithms for $40$-node star networks (illustrated in Figure~\ref{subfig:star-illustration})}
    \label{fig:classical-star-3}
\end{figure*}

\begin{figure*}[tp]
    \centering
    \begin{subfigure}{.25\textwidth}
        \centering
        \includegraphics[width=\linewidth]{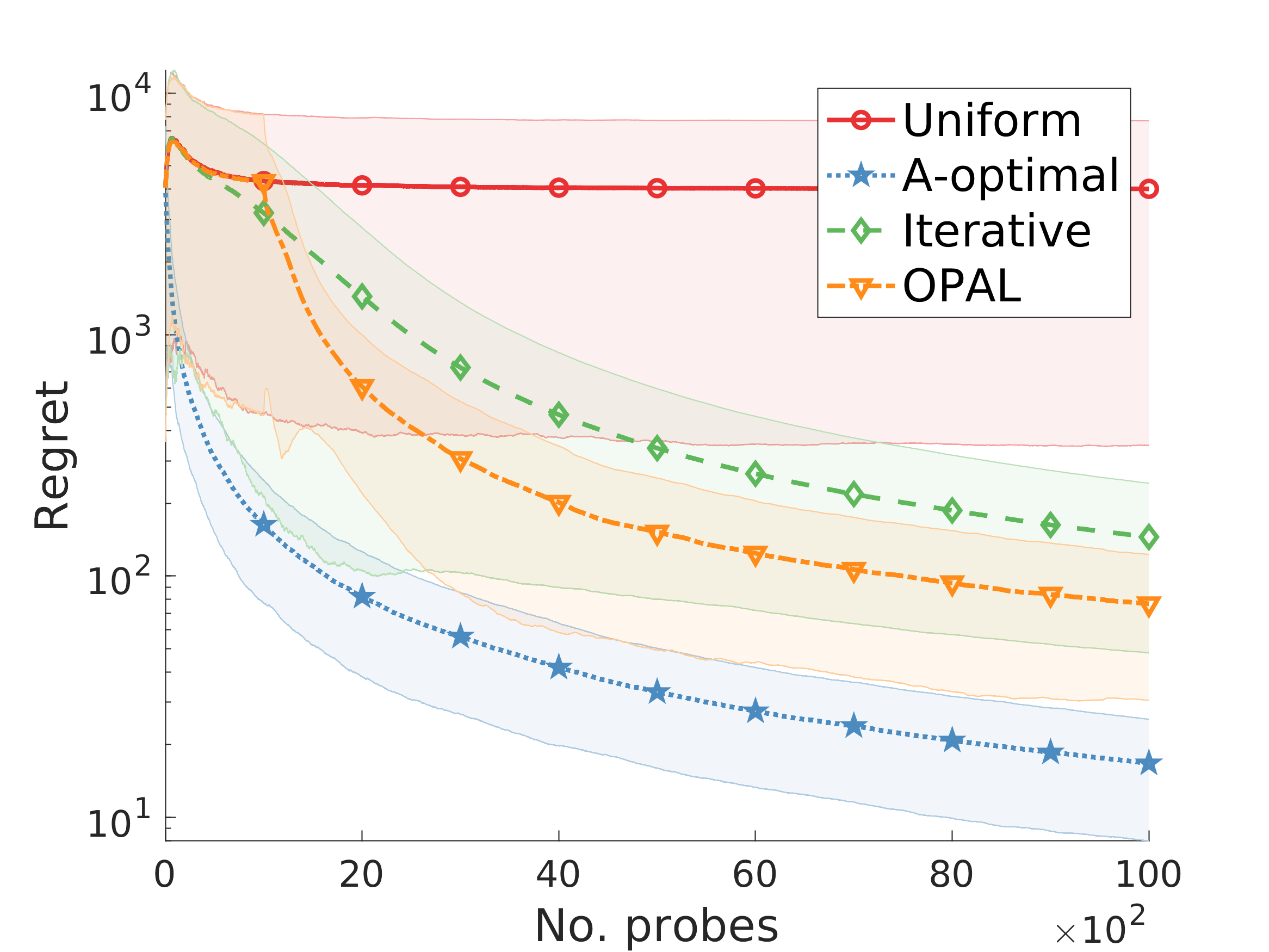}
        \caption{Regret: $R_t{=}F\!(\bm\mu,\!\bm\phi_t\!) {-} F\!(\bm\mu,\!\bm\phi^*\!)$}
    \end{subfigure}%
    \begin{subfigure}{.25\textwidth}
        \centering
        \includegraphics[width=\linewidth]{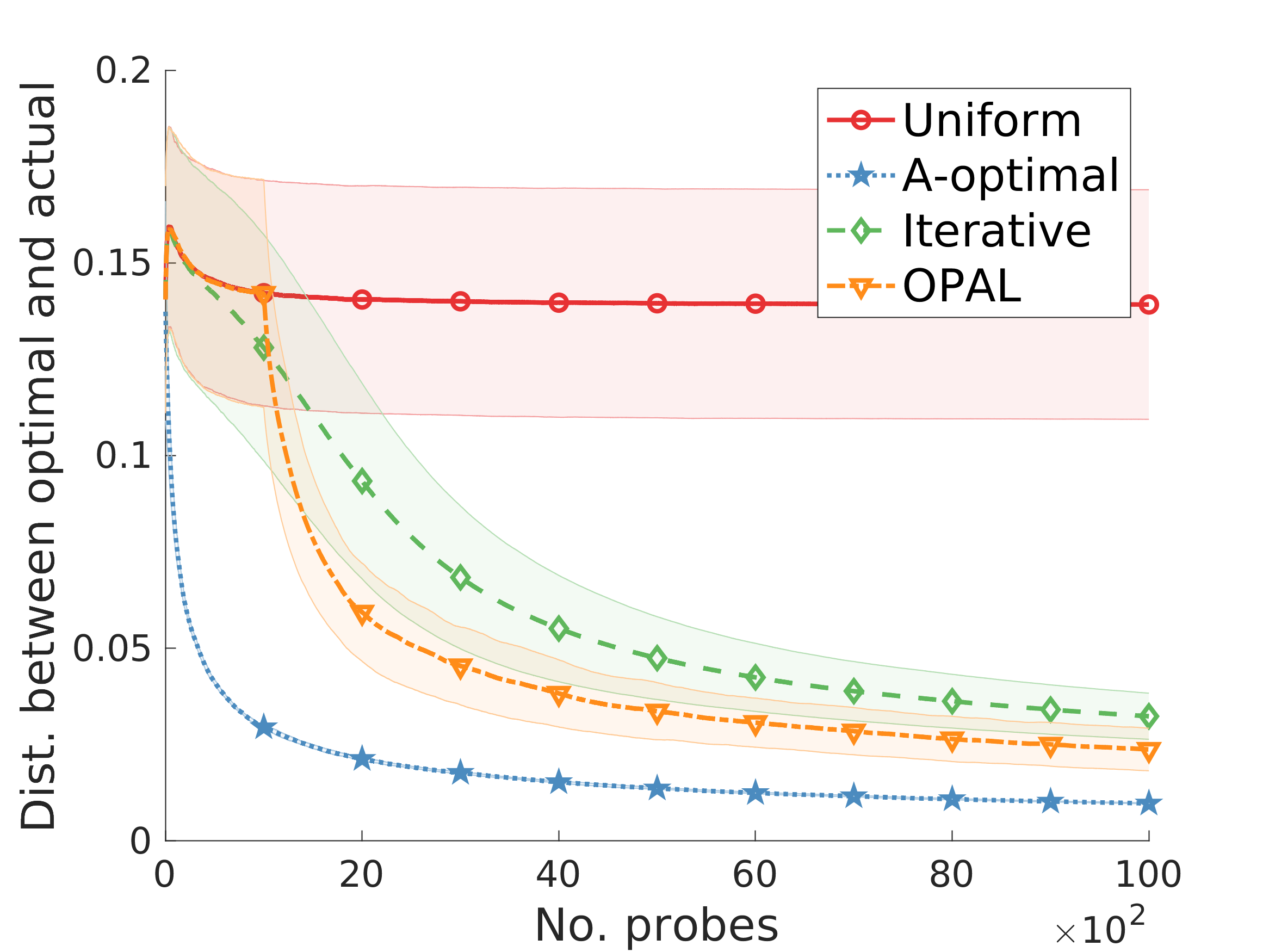}
        \caption{Dist. of actual: \(\norm*{\bm{\phi}^* - {\bm\phi}_t}_2\)}
    \end{subfigure}%
    \begin{subfigure}{.25\textwidth}
        \centering
        \includegraphics[width=\linewidth]{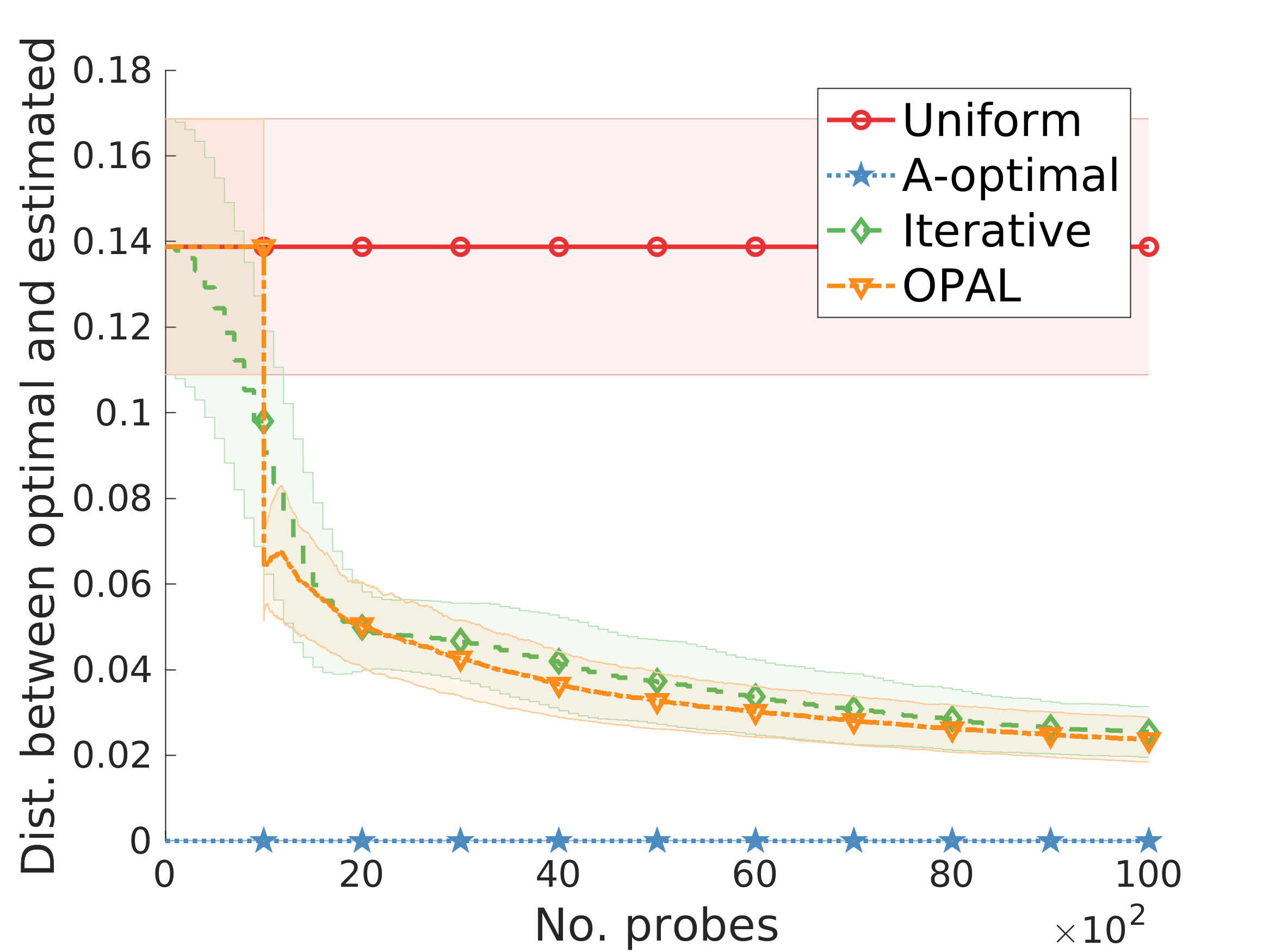}
        \caption{Dist. of estimated: \(\norm*{\bm{\phi}^* - \hat{\bm\phi}_t}_2\)}
    \end{subfigure}%
    \begin{subfigure}{.25\textwidth}
        \centering
        \includegraphics[width=\linewidth]{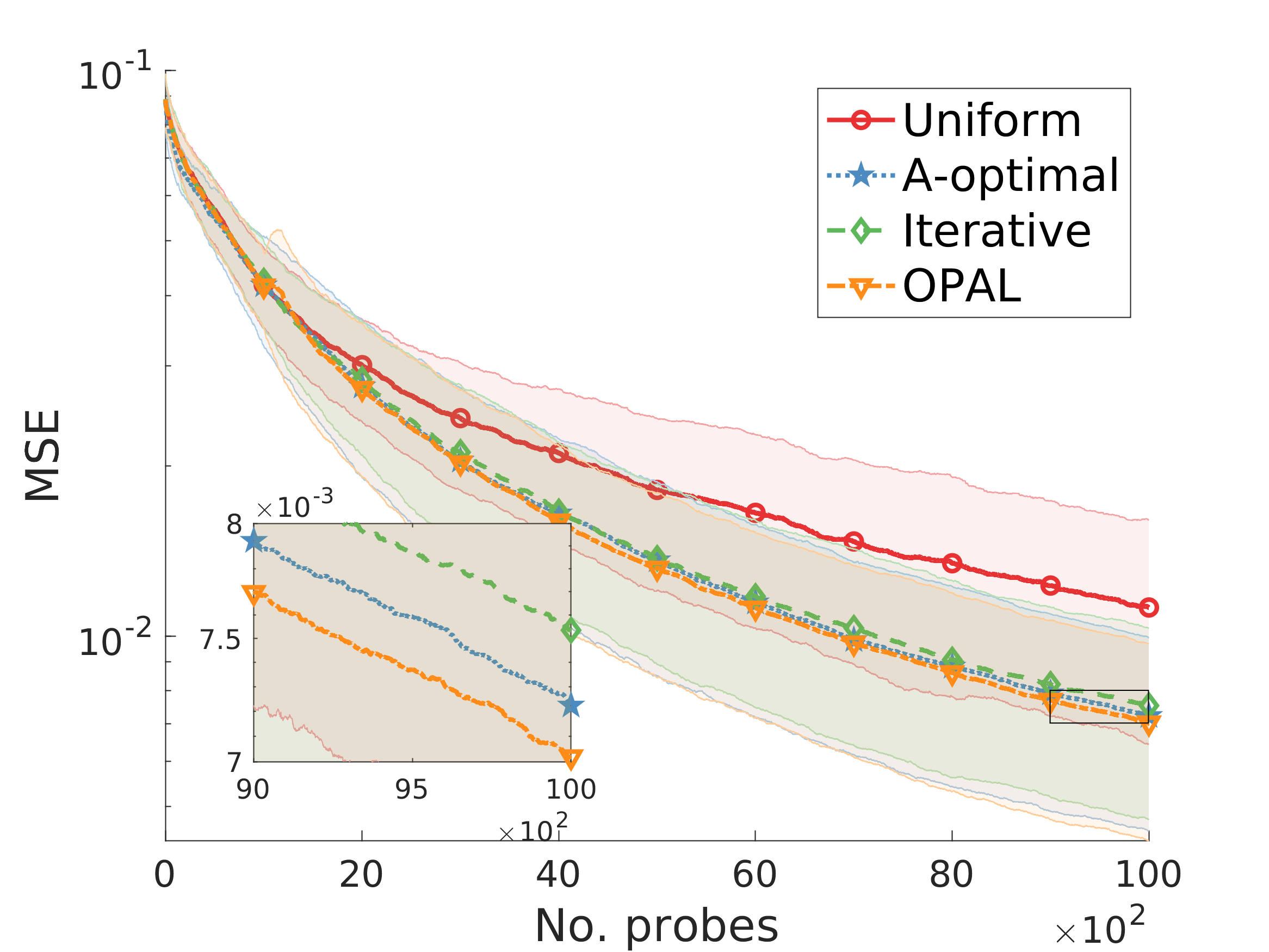}
        \caption{MSE: \(\mathbb{E}[\norm{\hat{\bm\mu}_t - \bm\mu}_2^2]\)}\label{subfig:er-mse}
    \end{subfigure}%
    \caption{Comparison of algorithms for \(20\)-node ER networks (illustrated in Figure~\ref{subfig:ER-illustration})
    }
    \label{fig:classical-ER}
\end{figure*}

\begin{figure*}[tp]
    \centering
    \begin{subfigure}{.25\textwidth}
        \centering
        \includegraphics[width=\linewidth]{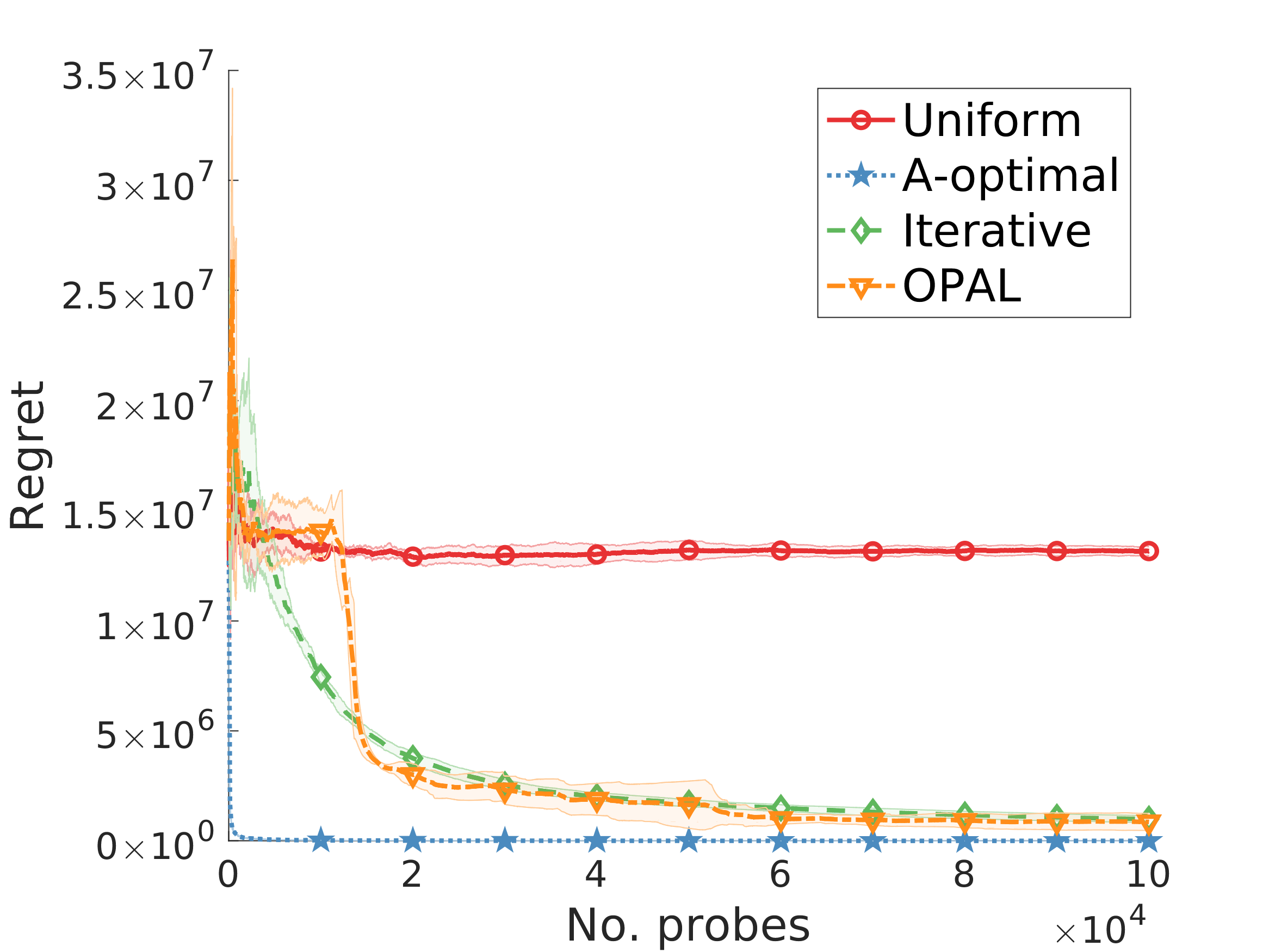}
        \caption{Regret: $R_t{=}F\!(\bm\mu,\!\bm\phi_t\!) {-} F\!(\bm\mu,\!\bm\phi^*\!)$}
    \end{subfigure}%
    \begin{subfigure}{.25\textwidth}
        \centering
        \includegraphics[width=\linewidth]{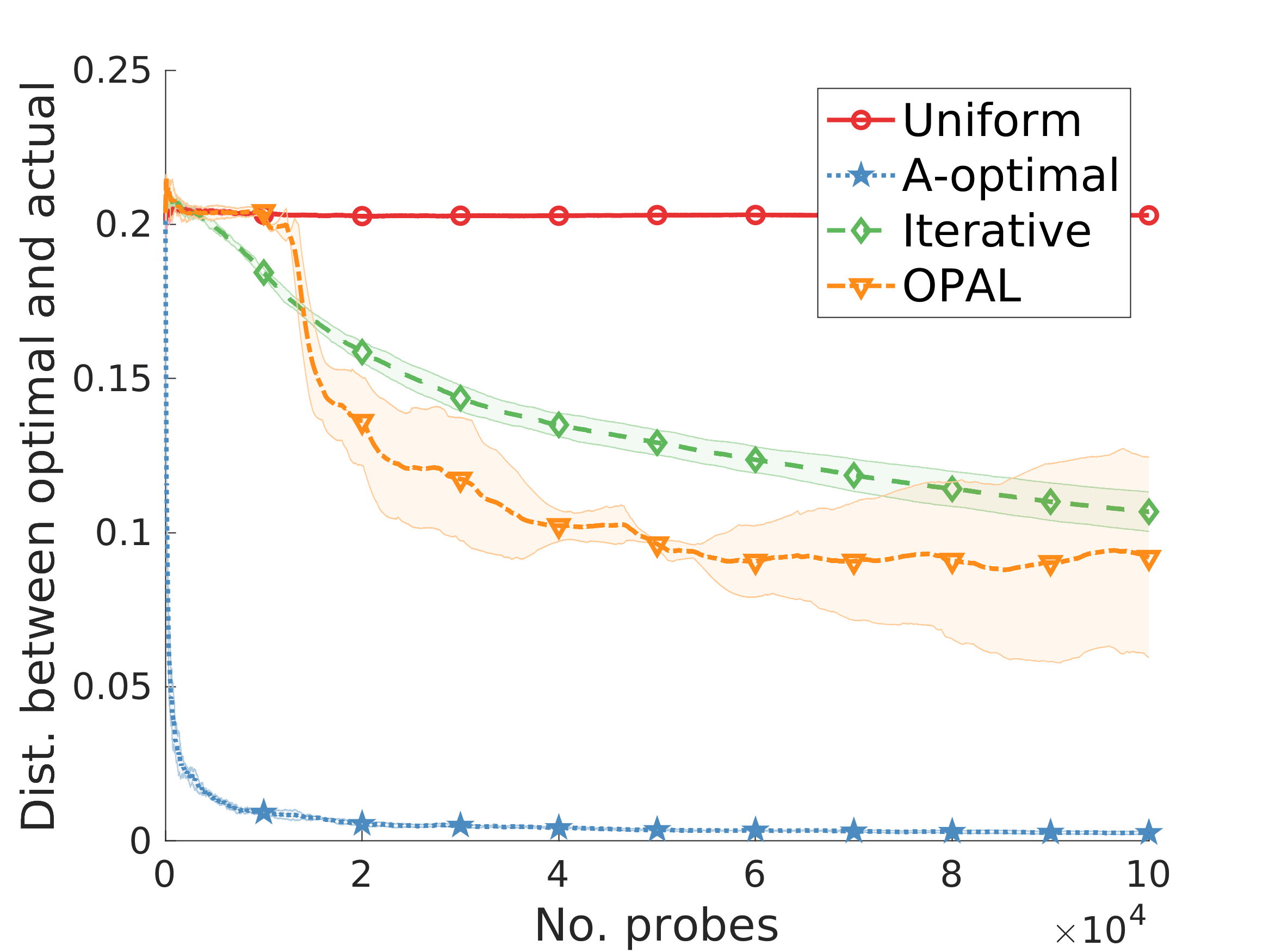}
        \caption{Dist. of actual: \(\norm*{\bm{\phi}^* - {\bm\phi}_t}_2\)}
    \end{subfigure}%
    \begin{subfigure}{.25\textwidth}
        \centering
        \includegraphics[width=\linewidth]{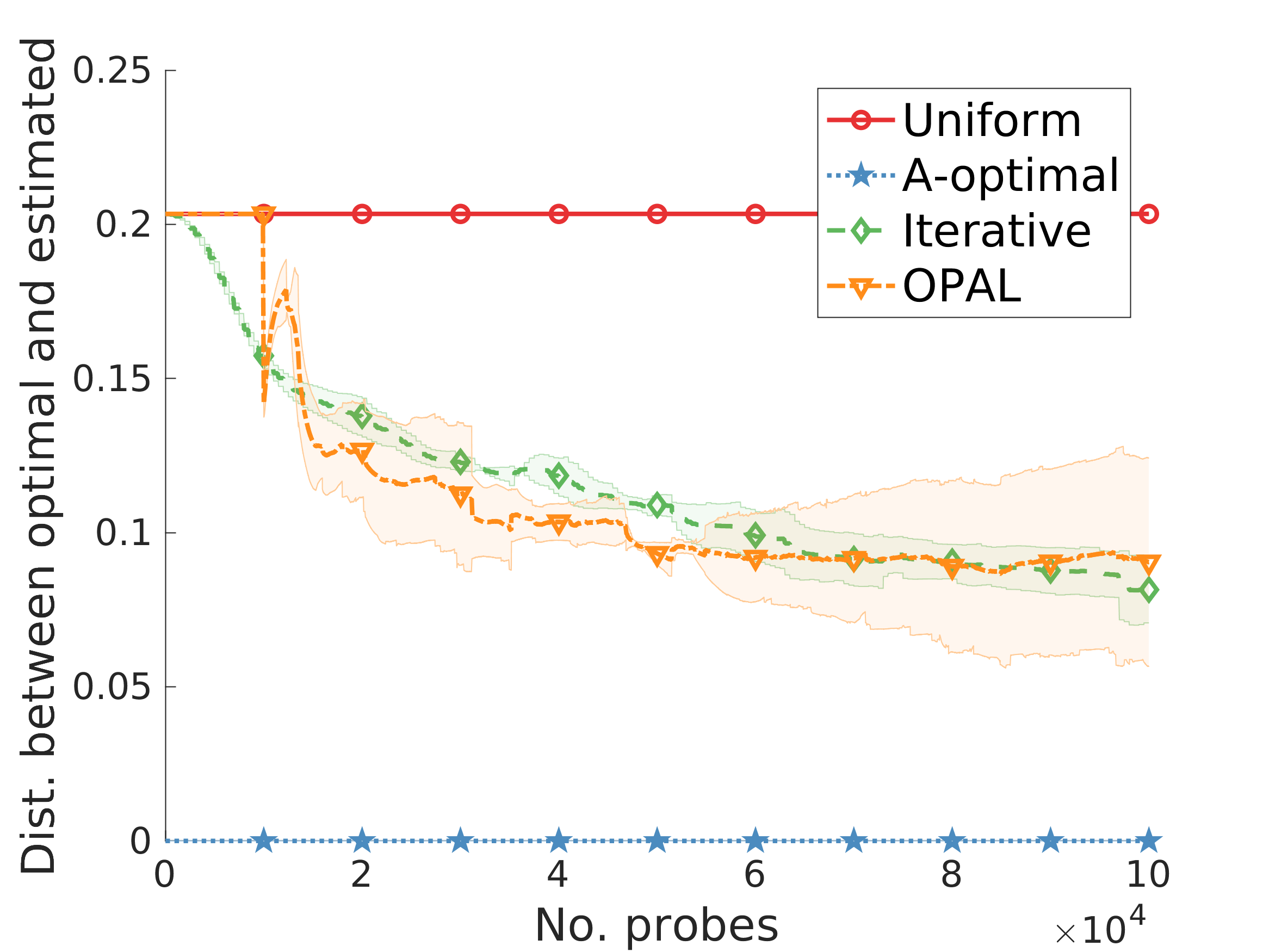}
        \caption{Dist. of estimated: \(\norm*{\bm{\phi}^* - \hat{\bm\phi}_t}_2\)}
    \end{subfigure}%
    \begin{subfigure}{.25\textwidth}
        \centering
        \includegraphics[width=\linewidth]{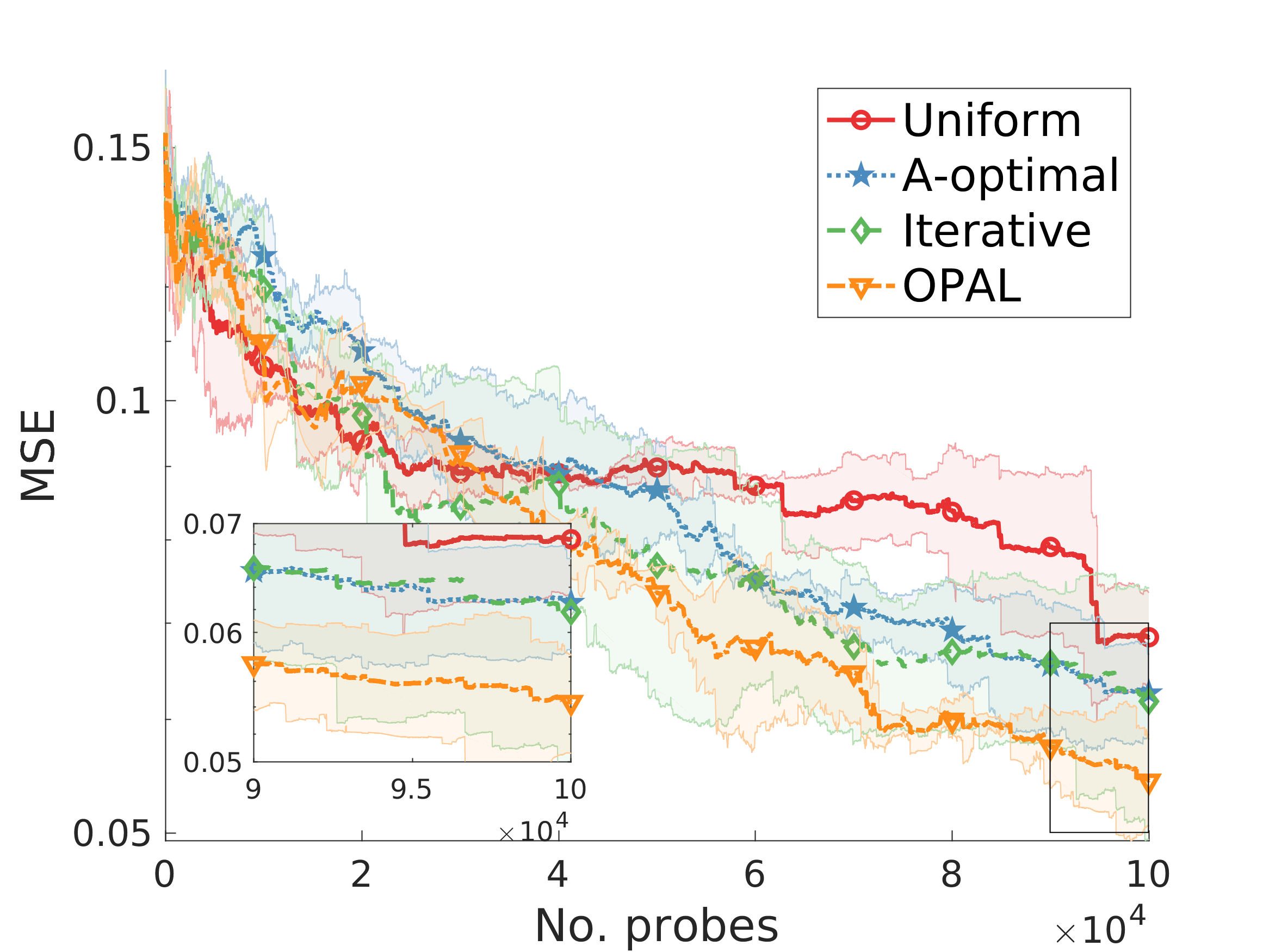}
        \caption{MSE: \(\mathbb{E}[\norm{\hat{\bm\mu}_t - \bm\mu}_2^2]\)}\label{subfig:roofnet-mse}
    \end{subfigure}%
    \caption{Comparison of algorithms for Roofnet (illustrated in Figure~\ref{subfig:roofnet-illustration})
    }
    \label{fig:classical-Roofnet}
\end{figure*}

We conducted experiments on three network topologies as illustrated in Figure~\ref{fig:topology-illustration}:
(1) a \(40\)-node star network, as shown in Figure~\ref{subfig:star-illustration};
(2) a random Erdős-Rényi (ER) network with \(20\) nodes and approximately \(35\) links (the specific number varies across scenarios), illustrated in Figure~\ref{subfig:ER-illustration}; and
(3) a topology based on Roofnet dataset~\citep{aguayo2004link} with \(37\) nodes and \(114\) links, as shown in Figure~\ref{subfig:roofnet-illustration}.
The initial sampling phase of \opal is set to \(T_0 = 0.1 T\) for both ER and Roofnet and \(T_0 = 0.35 T\) for the star network, where \(T\) is the total number of probes.
In Appendix~\ref{sec:lazy-simulation}, we conduct experiments to investigate the impact of lazy update batch \(B\) on \opal-lazy (Algorithm~\ref{alg:lazy-chasing}).

For each of these three experiments, we present four subfigures:
(a) regret,
(b) the distance between the \emph{actual} allocation and the optimal allocation,
(c) the distance between the \emph{estimated} optimal allocation and the true optimal allocation, and
(d) the mean squared error (MSE).
\emph{Smaller values for these four metrics indicate better algorithm performance.}
Regret (a) is the primary objective that \opal seeks to minimize.
MSE (d) is the practical objective of network tomography.


Figures~\ref{fig:classical-star-3},~\ref{fig:classical-ER} and~\ref{fig:classical-Roofnet} highlight the following key findings:
(I) \opal consistently outperforms both the iterative algorithm and uniform allocation across all metrics.
Additionally, \opal often matches the performance of the A-optimal allocation (offline optimal).
(II) The trends in regret (\(F(\bm\mu, \bm\phi_t) - F(\bm\mu, \bm\phi^*)\)) in subfigures (a) closely follow the trends in the distance between the actual allocation and the optimal allocation \(\|\bm{\phi}_t - \bm\phi^*\|_2\) in subfigures (b). This aligns with the identified Lipschitz continuity of the A-optimal design criterion (Condition~\ref{cond:lipschitz}).
(III) Comparing the distance between the estimated allocation and the optimal allocation \(\|\bm\phi^* - \hat{\bm\phi}_t\|_2\) in subfigures (c) with the distance of the actual allocation in subfigures (b), we observe that while both \opal and the iterative algorithm produce good estimated allocations \(\hat{\bm\phi}_t\) after a sufficient number of probes in subfigures (c), only \opal achieves a good actual allocation \(\bm\phi_t\) close to the optimal in subfigures (b). This underscores the importance of \opal's online nature, particularly the chase optimal bound strategy, in achieving a favorable actual allocation and minimizing regret.

Figures~\ref{subfig:star-mse},~\ref{subfig:er-mse} and~\ref{subfig:roofnet-mse} presents the MSE of the estimated parameters, which represents algorithms' practical performance.
We observe that (IV) the MSE of \opal is consistently lower than that of the two baselines, and it closely approaches the ideal A-optimal allocation (sometimes even surpasses it, likely due to the biased nature of the MLE estimator used in loss tomography).
Comparing with the state-of-the-art baseline (Iterative), \opal achieves \textbf{\(2.60\%\), \(7.36\%\), \(13.64\%\)} improvements in MSE for the star, ER, and Roofnet networks, respectively.

\section{Case Study for Quantum Network Tomography}
\label{sec:case-study-quantum}

This section presents a case study on the bit-flip quantum channel tomography in the quantum star network. Section~\ref{subsec:quantum-multicast-case} verifies the conditions outlined in Section~\ref{subsec:network-tomography-conditions} and presents the specific regret of \opal.
Subsequently, in Section~\ref{subsec:empirical-quantum}, we provide empirical results for the quantum scenario, validating the superiority of \opal over existing methods.

\subsection{Analytical Case 2: Quantum Bit-Flip Star Network Multicast Setting}\label{subsec:quantum-multicast-case}

We consider an \(L\)-link star network for A-optimal experimental design, where each link \(\ell \in \mathcal{L}\) is a \emph{quantum} bit-flip channel with an unknown bit-flip probability \(\mu_\ell\). As no unicast protocol in quantum network tomography has been proposed for this setting, we consider the root-independent (RI) multicast protocol devised by \citet{de2023characterization}.
The probing experiment set \(\mathcal{M}\) consists of \(M = L\) multicast probes, each selecting one link \(\ell \in \mathcal{L}\) as the root, multicasting quantum states to the remaining \(L - 1\) links, and measuring the received qubits at these links.
Denote \(A_{m,\ell}\) as the total non-flip counts from the measurements of probe \(m\) for each link \(\ell \neq m\).
Further details regarding the bit-flip channel and the RI multicast protocol are provided in Section~\ref{subsec:quantum-tomography-model} and Figure~\ref{fig:root-independent}.

\textbf{Lipschitz Continuity of the A-Optimal Design.}
We first derive the trace of the inverse of the Fisher information matrix for the quantum star network with bit-flip quantum channels as follows,
\(
F(\bm\mu; \bm\phi) = \tr \bm I^{-1} (\bm\mu; \bm\phi) = \sum_{m \in \mathcal{M}} \frac{(1 - \mu_m) \mu_m}{1 - \phi_m}.
\)
This expression verifies the Lipschitz continuity of the A-optimal design criterion (first part of Condition~\ref{cond:lipschitz}).
Next, by applying the method of Lagrange multipliers, we derive the A-optimal solution \(\bm\phi^*\) for the quantum star network with bit-flip quantum channels. The optimal allocation depends on the link parameters, as presented in~\eqref{eq:a-optimal-quantum-star} of Algorithm~\ref{alg:a-optimal-bit-flip}. Specifically, one first sorts in descending order the links by the value of \(\sqrt{(1 - \mu_\ell) \mu_\ell}\), and then iteratively eliminates the link with the largest value of \(\sqrt{(1 - \mu_\ell) \mu_\ell}\) until the condition in Line~\ref{line:eliminate-condition} is not satisfied. The optimal allocation is then given by~\eqref{eq:a-optimal-quantum-star} for the remaining links in the candidate set \(\mathcal{M}_{\text{opt}}\).

\begin{algorithm}[tp]
    \caption{A-Optimal Allocation for Bit-Flip Tomograph in Quantum Star Networks}\label{alg:a-optimal-bit-flip}
    \begin{algorithmic}[1]
        \Input link parameters \(\bm\mu\) and candidate probe set \(\mathcal{M}_{\text{opt}} \gets \{1,2,3,\dots, M\}\) (\(M=L\))
        \While{\(\sum_{\ell \in \mathcal{M}_{\text{opt}}} \sqrt{(1 - \mu_\ell) \mu_\ell} -  (\abs{\mathcal{M}_{\text{opt}}} - 1) \max_{\ell' \in \mathcal{M}_{\text{opt}}} \sqrt{(1 - \mu_{\ell'}) \mu_{\ell'}} < 0\)} \label{line:eliminate-condition}
        \State \(\mathcal{M}_{\text{opt}} \gets \mathcal{M}_{\text{opt}} \setminus \{\argmax_{\ell' \in \mathcal{M}_{\text{opt}}} \sqrt{(1 - \mu_{\ell'}) \mu_{\ell'}}\}\)
        \EndWhile
        \Output \(\phi_\ell^* =\)
        \begin{equation}\label{eq:a-optimal-quantum-star}
            \begin{cases}
                \!\!\frac{\sum_{\ell' \in \mathcal{M}_{\text{opt}}} \sqrt{(1 {-} \mu_{\ell'}) \mu_{\ell'}} - (\abs{\mathcal{M}_{\text{opt}}} {-} 1)\sqrt{(1 {-} \mu_\ell) \mu_\ell}}{\sum_{\ell' \in \mathcal{M}_{\text{opt}}} \sqrt{(1 {-} \mu_{\ell'}) \mu_{\ell'}}}
                 & \!\! \!\! \text{for } \ell {\in} \mathcal{M}_{\text{opt}}
                \\
                0
                 & \!\! \!\! \text{for } \ell {\notin} \mathcal{M}_{\text{opt}}
            \end{cases}
        \end{equation}
    \end{algorithmic}
\end{algorithm}

While the expression for the optimal allocation solution \(\bm\phi^*\) in~\eqref{eq:a-optimal-quantum-star} is piecewise, these pieces are continuous at all critical points with respect to the link parameters. Because altering the inequality in Line~\ref{line:eliminate-condition} from ``less than'' to ``less than or equal to'' does not change the solution.
Hence, the Lipschitz continuity of the A-optimal design allocation (second part of Condition~\ref{cond:lipschitz}) is verified.

\textbf{MLE Estimator for Link Parameters.}
The observations at each link \(\ell\) from the RI multicast follow a Bernoulli distribution with success probability \(\mu_\ell\)~\citep[Table II and Eq. (17)]{de2023characterization}. Hence, the maximum likelihood estimator (MLE) for the link parameter \(\mu_\ell\) is calculated by the observations from all other RI probes \(m \in \mathcal{M} \setminus \{\ell\}\) as follows:
\begin{equation}\label{eq:mle-quantum-star}
    \hat\mu_\ell = \frac{\sum_{m \in \mathcal{M} \setminus \{\ell\}} A_{m,\ell}}{\sum_{m \in \mathcal{M} \setminus \{\ell\}} S_{m}}, \quad \text{for all links } \ell \in \mathcal{L},
\end{equation}
where we recall that \(S_{m}\) is the total number of qubits sent by probe \(m\), and \(A_{m,\ell}\) is the total number of qubits not flipped at link \(\ell\) from the measurements of probe \(m\).

\textbf{Confidence Interval of MLE Estimator.}
By Hoeffding's inequality, the confidence interval for the MLE estimator in~\eqref{eq:mle-quantum-star} is given by,
for some constant \(C_2 > 0\)
\begin{equation}
    \label{eq:confidence-interval-quantum-star}
    \mu_\ell \in \left( \hat\mu_\ell - C_2\sqrt{\frac{\log \delta^{-1}}{\sum_{m \in \mathcal{M} \setminus \{\ell\}} S_{m}}}, \hat\mu_\ell + C_2\sqrt{\frac{\log \delta^{-1}}{\sum_{m \in \mathcal{M} \setminus \{\ell\}} S_{m}}} \right).
\end{equation}
with confidence level \(1 - \delta\), for all links  \(\ell \in \mathcal{L}.\)  Therefore, Condition~\ref{cond:finite-confidence-interval} is verified with all \(\gamma_{m,\ell} = \frac{1}{2}\) and thus \(\gamma_{\min} = \frac{1}{2}.\)

\begin{figure*}[tp]
    \centering
    \begin{subfigure}{.25\textwidth}
        \centering
        \includegraphics[width=\linewidth]{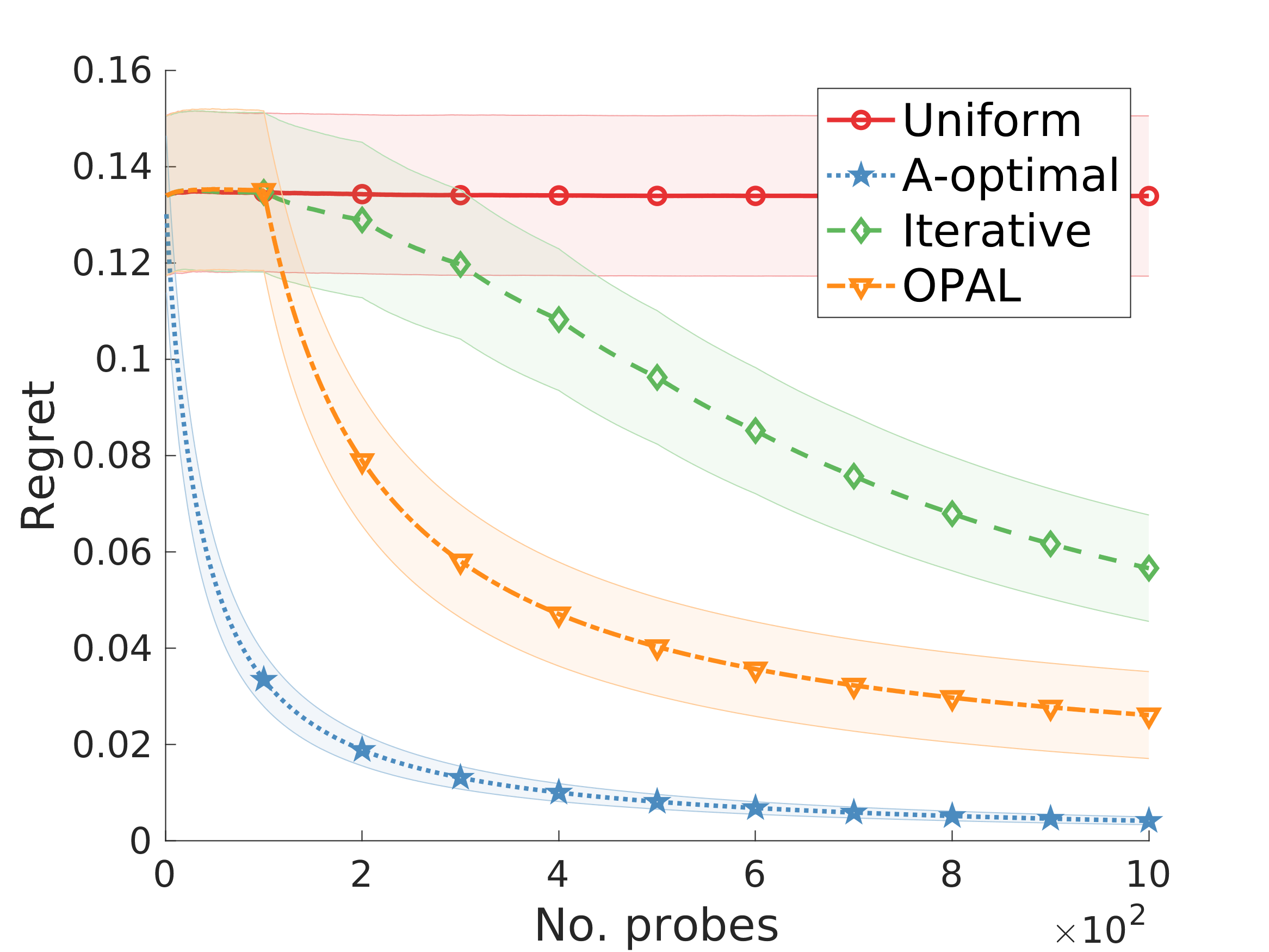}
        \caption{Regret: $R_t{=}F\!(\bm\mu,\!\bm\phi_t\!) {-} F\!(\bm\mu,\!\bm\phi^*\!)$}
    \end{subfigure}%
    \begin{subfigure}{.25\textwidth}
        \centering
        \includegraphics[width=\linewidth]{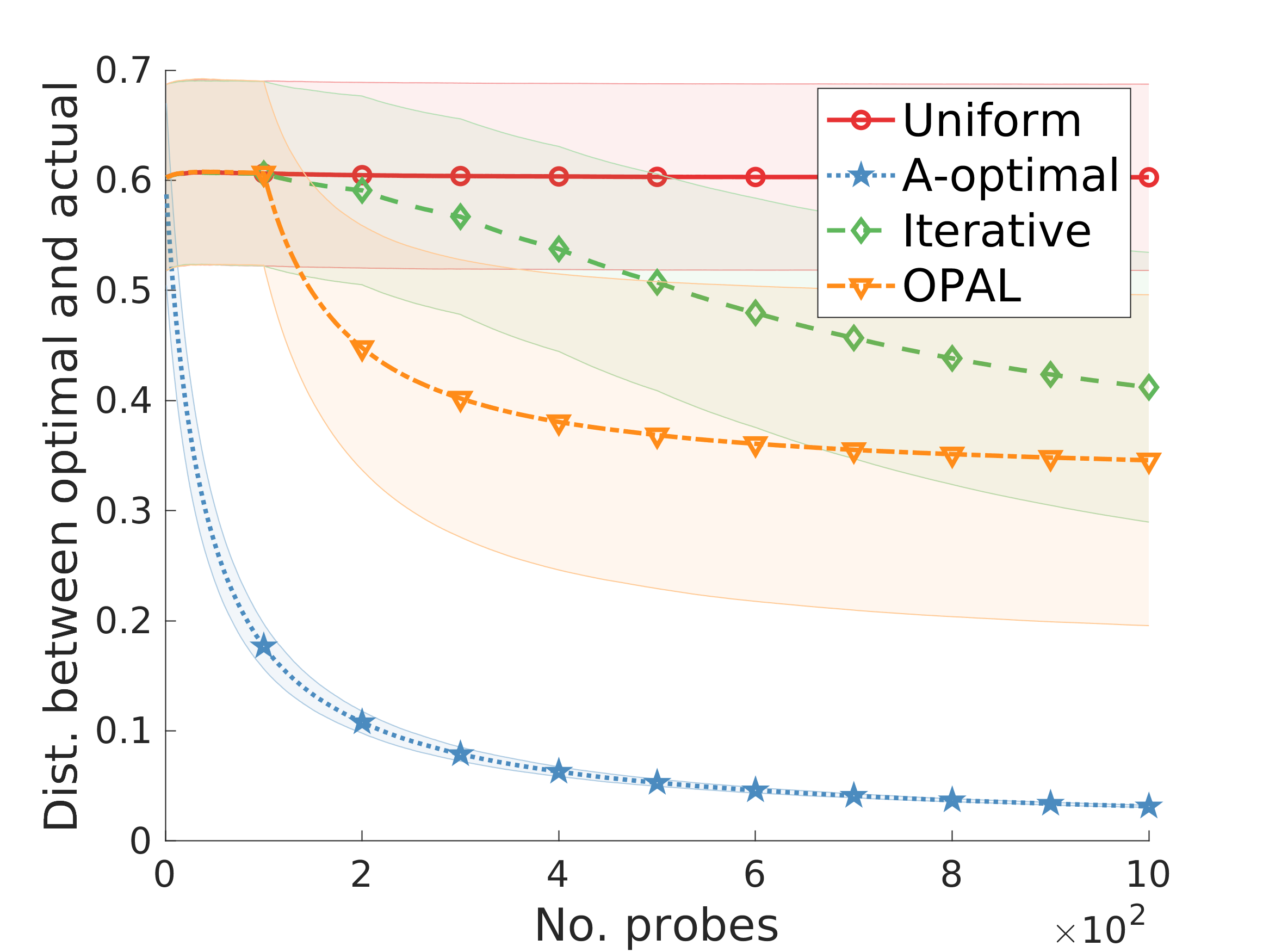}
        \caption{Dist. of actual: \(\norm*{\bm{\phi}^* - {\bm\phi}_t}_2\)}
    \end{subfigure}%
    \begin{subfigure}{.25\textwidth}
        \centering
        \includegraphics[width=\linewidth]{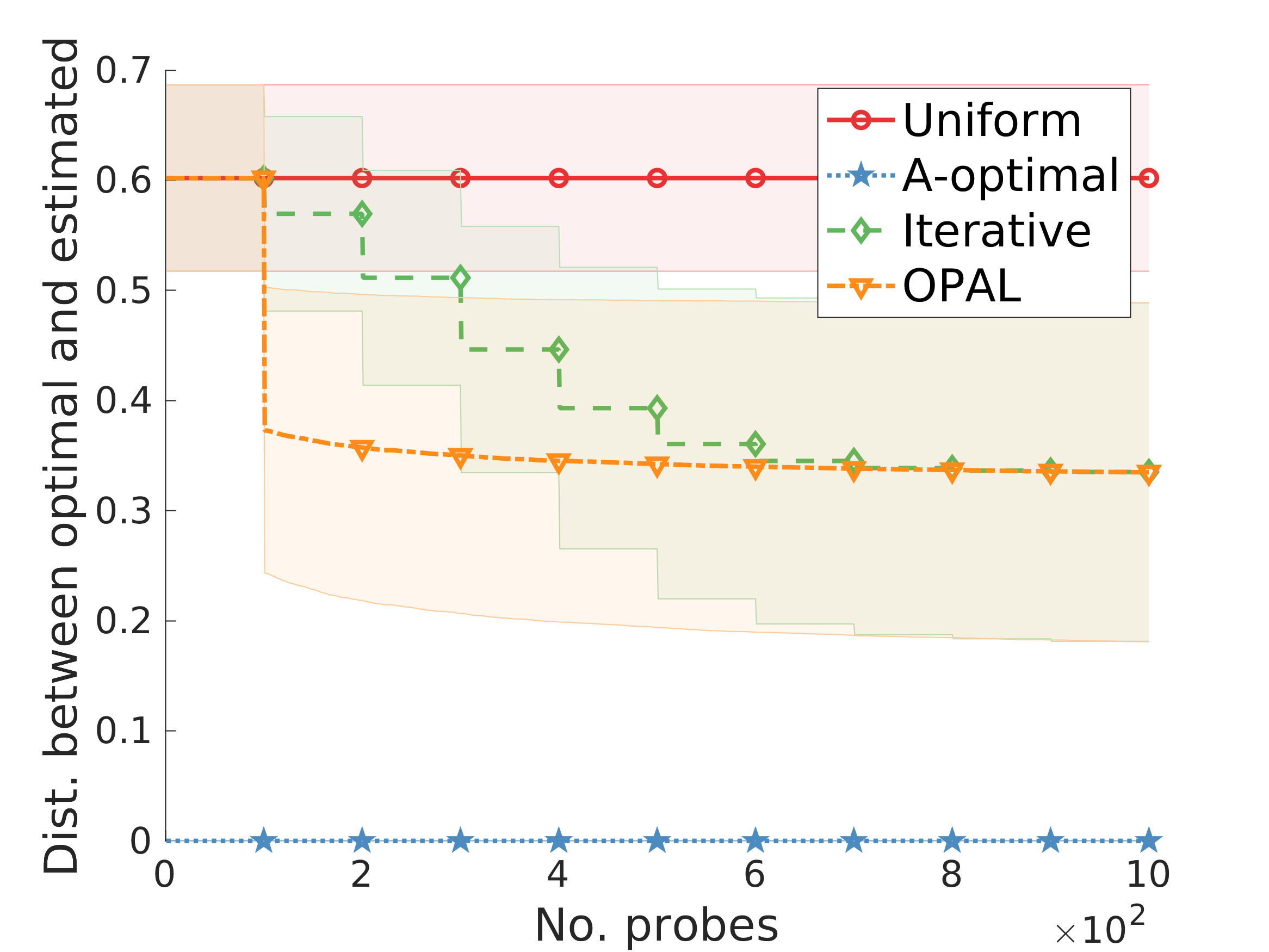}
        \caption{Dist. of estimated: \(\norm*{\bm{\phi}^* - \hat{\bm\phi}_t}_2\)}
    \end{subfigure}%
    \begin{subfigure}{.25\textwidth}
        \centering
        \includegraphics[width=\linewidth]{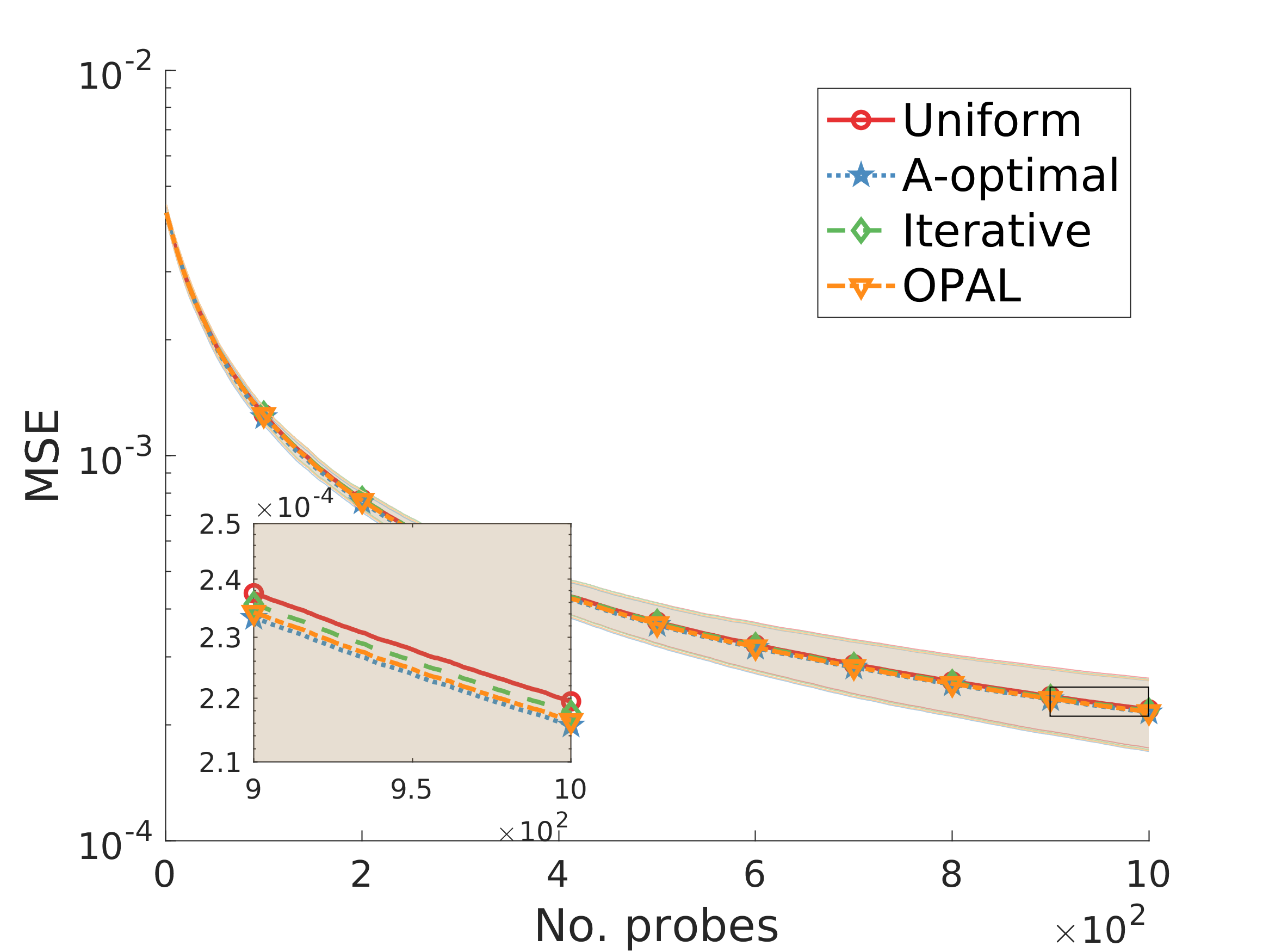}
        \caption{MSE: \(\mathbb{E}[\norm{\hat{\bm\mu}_t - \bm\mu}_2^2]\)}\label{subfig:quantum-star-mse}
    \end{subfigure}%
    \caption{Comparison of algorithms for $40$-node quantum star networks}
    \label{fig:quantum-star-3}
\end{figure*}

\textbf{Regret of \opal on the Bit-Flip Quantum Star Network with RI Multicast.}
With the Lipschitz continuity of the A-optimal design criterion and the finite confidence interval with \(\gamma_{\min} = \frac{1}{2}\), we apply Theorem~\ref{thm:convergence-rate} to derive the regret of \opal on the quantum star network with bit-flip quantum channels as follows,
\(
R_T = \tr \bm I^{-1}(\bm\mu; {\bm\phi}_T) - \tr\bm I^{-1}(\bm\mu; \bm\phi^*) \le C \left( \frac{\log T}{\xi T} \right)^{\frac{1}{2}} + \alpha\xi.
\)
Unlike the regret derivation in Section~\ref{subsec:classical-unicast-case}, we cannot cancel the second \(\alpha\xi\) term by setting a constant input \(\xi\) in this quantum network tomography task. Because the A-optimal allocation in~\eqref{eq:a-optimal-quantum-star} inherently includes probes with a zero allocation ratio, i.e., \(\phi_m^* = 0\) for some \(m\), the second term remains non-zero.
This aligns with the intuition that  it is difficult for an algorithm to avoid performing these ineffective probes without knowledge of the exact link parameters a priori, leading to a regret (convergence rate) slower than \(\tilde O(T^{-\gamma_{\min}})\).

By selecting \(\xi = T^{-\frac{1}{3}}\), we obtain the worst-case convergence (regret) guarantee of \opal on the quantum bit-flip star network under the RI multicast, expressed as,
\[
    R_T = \tr \bm I^{-1}(\bm\mu; {\bm\phi}_T) - \tr\bm I^{-1}(\bm\mu; \bm\phi^*) \le O\left(\left( \frac{\log T}{T} \right)^{\frac{1}{3}}\right).
\]


\textbf{Complexity of \opal for Quantum Bit-Flip Star Network.}
The major computational cost of applying \opal to the bit-flip star networks is from executing Algorithm~\ref{alg:a-optimal-bit-flip} (Line~\ref{line:greedy-a-optimal} in \opal), whose time complexity is \(O(L)\).
This improvement from \(O(L^3)\) in classical networks to \(O(L)\) in quantum networks is due to the efficient root-independent multicast protocol~\citep{de2023characterization} designed explicitly for quantum networks.


\subsection{Simulations for Quantum Network Tomography}\label{subsec:empirical-quantum}

This section reports the empirical performance of \opal algorithm in a \(40\)-node quantum bit-flip star network tomography task:
one center node connects to \(39\) other root nodes via bit-flip quantum channels (links).
The set of probing experiments consists of all \(39\) root-independent (RI) multicast probes, each selecting the end node of one link as the root node and the other \(38\) end nodes as the multicast receivers.
We consider the same three baselines as in Section~\ref{subsec:empirical-classical}.
The initial sampling phase of \opal is set to \(T_0 = 0.1T\) rounds.
The experiment is conducted with \(20\) scenarios, each with a set of randomly generated bit-flip probabilities \(\bm\mu\) from a uniform distribution over \((0.1, 1)\) for the \(40\)-node quantum star network.
For each scenario, we run \(100\) Monte Carlo simulations and take the average as the performance of this scenario.
We report the average performance of these \(20\) scenarios as solid lines, and the shaded areas represent their standard deviation across these scenarios.

The RI protocol considered in the experiment is specifically developed for quantum bit-flip star network,
and it is already the most advanced in the QNT literature~\citep{de2023characterization,de2024quantum}.
Therefore, other general QNT experiments under other more complex quantum channels and network topologies---e.g., the ER and Roofnet topologies in the classical setting at Section~\ref{subsec:empirical-classical}---are invalid here.
This quantum star network experiment serves as a proof-of-concept for the applicability of \opal to QNT.
Once more advanced QNT protocols are developed for more general network topologies,
one can apply \opal to optimizing the probe allocation over them in more general scenarios.

Figure~\ref{fig:quantum-star-3} presents the six performance metrics of algorithms for bit-flip tomography in a \(40\)-node quantum star network. The results demonstrate the superior performance of \opal in the quantum network tomography scenario and support the same observations as stated in Section~\ref{subsec:empirical-classical}.
However, the improvements of \opal over the baselines are not as significant as in the classical case (e.g., Figures~\ref{subfig:quantum-star-mse}).
This is because the optimal allocation of the RI multicast protocol \(\bm\phi^*\) in Algorithm~\ref{alg:a-optimal-bit-flip} is easier to estimate than those in classical settings, and hence even for the simplistic uniform allocation (the naive baseline), the MSE of the estimated link parameters is already small and therefore, there is less room to illustrate the improvement.
Applying \opal to QNT protocols for more complex quantum network topologies is expected to yield more significant improvements, which is infeasible to verify in this paper due to the lack of existing QNT protocols for general topologies.

\section{Conclusion}


This paper introduces a novel and general online experimental design algorithm for network tomography, termed online probe allocation (\opal). Theoretically, \opal is the first algorithm to offer rigorous regret guarantees for network tomography. We establish these guarantees by identifying two critical conditions: Lipschitz continuity and confidence interval concentration.
On the practical side, we validate these theoretical conditions in classical loss unicast networks and quantum bit-flip multicast networks, representing key use cases for network tomography. Empirically, we illustrate the superior performance of \opal compared to existing methods in classical and quantum network tomography scenarios.




%

\bibliographystyle{IEEEtranN}
\bibliography{reference.bib}

\newpage

\section*{Biography}



\begin{IEEEbiographynophoto}{Xuchuang Wang} (\emph{Member, IEEE}) received the BEng degree from the School of Electronic and Information Engineering, Xi'an Jiaotong University, in 2019, and the Ph.D. degree from the Department of Computer Science and Engineering, Chinese University of Hong Kong (CUHK), in 2023. He is currently a postdoctoral researcher at the Manning College of Information and Computer Sciences, University of Massachusetts, Amherst. His research interests include performance evaluation and optimization of multi-agent learning systems and quantum networks.
\end{IEEEbiographynophoto}

\begin{IEEEbiographynophoto}{Yu-Zhen Janice Chen}
  received the B.Sc. degree in computer science from the
  Chinese University of Hong Kong in 2019. She is currently pursuing the Ph.D.
  degree with the University of Massachusetts, Amherst. Her research interests
  include statistical machine learning, sequential decision-making, performance
  analysis, modeling, and algorithm design for computing systems.
\end{IEEEbiographynophoto}

\begin{IEEEbiographynophoto}{Matheus Guedes de Andrade} received the B.A. degree in computer and information engineering from the Federal University of Rio de Janeiro (UFRJ),
  in 2018, and the M.Sc. degree in computer and systems engineering from COPPE UFRJ, in 2020. He is currently pursuing the
  Ph.D. degree in CS with the Manning College of Information
  and Computer Science, UMass Amherst, under the supervision
  of Prof. Don Towsley. He worked as the Student Leadership
  Council Chair at the Center for Quantum Networks (CQN),
  from January 2021 to January 2023. His main research interests include distributed quantum computing, quantum network
  tomography, and performance evaluation of quantum networks.
\end{IEEEbiographynophoto}

\begin{IEEEbiographynophoto}{Mohammad Hajiesmaili} (\emph{Member, IEEE}) is an associate professor
  of computer science at the University of Massachusetts
  Amherst. His research interests
  include optimization, machine learning, and algorithms.
  Before joining UMass,
  Mohammad was a postdoctoral fellow with the
  Johns Hopkins University, from 2017 to 2018,
  and with the Chinese University of Hong Kong,
  from 2015 to 2016. He received his Ph.D. and
  M.Sc. degrees from the University of Tehran,
  and his B.Sc. degree from Sharif University of
  Technology.
\end{IEEEbiographynophoto}

\begin{IEEEbiographynophoto}{John C.S. Lui}
  (\emph{Fellow, IEEE}) received the Ph.D. degree in computer science from UCLA. He is currently the Choh-Ming Li chair professor with the Department of Computer Science and Engineering (CSE), The Chinese University of Hong Kong (CUHK). After his graduation, he joined the IBM Laboratory and participated in research and development projects on file systems and parallel I/O architectures. He later joined the CSE Department with CUHK. His current research interests are in quantum networks, online learning algorithms and applications (e.g., multi-armed bandits, reinforcement learning), machine learning on network sciences and networking systems, large scale data analytics, network/system security, network economics, large scale storage systems, and performance evaluation theory. He is an elected member of the IFIP WG 7.3, fellow of ACM, senior research fellow of the Croucher Foundation, fellow of the Hong Kong Academy of Engineering Sciences (HKAES).
\end{IEEEbiographynophoto}

\begin{IEEEbiographynophoto}{Ting He}
  (\emph{Senior Member, IEEE}) eceived the Ph.D. degree in ECE
  from Cornell University. Dr. He is an Associate
  Professor in the School of EECS at the Pennsylvania
  State University, University Park, PA. Her interests
  span computer networking, performance evaluation,
  and machine learning. Dr. He has served as Associate Editor for IEEE Transactions on Communications and IEEE/ACM Transactions on Networking,
  General Co-Chair of IEEE RTCSA, TPC Co-Chair
  of ACM MobiHoc and IEEE ICCCN, and Area TPC
  Chair of IEEE INFOCOM. She received multiple
  paper awards from IEEE Communications Society, ICDCS, SIGMETRICS,
  ICASSP, IMC, and SmartGridComm.
\end{IEEEbiographynophoto}

\begin{IEEEbiographynophoto}{Don Towsley}
  (\emph{Life Fellow, IEEE}) received the B.A. degree in physics
  and the Ph.D. degree in computer science from the University of Texas, in 1971 and 1975, respectively. He is currently
  a Distinguished Professor with the Manning College of Information and Computer Sciences, University of Massachusetts.
  He has held Visiting Positions at numerous universities and
  research labs. His research interests include quantum communications and networking and distributed quantum sensing and
  computing. He was a Co-Founder and a Co-Editor-in-Chief of
  the ACM Transactions on Modeling and Performance Evaluation of Computing Systems (TOMPECS), and has served as the
  Editor-in-Chief for IEEE/ACM Transactions on Networking (ToN) and on numerous editorial boards. He has served as the Program Co-Chair of several conferences including INFOCOM
  2009. He is a Corresponding Member of the Brazilian Academy of Sciences and has received several achievement awards
  including the 2007 IEEE Koji Kobayashi Award, the 2007 ACM
  SIGMETRICS, the 2008 ACM SIGCOMM Awards, and the
  2023 Network Science Society Euler Award. He has received
  numerous paper awards including the 2012 ACM SIGMETRICS
  Test-of-Time Award, the 2008 SIGCOMM Test-of-Time Paper
  Award, and the 2018 SIGMOBILE Test-of-time Award. He also
  received the 1998 IEEE Communications Society William Bennett Best Paper Award. Last, he has been elected as a Fellow
  of ACM.
\end{IEEEbiographynophoto}

\vfill

\appendix


\section{Detailed Proof of Theorem~\ref{thm:convergence-rate}}\label{sec:proof-convergence-rate}
After the initial sampling phase, we have \(S_{m,t} = \xi T\) and \(t=T_0 = M\xi T\). In addition, the gap between the estimated optimal allocation \(\hat{\bm\phi}^*_{T_0}\) and the actual optimal allocation \(\bm\phi^*\) can be characterized as follows.
\begin{equation}\label{eq:bound-init-allocation}
    \begin{split}
        \norm{\hat{\bm\phi}^*_{T_0} - {\bm\phi}^*}_\infty
         & \overset{(\text{a})}\le \beta \norm{\hat{\bm\mu}_{T_0} - \bm\mu}_\infty
        \\
         & \overset{(\text{b})}\le \beta\max_{\ell\in\mathcal{L}} \sum_{m\in\mathcal{M}} c_{\ell,m} \left(\frac{\log \delta}{\xi T}\right)^{\gamma_{\ell,m}}
        \\
         & \overset{(\text{c})}\le \beta c_{\max} \left(\frac{\log \delta}{\xi T}\right)^{\gamma_{\min}},
    \end{split}
\end{equation}
where inequality (a) is due to the second Lipschitz continuity in Condition~\ref{cond:lipschitz},
inequality (b) follows from the confidence radius of estimates in Condition~\ref{cond:finite-confidence-interval} with a probability of at least \(1-L\delta\),
and inequality (c) follows from the definition of
\(c_{\max} = \max_{\ell\in\mathcal{L}} \sum_{m\in\mathcal{M}} c_{\ell,m}\) and \(\gamma_{\min} = \min_{(\ell, m)\in\mathcal{L}\times\mathcal{M}:\gamma_{\ell,m}>0} \gamma_{\ell,m}\).



Next, we bound the regret --- the difference between the OED criterion \(F\) of the actual allocation \(\bm\phi_T\) and the optimal allocation \(\bm\phi^*\) --- as follows.
\[
    \begin{split}
         & \quad\, F(\bm\mu; \bm\phi_T) - F(\bm\mu; \bm\phi^*)
        \overset{\text{(d)}}\le \alpha \norm{\bm\phi_T - {\bm\phi}^*}_\infty
        \\
         &
        \overset{\text{(e)}}\le \alpha \norm{\bm\phi_T - \hat{\bm\phi}_T^*}_\infty + \alpha \norm{\hat{\bm\phi}_T^* - \bm\phi^*}_\infty
        \\
         & \overset{\text{(f)}}\le \alpha \norm{\bm\phi_T - \hat{\bm\phi}_T^*}_\infty + \alpha\beta \norm{\hat{\bm\mu}_T - \bm\mu}_\infty
        \\
         &
        \overset{\text{(g)}}\le \alpha \norm{\bm\phi_T - \hat{\bm\phi}_T^*}_\infty + \alpha\beta\max_{\ell\in\mathcal{L}} \sum_{m\in\mathcal{M}} c_{\ell,m} (2\log T / S_{m,T})^{\gamma_{\ell,m}}
        \\
         & \overset{\text{(h)}}\le \alpha \norm{\bm\phi_T - \hat{\bm\phi}_T^*}_\infty + \alpha\beta c_{\max} \left(\frac{\log \delta}{\xi T}\right)^{\gamma_{\min}}
        \\
         & \overset{\text{(i)}}\le \alpha \left( \frac{1}{(1-\xi) T} + 3\beta c_{\max} \left(\frac{\log \delta}{\xi T}\right)^{\gamma_{\min}} \right) + \alpha\beta c_{\max} \left(\frac{\log \delta}{\xi T}\right)^{\gamma_{\min}}
        \\
         & = 4\alpha\beta c_{\max}\left(\frac{\log \delta}{\xi T}\right)^{\gamma_{\min}}
        + \frac{\alpha}{(1-\xi) T},
    \end{split}
\]
where inequality (d) follows from the Lipschitz continuity of function \(F\) with respect to allocations \(\bm\phi\) in Condition~\ref{cond:lipschitz},
inequality (e) follows from the triangle inequality,
inequality (f) follows from the Lipschitz continuity of the optimal allocation in function \(\bm\phi^*(\bm\mu)\) concerning the network parameters \(\bm\mu\) in Condition~\ref{cond:lipschitz},
inequality (g) follows from the confidence radius of estimates in Condition~\ref{cond:finite-confidence-interval}, holding with a probability of at least \(1-L\delta\),
inequality (h) is by telescoping the maximization term \(\max_{\ell\in\mathcal{L}}\) and the sum \(\sum_{m'\in\mathcal{M}}\)
by their upper bounds, and \(S_{m,T} \ge \xi T\) for any probing experimental \(m\) due to the initial sample phase input \(\mathcal{S}_0 = (\xi T)_{m\in\mathcal{M}}\).
Finally, inequality (i) is by~\eqref{eq:phi_star_diff} for deriving the first term in~\eqref{eq:convergence-rate}; the second term in~\eqref{eq:convergence-rate} is derived by~\eqref{eq:phi_star_diff_for_small_phi} (see below for a detailed derivation).


Last, we present the detailed telescoping derivation of inequality (i) in the above proof for bounding \(\norm{\bm\phi_T - \hat{\bm\phi}_T^*}_\infty\). The high-level derivation technique is illustrated in Fig.~\ref{fig:bound-illustration}.
We consider two cases: (1) all probes have optimal allocations greater than \(\xi\), and (2) there exists at least one probe with optimal allocation less than \(\xi\).

\emph{Case (1).} When \(\phi^*_{m} > \xi\) for all probes \(m\), we have
\begin{equation}\label{eq:phi_star_diff}
    \begin{split}
        \norm*{{\bm\phi}_{T} - \hat{\bm\phi}_{T}^*}_\infty
         & \overset{(\text{j})}\le \norm*{{\bm\phi}_{T} - \hat{\bm\phi}^*_{T_0}}_\infty + \norm*{\hat{\bm\phi}^*_{T_0} - \hat{\bm\phi}_{T}^*}_\infty
        \\
         & \overset{(\text{k})}\le  \frac{1}{(1-\xi) T} + \beta c_{\max} \left(\frac{\log \delta}{\xi T}\right)^{\gamma_{\min}} + \norm*{\hat{\bm\phi}^*_{T_0} - \hat{\bm\phi}_{T}^*}_\infty
        \\
         & \overset{(\text{l})}\le  \frac{1}{(1-\xi) T} + 3\beta c_{\max} \left(\frac{\log \delta}{\xi T}\right)^{\gamma_{\min}},
    \end{split}
\end{equation}
where inequality (j) follows from the triangle inequality,
inequality (k) is by noticing that while the chased estimated optimal allocation \(\hat{\bm\phi}_t^*\) changes over time, it always lies in the confidence interval of the initial estimated optimal allocation \(\hat{\bm\phi}^*_{T_0}\) in~\eqref{eq:bound-init-allocation},
and inequality (l) follows from applying triangle inequality to show that \(\norm*{\hat{\bm\phi}^*_{T_0} - \hat{\bm\phi}_{T}^*}_\infty \le \norm*{\hat{\bm\phi}^*_{T_0} - {\bm\phi}^*}_\infty + \norm*{{\bm\phi}^* - \hat{\bm\phi}_{T}^*}_\infty\), and both terms in RHS are less than \(\beta c_{\max} \left(\frac{\log \delta}{\xi T}\right)^{\gamma_{\min}}\) as~\eqref{eq:bound-init-allocation} shows.

\emph{Case (2).} When there exists at least one probe \(m\) such that \(\phi^*_m \le \xi < \xi + \beta c_{\max} \left(\frac{\log \delta}{\xi T}\right)^{\gamma_{\min}}\),
by modifying the derivation of inequality~\eqref{eq:phi_star_diff} and noticing that its LHS is for infinity norm,  we have \begin{equation}\label{eq:phi_star_diff_for_small_phi}
    \begin{split}
        \norm*{{\bm\phi}_{T} - \hat{\bm\phi}_{T}^*}_\infty
         & \le \max_{m: \phi_m^* > \xi} \abs{{\phi}_{m,T} - \hat{\phi}_{m,T}^*}                                                                                 \\
         & \qquad + \max_{m: \phi^*_m \le  \xi + \beta c_{\max} \left(\frac{\log \delta}{\xi T}\right)^{\gamma_{\min}}} \abs{{\phi}_{m,T} - \hat{\phi}_{m,T}^*}
        \\
         & \le \frac{1}{(1-\xi) T} + 7\beta c_{\max} \left(\frac{\log \delta}{\xi T}\right)^{\gamma_{\min}} + 2\xi,
    \end{split}
\end{equation}
where the first maximization term follows the same derivation as~\eqref{eq:phi_star_diff} in Case (1),
and the second term is by noticing that
\(\hat{\phi}_{m,t}^* \in \left( \phi_m^* - \beta c_{\max} \left(\frac{\log \delta}{\xi T}\right)^{\gamma_{\min}}, \phi_m^* + \beta c_{\max} \left(\frac{\log \delta}{\xi T}\right)^{\gamma_{\min}}  \right)\)
for all time slots \(t\) in the chasing phase with high probability, and hence the actual allocation \(\bm\phi_t\), chasing the estimated optimal allocation \(\hat{\bm\phi}_t^*\), is also less than \(\phi_m^* + \beta c_{\max} \left(\frac{\log \delta}{\xi T}\right)^{\gamma_{\min}}\).
Therefore, the second term in~\eqref{eq:phi_star_diff_for_small_phi} is at most \(2\left(  \phi_m^* + \beta c_{\max} \left(\frac{\log \delta}{\xi T}\right)^{\gamma_{\min}}\right) \le 2\xi + 4\beta c_{\max} \left(\frac{\log \delta}{\xi T}\right)^{\gamma_{\min}}\).

At the end, replacing \(\delta\) by \(1/(LT^2)\) yields the final regret bound. Over the proof procedure, the confidence interval in Condition~\ref{cond:finite-confidence-interval} is applied multiple times. Applying a union bound of all these confidence interval applications leads to a failure probability of at most \(LT\delta\), which becomes \(1/T\) after the substitution.\qed

\begin{figure*}[tb]
    \centering
    \begin{subfigure}{.25\textwidth}
        \centering
        \includegraphics[width=\linewidth]{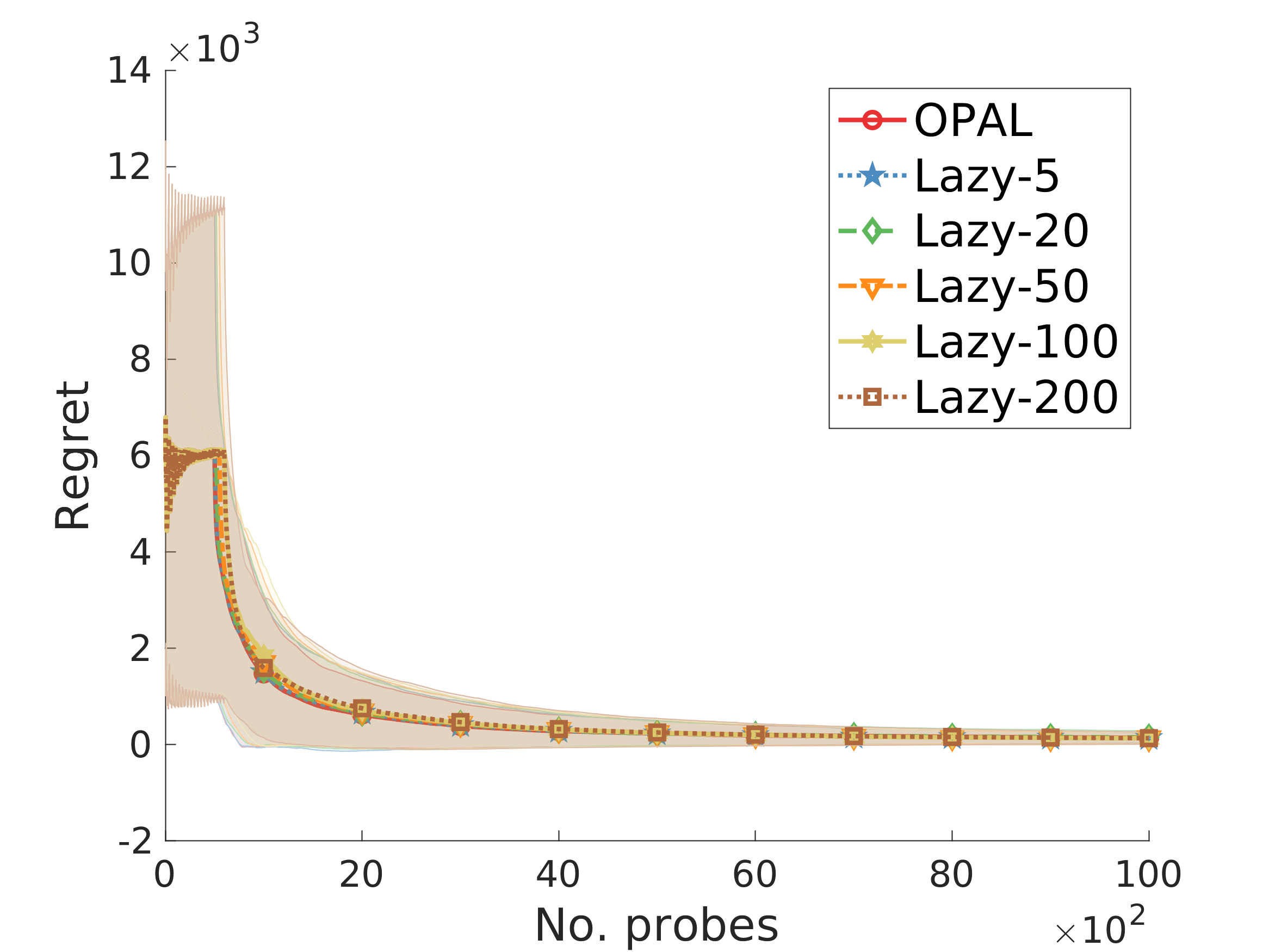}
        \caption{Regret: $R_t{=}F\!(\bm\mu,\!\bm\phi_t\!) {-} F\!(\bm\mu,\!\bm\phi^*\!)$}
    \end{subfigure}%
    \begin{subfigure}{.25\textwidth}
        \centering
        \includegraphics[width=\linewidth]{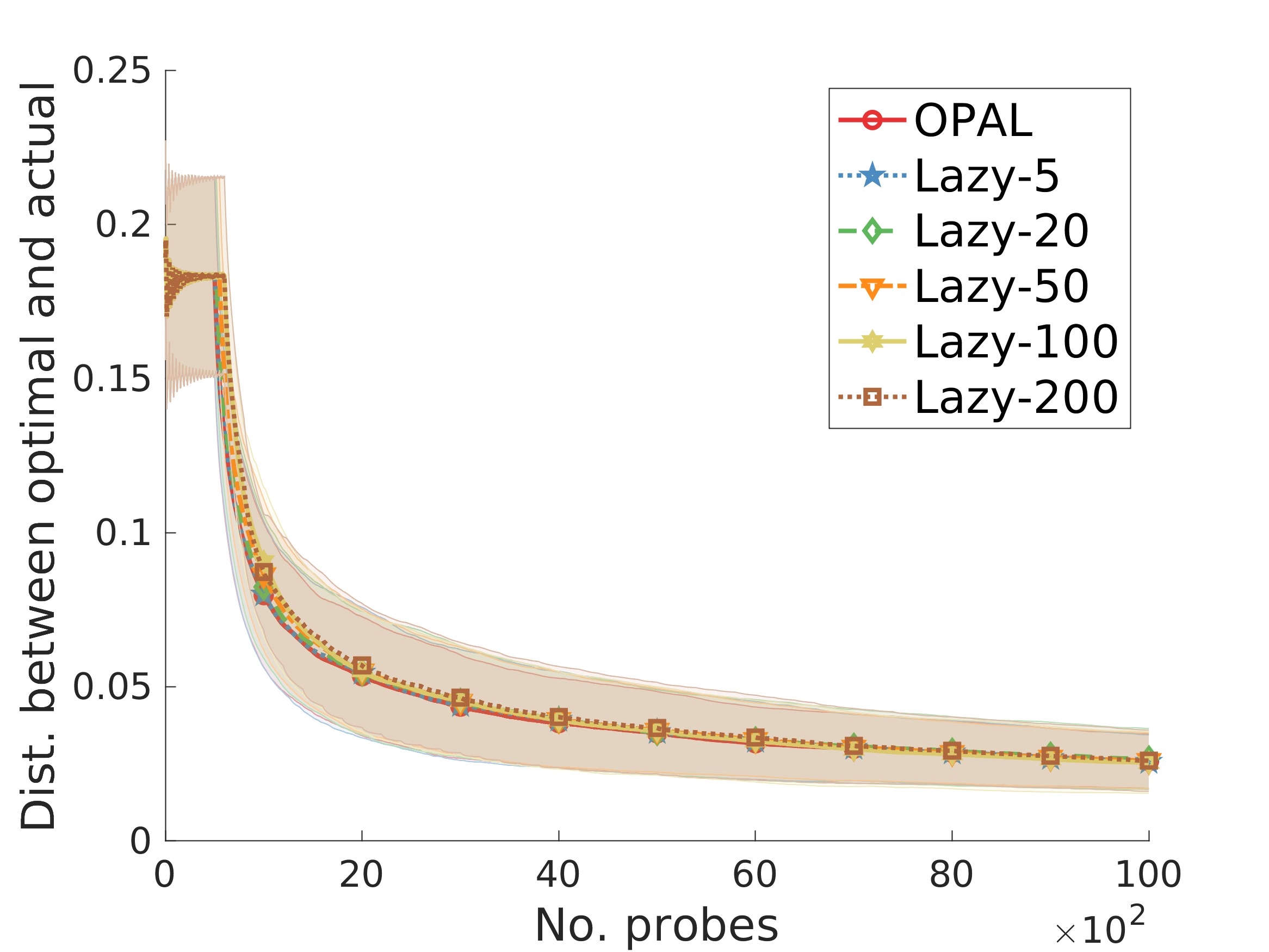}
        \caption{Dist. of actual: \(\norm*{\bm{\phi}^* - {\bm\phi}_t}_2\)}
    \end{subfigure}%
    \begin{subfigure}{.25\textwidth}
        \centering
        \includegraphics[width=\linewidth]{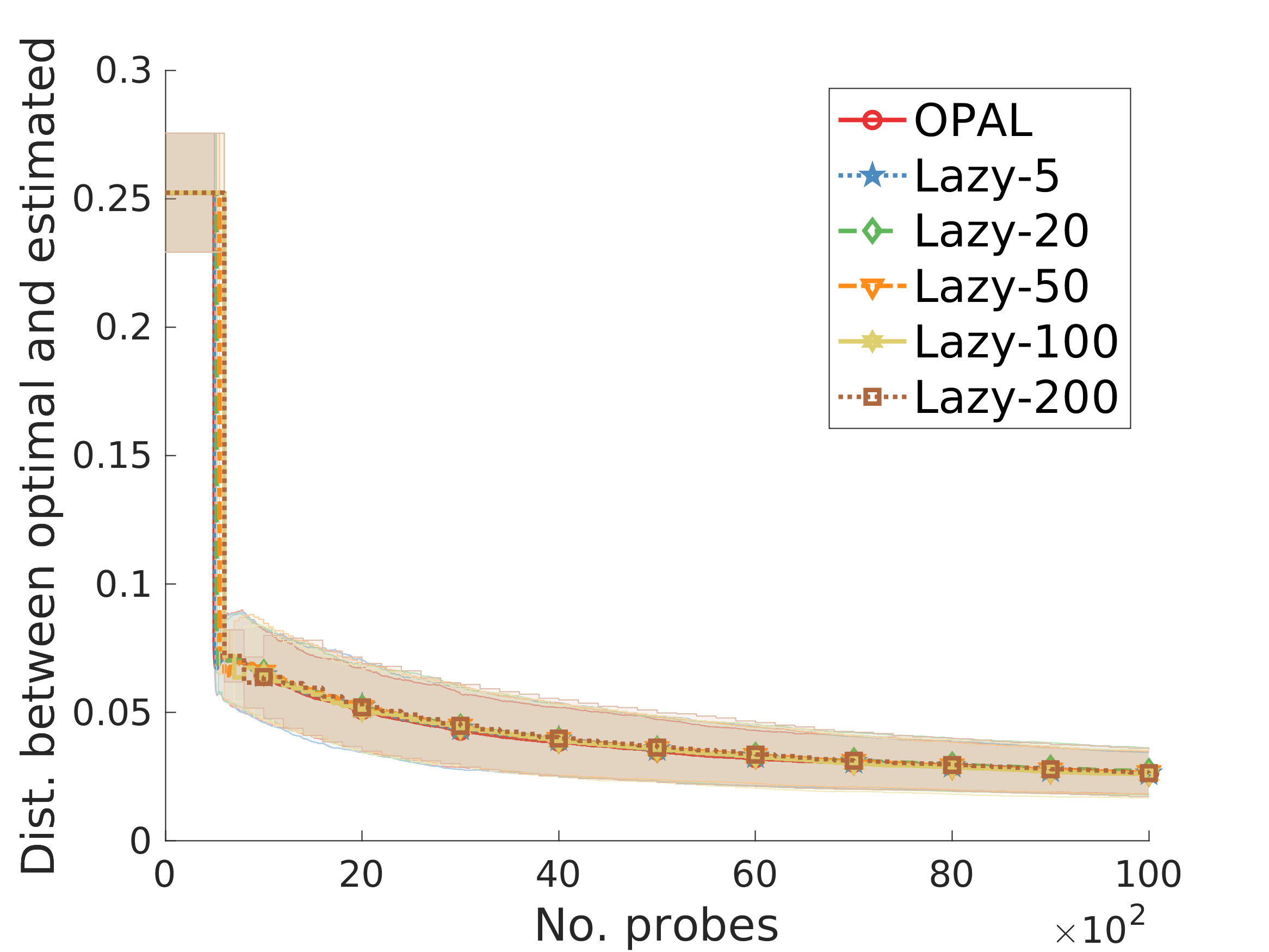}
        \caption{Dist. of estimated: \(\norm*{\bm{\phi}^* - \hat{\bm\phi}_t}_2\)}
    \end{subfigure}%
    \begin{subfigure}{.25\textwidth}
        \centering
        \includegraphics[width=\linewidth]{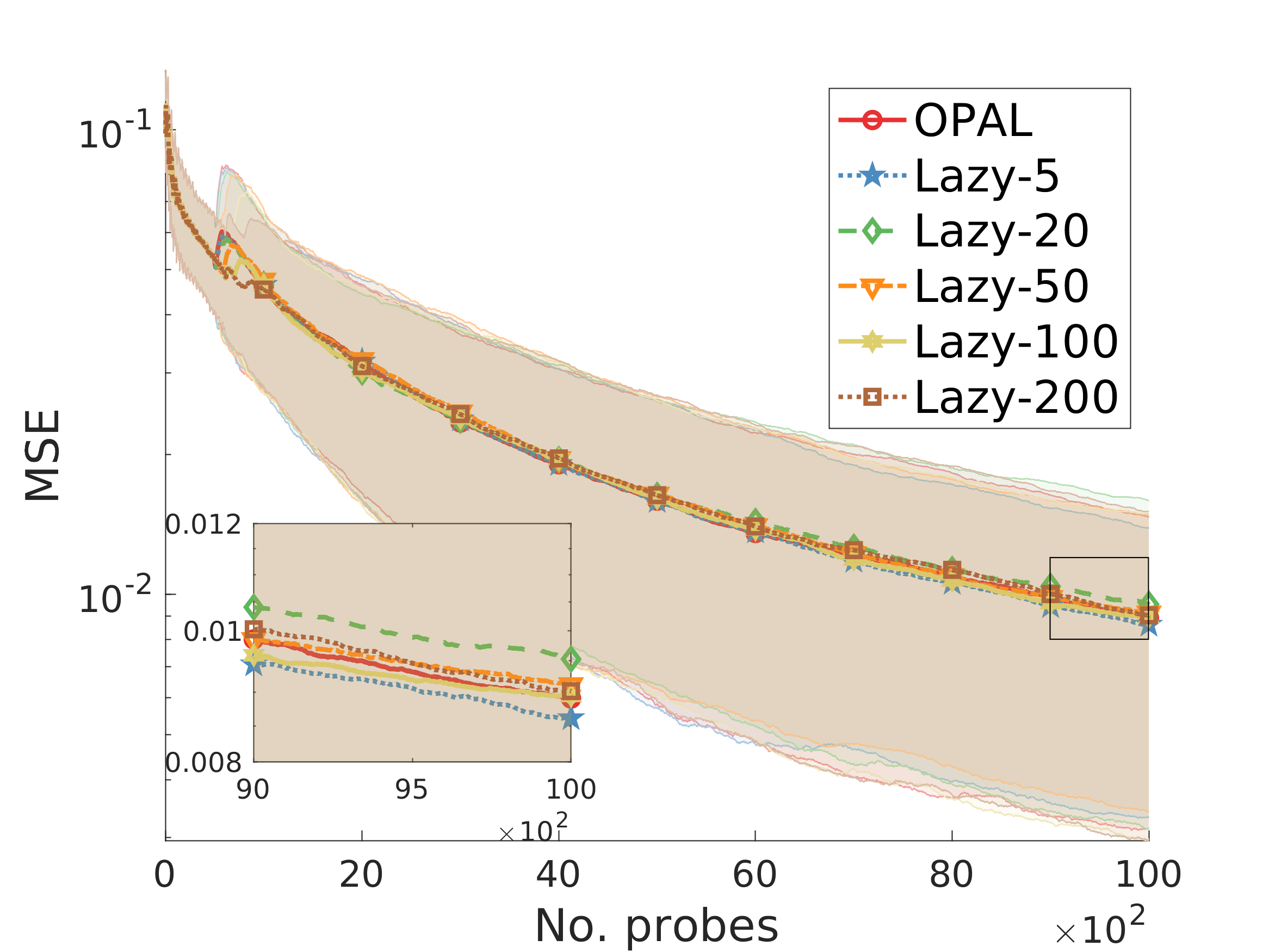}
        \caption{MSE: \(\mathbb{E}[\norm{\hat{\bm\mu}_t - \bm\mu}_2^2]\)}
    \end{subfigure}%
    \caption{Impact of lazy update batch size on \opal in loss tomography in classical ER network (zero initial sampling)}
    \label{fig:classical-lazy}
\end{figure*}

\begin{figure*}[tb]
    \centering
    \begin{subfigure}{.25\textwidth}
        \centering
        \includegraphics[width=\linewidth]{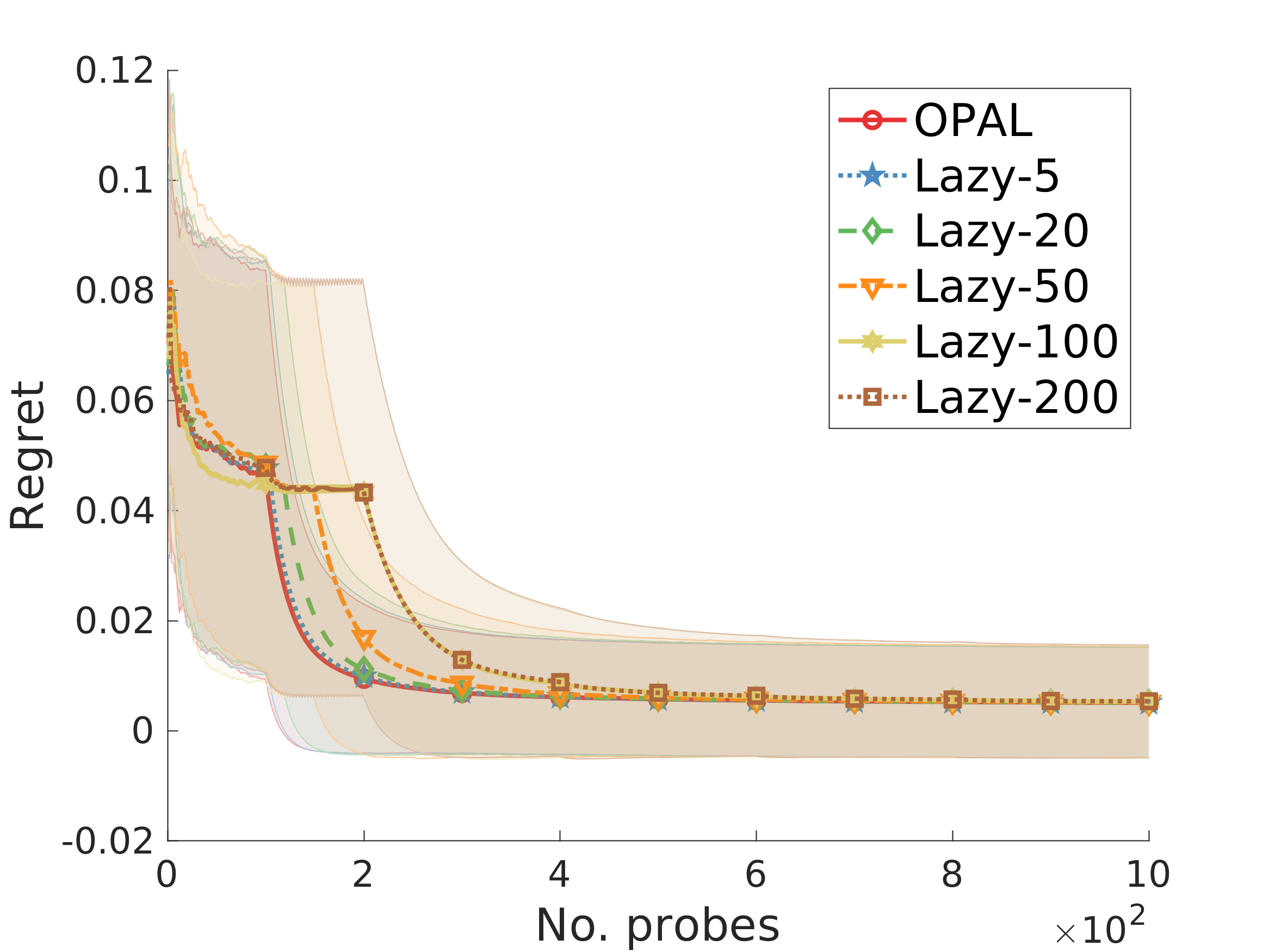}
        \caption{Regret: $R_t{=}F\!(\bm\mu,\!\bm\phi_t\!) {-} F\!(\bm\mu,\!\bm\phi^*\!)$}
    \end{subfigure}%
    \begin{subfigure}{.25\textwidth}
        \centering
        \includegraphics[width=\linewidth]{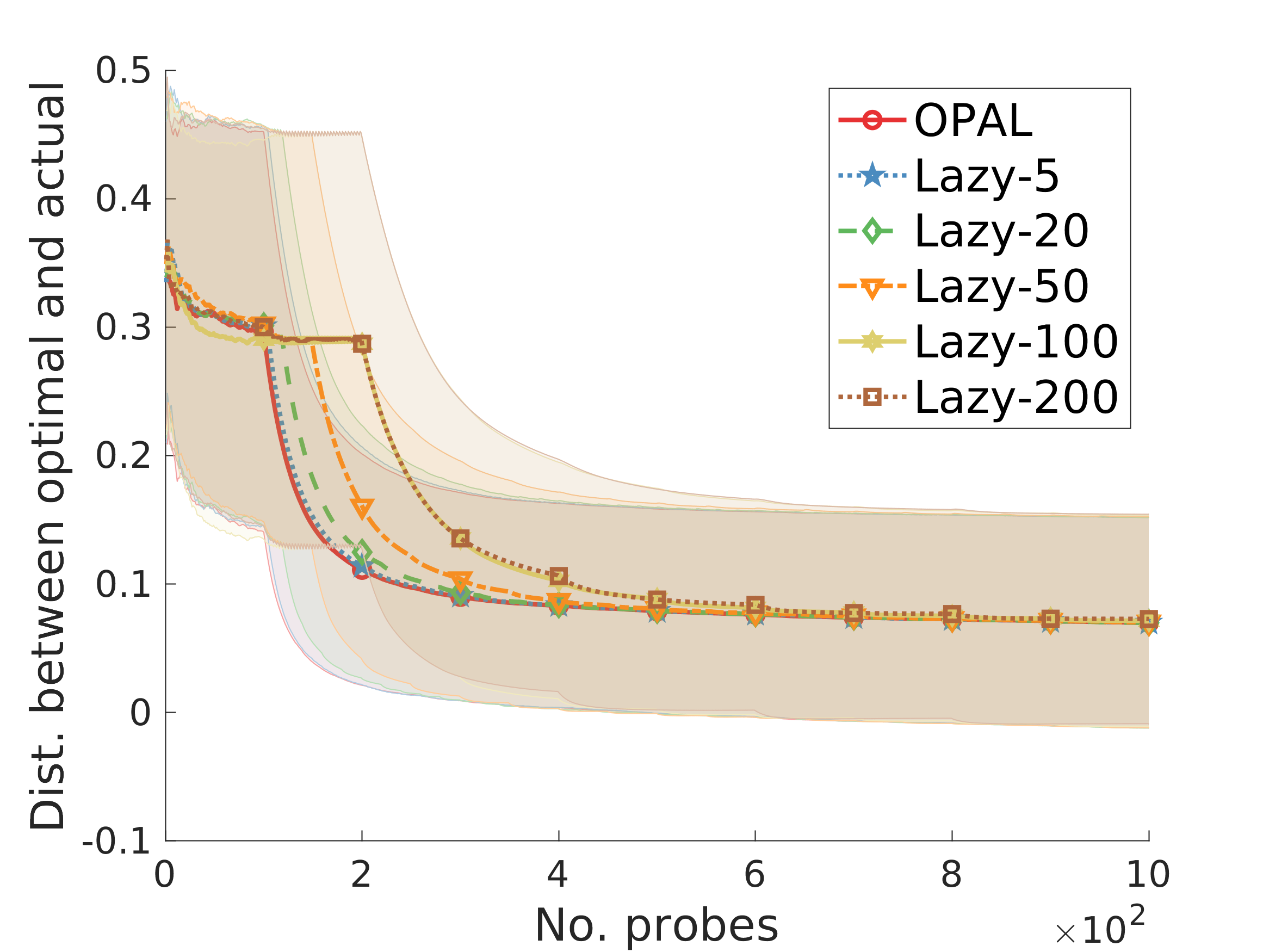}
        \caption{Dist. of actual: \(\norm*{\bm{\phi}^* - {\bm\phi}_t}_2\)}
    \end{subfigure}%
    \begin{subfigure}{.25\textwidth}
        \centering
        \includegraphics[width=\linewidth]{figures/lazycomparison/Lazy_ER_RSeed_Trial_Num_5_Node_10_mc_num=100_num_probes=10000_warm_num=500_iter_len=1_Phi_Est_Dist.png}
        \caption{Dist. of estimated: \(\norm*{\bm{\phi}^* - \hat{\bm\phi}_t}_2\)}
    \end{subfigure}%
    \begin{subfigure}{.25\textwidth}
        \centering
        \includegraphics[width=\linewidth]{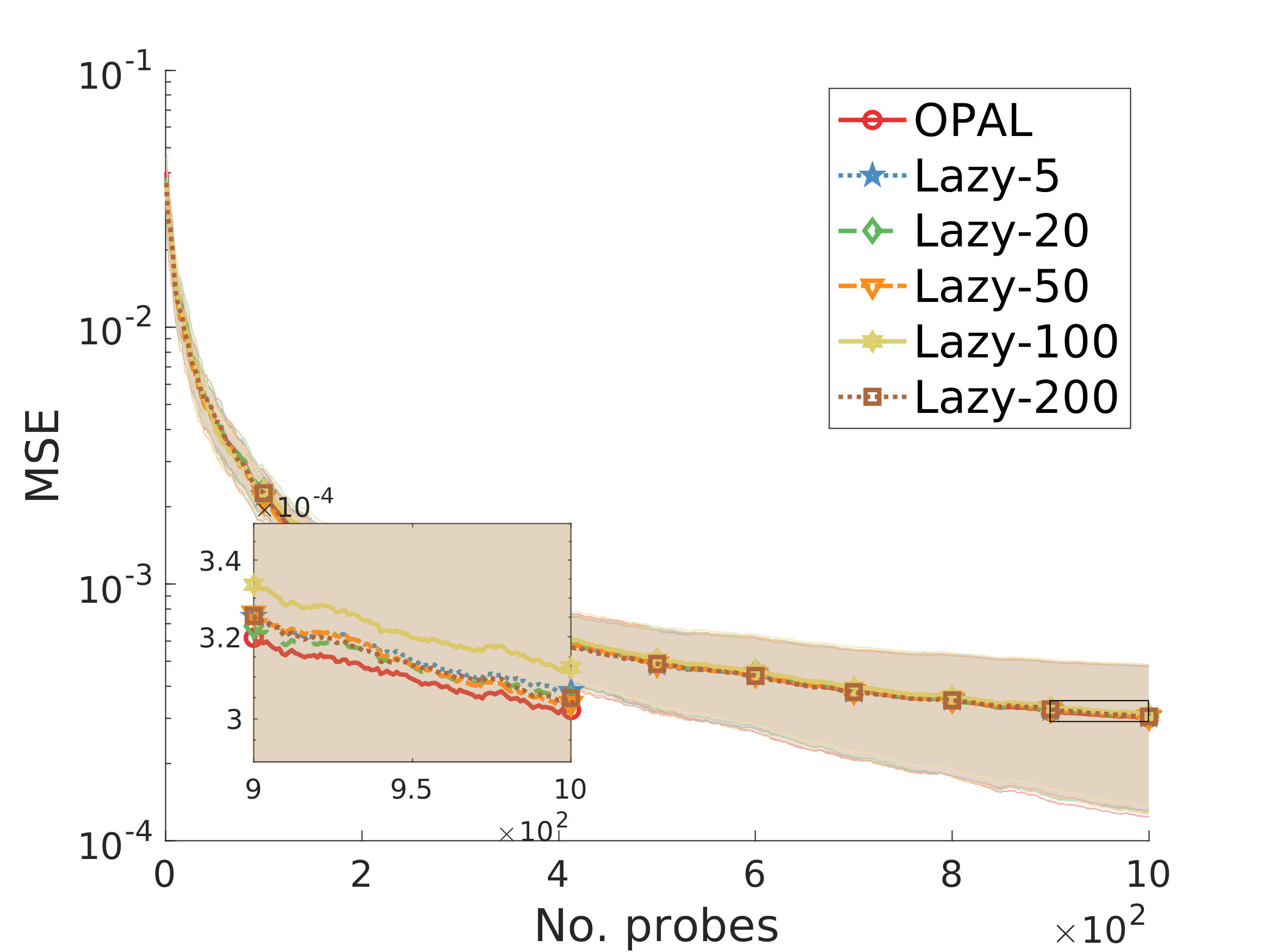}
        \caption{MSE: \(\mathbb{E}[\norm{\hat{\bm\mu}_t - \bm\mu}_2^2]\)}
    \end{subfigure}%
    \caption{Impact of lazy update batch size on \opal in bit-flip tomography in quantum bit-flip star network (\(5\%\) initial sampling)}
    \label{fig:quantum-lazy}
\end{figure*}

\section{Detailed Proof of Theorem~\ref{thm:MLE-estimator-fix-radius}}\label{app:proof-case-study-classical}

The next two lemmas
show the impact of the entries of the inverse matrix \(\bm Q^{-1}\) on the link parameter estimates.

\begin{lemma}\label{lma:MLE-estimator-fix-radius}
    For any estimator \(\hat\mu_\ell = \prod_{m\in\mathcal{M}} \hat\nu_m^{\kappa_{\ell, m}}\), where the \(\kappa_{\ell, m} \in \mathbb R\) is the \(\ell^{\text{th}}\)-row, \(m^{\text{th}}\)-column entry of matrix \(\bm Q^{-1}\) and \(\kappa_{\ell, m}' \coloneqq \min\{\frac 1 2, \abs*{\kappa_{\ell, m}}\}\),  we have, with a probability of at least \(1-\delta\) (for any $\delta\in (0,M^{-1})$),
    \begin{equation}
        \begin{split}
            \mu_\ell \in
             & \left(
            \hat\mu_\ell - \sum_{m\in\mathcal{M}:  \kappa_{\ell, m}\neq 0}
            c_{m,\ell}\left(  \frac{\log (\delta^{-1}/M)}{S_m} \right)^{\kappa_{\ell, m}'},
            \right.
            \\
             & \qquad \left.
            \hat\mu_\ell + \sum_{m\in\mathcal{M}:  \kappa_{\ell, m}\neq 0}
            c_{m,\ell} \left(  \frac{\log (\delta^{-1}/M)}{S_m} \right)^{\kappa_{\ell, m}'}
            \right),
        \end{split}
    \end{equation}
    where \(c_{m,\ell}\)'s
    are positive constants depending only on the network topology and link parameters.
\end{lemma}

As Lemma~\ref{lma:MLE-estimator-fix-radius} shows, the actual concentration rate (exponents \(\kappa'_{\ell,m}\)) of the MLE depends on the entries of the inverse measurement matrix \(\bm Q^{-1}\).
For unicast loss tomography in a star network, we prove a lemma that characterizes the inverse matrix \(\bm Q^{-1}\)  as follows,

\begin{lemma}\label{lma:inverse-matrix-fix-entry}
    For unicast loss tomography in a star network, inverse measurement matrix  \(\bm Q^{-1}\) only has \(\left\{0, \pm \frac{1}{2}, \pm 1\right\}\) entries.
\end{lemma}

Combining both lemmas yields the concentration rate of the MLE estimator for the link parameters \(\bm\mu\) in the star network (Theorem~\ref{thm:MLE-estimator-fix-radius}).

Below, we provide the proofs of Lemmas~\ref{lma:MLE-estimator-fix-radius} and~\ref{lma:inverse-matrix-fix-entry}.

\begin{proof}[Proof of Lemma~\ref{lma:MLE-estimator-fix-radius}]
    In this proof, we fix a link \(\ell\in\mathcal{L}\) and omit its appearance in subscript for simplicity.
    Denote \(\mathcal{M}^+ \coloneqq \{m\in \mathcal{M}: \kappa_{\ell, m} > 0\}\) and \(\mathcal{M}^- \coloneqq \{m\in \mathcal{M}: \kappa_{\ell, m} < 0\}\) as the sets of probes with positive and negative entries in the inverse matrix \(\bm Q^{-1}\)'s \(\ell^{\text{th}}\) column, respectively.
    Then, we have
    \begin{align}
        \hat\mu_\ell
         & = \prod_{m\in\mathcal{M}} \hat\nu_m^{\kappa_{m}}
        = \frac{\prod_{m\in\mathcal{M}^+} \hat\nu_m^{\kappa_{m}}}{\prod_{m\in\mathcal{M}^-} \hat\nu_m^{-\kappa_{m}}}
        \\
         & \overset{(a)}\le \frac{\prod_{m\in\mathcal{M}^+} (\nu_m + r_m)^{\kappa_{m}}}{\prod_{m\in\mathcal{M}^-} (\nu_m - r_m)^{-\kappa_{m}}}
        \\
         & = \frac{\prod_{m\in\mathcal{M}^+} \nu_m^{\kappa_{m}}}{\prod_{m\in\mathcal{M}^-} \nu_m^{-\kappa_{m}}} \cdot \frac{\prod_{m\in\mathcal{M}^+} (1 + r_m/\nu_m)^{\kappa_{m}}}{\prod_{m\in\mathcal{M}^-} (1 - r_m/\nu_m)^{-\kappa_{m}}}
    \end{align}
    \begin{align}
         & \overset{(b)} = \mu_\ell \cdot \frac{\prod_{m\in\mathcal{M}^+} (1 + r_m')^{\kappa_{m}}}{\prod_{m\in\mathcal{M}^-} (1 - r_m')^{-\kappa_{m}}}
        \\
         & = \mu_\ell + \mu_\ell \left( \frac{\prod_{m\in\mathcal{M}^+} (1 + r_m')^{\kappa_{m}}}{\prod_{m\in\mathcal{M}^-} (1 - r_m')^{-\kappa_{m}}} - 1 \right)
        \\
         & = \mu_\ell + \mu_\ell \times
        \\
         & \frac{
            \prod_{m{\in}\mathcal{M}^+} (1 + r_m')^{2\kappa_{m}} - \prod_{m{\in}\mathcal{M}^-} (1 - r_m')^{-2\kappa_{m}}}{
            \prod_{m{\in}\mathcal{M}^-} (1 {-} r_m')^{-\kappa_{m}}\!\!\left(
            \prod_{m{\in}\mathcal{M}^+} (1 {+} r_m')^{\kappa_{m}}
            {+} \prod_{m{\in}\mathcal{M}^-} (1 {-} r_m')^{-\kappa_{m}} \!\right)}
        \\
         & \overset{(c)}\le \mu_\ell + C \sum_{m\in\mathcal{M}^+ \cup\mathcal{M}^-} (r_m')^{2\kappa_m'}, \quad\text{for some constant }C,
    \end{align}
    where inequality (a) is because \(\nu_m \in (\hat{\nu}_m - r_m, \hat{\nu}_m + r_m)\) by the Hoeffding's inequality and \(r_m  \coloneqq \sqrt{\frac{\log (\delta^{-1}/M)}{S_m}}\),
    equation (b) is due to \(\mu_\ell
    = \prod_{m\in\mathcal{M}} \nu_m^{\kappa_{m}}\) when the MLE estimator has accurate inputs and denote \(r_m' \coloneqq r_m / \nu_m\),
    and inequality (c) is by noticing that (c1) for the nominator, \(2\kappa_m \in \{1,2\}, \forall m \in \mathcal{M}^+\) and \(-2\kappa_m \in \{1,2\}, \forall m \in \mathcal{M}^-\), and therefore, the nominator can be upper bounded by \(O(\sum_{m\in\mathcal{M}^+ \cup \mathcal{M}^-} r_m')\) while all other terms are with high orders can be omitted,
    and (c2) the denominator can be lower bounded by a constant when \(S_m\ge O(\log T)\) for all \(m\in\mathcal{M}\).
\end{proof}

\begin{proof}[Proof of Lemma~\ref{lma:inverse-matrix-fix-entry}]
    \textbf{Step 1. Block diagonalization.} By row and column switch, one can transform matrix \(\bm Q\) to a block diagonal matrix, as follows,
    \[
        \bm Q \longleftrightarrow \begin{bmatrix}
            \bm Q_1 & 0       & \ldots & 0       \\
            0       & \bm Q_2 & \ldots & 0       \\
            \vdots  & \vdots  & \ddots & \vdots  \\
            0       & 0       & \ldots & \bm Q_H
        \end{bmatrix},
    \]
    where each submatrix \(\bm Q_h\) is a square matrix that cannot be further block diagonalized.

    \textbf{Step 2. Determinant calculation.} By Laplace expansion and Leibniz's equation, we show \(\det(\bm Q_h) \in \{\pm 2\}\) for all \(h = 1,\dots,H\).

    For each row of the matrix \(\bm Q_h\), as in a unicast star network, it contains only two \(1\) entries while others are all zero entries. If there exists a column \(c\) in matrix \(\bm Q_h\) that only has one \(1\) entry in row \(r\) and all other entries are zero, then we apply the Laplace expansion on this single \(1\) entry, leading to \[
        \det (\bm Q_h) =  (-1)^{r+c} \det (\bm Q_{h,\text{sub},r,c}).
    \]
    For the submatrix \(\bm Q_{h,\text{sub},r,c}\), if there exists another column \(c'\) that only has one \(1\) entry, then we can apply the Laplace expansion on this single \(1\) entry again.
    By iteratively applying the Laplace expansion, we finally obtains a matrix \(\bm Q'_h\) with no single \(1\)-entry column.
    As each row has exactly two \(1\) entries, no single \(1\)-entry column also implies there is no column with more than two \(1\) entries.
    Therefore, \(\bm Q'_h\) is a square matrix with all columns and rows having exactly two \(1\) entries.
    Following the block diagonalization in Step 1, we also note that the matrix \(\bm Q'_h\) cannot be further decomposed into smaller block diagonal matrices, either.

    Then, we apply the Leibniz's equation to calculate the determinant of \(\bm Q'\) as follows,
    \[
        \det (\bm Q'_h) = \sum_{\sigma \in \text{perm.}} \text{sgn}(\sigma) \prod_{\ell=1}^L (\bm Q'_h)_{\ell, \sigma(\ell)},
    \]
    where \(\sigma\) is a permutation of the column of the matrix \(\bm Q'\), and \(\text{sgn}(\sigma)\) is the sign of the permutation.
    Because each row has exactly two \(1\) entries and the matrix cannot be further decomposed, the Leibniz's equation only has two non-zero terms with value either \(1\) or \(-1\) in the summation, and noticing the matrix is non-singular, the determinant is \(\pm 2\).

    \textbf{Step 3. Submatrix inversion.} Notice that for the block diagonal matrix, the inverse matrix is also block diagonalized, as follows,
    \[
        \begin{bmatrix}
            \bm Q_1 & 0       & \ldots & 0       \\
            0       & \bm Q_2 & \ldots & 0       \\
            \vdots  & \vdots  & \ddots & \vdots  \\
            0       & 0       & \ldots & \bm Q_H
        \end{bmatrix}^{-1}
        = \begin{bmatrix}
            \bm Q_1^{-1} & 0            & \ldots & 0            \\
            0            & \bm Q_2^{-1} & \ldots & 0            \\
            \vdots       & \vdots       & \ddots & \vdots       \\
            0            & 0            & \ldots & \bm Q_H^{-1}
        \end{bmatrix}.
    \]
    Hence, we only need to show that for each submatrix \(\bm Q_h\), the inverse matrix \(\bm Q_h^{-1}\) only has \(\left\{0, \pm \frac{1}{2}, \pm 1\right\}\) entries.
    As the inverse matrix \(\bm Q_h^{-1}\) can be calculated as follows,
    \[
        \bm Q_h^{-1} = \frac{1}{\det(\bm Q_h)} \bm C_h^T,
    \]
    where \(\bm C_h\) is the cofactor matrix of \(\bm Q_h\).
    The entries of \(\bm C_h\) are the determinants of the submatrices of \(\bm Q_h\), whose values are in \(\{0, \pm 1, \pm 2\}\) by applying the Leibniz's equation.
    As \(\det(\bm Q_h) = \pm 2\), the entries of \(\bm Q_h^{-1}\) are in \(\{0, \pm \frac{1}{2}, \pm 1\}\).
\end{proof}

\section{Additional Experiments: Impact of Lazy Update Batch Size}\label{sec:lazy-simulation}

In this appendix, we present additional experiments to investigate the impact of the lazy updates on \opal.
Figures~\ref{fig:classical-lazy} and~\ref{fig:quantum-lazy} illustrate the impact of the lazy update batch size on \opal, indicating that the performance of \opal remains robust regardless of the chosen lazy update batch size. From this observation, when the computation cost is a concern, one can choose a larger lazy update batch size to reduce the number of updates while maintaining a relatively good performance.

\end{document}